\documentclass[11pt,letterpaper]{article}
\usepackage{amsmath, amsthm, amssymb,dsfont, mathrsfs}
\usepackage[pdftex,left=1in,top=1in,bottom=1in,right=1in]{geometry}
\usepackage{filecontents} \immediate\write18{bibtex \jobname}
\usepackage[mathscr]{euscript}
\usepackage{graphicx}
\usepackage{color}
\usepackage[inline]{enumitem}
\usepackage{etoolbox}
\newcommand{\Mnote}[1]{}

\newtoggle{full} 
\newtoggle{opt} 
\togglefalse{opt}

\toggletrue{full}
\usepackage{titlesec}
\usepackage[initials,alphabetic]{amsrefs}

\newcommand{\TextStyle}{\iftoggle{full}{}{\textstyle}}
\newcommand{\InLine}[1]{\iftoggle{full}{\begin{equation*} #1 \end{equation*}}{$#1$}}

\title{Capacity Upper Bounds for Deletion-Type Channels\thanks{A preliminary version %
of this work appears in Proceedings of the
50th ACM Symposium on Theory of Computing (STOC~2018).
}}

\author{{\sc Mahdi Cheraghchi}\thanks{%
Email: $\langle$m.cheraghchi@imperial.ac.uk$\rangle$. }\\
Department of Computing \\
Imperial College London \\ London, UK 
}
\date{}

 \newcommand{\cD}{\mathcal{D}}
\newcommand{\cE}{\mathcal{E}} \newcommand{\cX}{\mathcal{X}}
\newcommand{\cY}{\mathcal{Y}} 
 
\newcommand{\cU}{\mathcal{U}} 
 
 \newcommand{\R}{\mathbb{R}}
\newcommand{\E}{\mathds{E}} \newcommand{\supp}{\mathsf{supp}}
 
\newcommand{\eps}{\epsilon} \newcommand{\innr}[1]{\langle #1 \rangle}
\newcommand{\cP}{\mathcal{P}} 
\newcommand{\N}{\mathds{N}}
\newcommand{\NN}{\mathds{N}^{\geq 0}}

\newtheorem{thm}{Theorem} \newtheorem{prop}[thm]{Proposition}
\newtheorem{claim}[thm]{Claim} \newtheorem{lem}[thm]{Lemma}
  \newtheorem{coro}[thm]{Corollary} \theoremstyle{definition}
 
\newtheorem{remark}[thm]{Remark} 
\newtheorem*{claimS}{Claim}

\newcommand{\bP}{\mathbf{P}}
\newcommand{\del}[2]{\frac{\partial{#1}}{\partial{#2}}}

\newcommand{\Bin}{{\mathsf{Bin}}}
\newcommand{\bZero}{\boldsymbol{0}}
\newcommand{\bOne}{\boldsymbol{1}}

\newcommand{\Iff}{{if and only if}}
\newcommand{\sX}{X^{\star}}
\newcommand{\sY}{Y^{\star}}
\newcommand{\KL}[2]{D_{\mathsf{KL}}({#1}\|{#2})}

\newcommand{\sC}{{\mathscr{C}}}

\newcommand{\capa}{{\mathsf{Cap}}}

\newcommand{\ch}{\mathsf{Ch}}
\newcommand{\ber}{{\mathsf{Ber}}}
\newcommand{\lam}{\lambda}
\newcommand{\Lam}{\Lambda}
\newcommand{\lsigma}{\underline{\sigma}}
\newcommand{\usigma}{\overline{\sigma}}
\newcommand{\li}{\mathrm{Li}}
\newcommand{\ibin}{\mathsf{InvBin}}
\newcommand{\nbin}{\mathsf{NegBin}}
\newcommand{\bDown}{\underline{\beta}}
\newcommand{\bUp}{\overline{\beta}}
\newcommand{\aDown}{\underline{\alpha}}
\newcommand{\aUp}{\overline{\alpha}}
\newcommand{\Ex}{\mathcal{E}}
\newcommand{\CBer}{C_{\mathsf{Ber}}}
\newcommand{\Exp}[1]{e^{#1}} 

\newcommand{\Dots}[3]{{#1},{#2},\ldots,{#3}}
\newcommand{\DotsII}[2]{{#1},{#2},\ldots}
\newcommand{\DotsAZ}[2]{{#1},\ldots,{#2}}
\newcommand{\DotsIII}[3]{{#1},{#2},{#3},\ldots}

\usepackage[hidelinks]{hyperref}

\iftoggle{full}{}{%
\titleformat{\part}
  {\normalfont\fontsize{14}{15}\bfseries}{Part \thepart}{1em}{}
\titleformat{\section}
  {\normalfont\fontsize{12}{15}\bfseries}{\thesection}{1em}{}
\titleformat{\subsection}
  {\normalfont\fontsize{11}{15}\bfseries}{\thesubsection}{1em}{}
\titleformat{\subsubsection}
  {\normalfont\fontsize{11}{15}\bfseries}{\thesubsubsection}{1em}{}
}
  
\begin{document}

\pagenumbering{roman} 

\maketitle

\begin{abstract}
We develop a systematic approach, based on convex programming and real analysis, 
for obtaining upper bounds on the capacity of the binary deletion
channel and, more generally, channels with i.i.d.\ insertions and deletions.
Other than the classical deletion channel, we give a special attention
to the Poisson-repeat channel introduced by Mitzenmacher and Drinea (IEEE Transactions
on Information Theory, 2006). 
Our framework can be applied to obtain capacity upper bounds for any repetition
distribution (the deletion and Poisson-repeat channels corresponding to the special cases
of Bernoulli and Poisson distributions).  Our techniques essentially reduce the task
of proving capacity upper bounds to maximizing a univariate, real-valued, and
often concave function over a bounded interval. The corresponding
univariate function is carefully designed according to the underlying distribution of repetitions
and the choices vary depending on the desired strength of the upper bounds as well as 
the desired simplicity of the function (e.g., being only efficiently computable versus 
having an explicit closed-form expression in terms of elementary, or common special, functions). 
Among our results, we show the following:
\begin{enumerate}
%
\item The capacity of the binary deletion channel with deletion probability $d$
is at most $(1-d) \log \varphi$ for $d \geq 1/2$, and, assuming the capacity
function is convex, is at most
$1-d \log(4/\varphi)$ for $d<1/2$, where $\varphi=(1+\sqrt{5})/2$
is the golden ratio. This
is the first nontrivial capacity upper bound for 
any value of $d$ outside the limiting case $d \to 0$
that is fully explicit and proved without 
computer assistance.

\item We derive the first set of capacity upper bounds for the Poisson-repeat
channel. Our results uncover further striking connections between this channel
and the deletion channel, and suggest, somewhat counter-intuitively, that the Poisson-repeat channel
is actually analytically simpler than the deletion channel and 
may be of key importance to a complete understanding of the deletion channel.

\item We derive several novel upper bounds on the capacity of the deletion channel.
All upper bounds are maximums of efficiently computable, and concave, univariate
real functions over a bounded domain. In turn, we upper bound these functions
in terms of explicit elementary and standard special functions, whose maximums
can be found even more efficiently (and sometimes, analytically, for example for $d=1/2$).
\end{enumerate}
Along the way, we develop several new techniques of potentially independent interest.
For example, %
\iftoggle{opt}{
we develop systematic techniques to study channels with integral inputs and outputs, 
which may be of interest to such problems as the Poisson channel
problem in information theory and optical communications.
}{we develop systematic techniques to study channels with mean constraints
over the reals. }
Furthermore, we motivate the study of novel probability distributions over non-negative
integers 
as well as novel special functions 
which could be of interest to mathematical analysis.\\


\end{abstract}

\newpage

\iftoggle{full}{\tableofcontents}{{\small\tableofcontents}}

\newpage

\pagenumbering{arabic} 

\iftoggle{full}{}{\part{Extended abstract}}

\section{Introduction}
\label{sec:intro}

The binary deletion channel is generally regarded as the simplest model
for communication in presence of synchronization errors.
In this model, a transmitter encodes messages as a
(potentially unbounded) stream of bits
which is then sent to a receiver over a communications channel.
The channel does not corrupt bits. However, each bit may be
independently discarded by the channel with a deletion probability
$d\in [0,1)$. The receiver receives the sequence of 
undiscarded bits, in their respective order, and has to reconstruct
the sent message with vanishing failure probability.
Despite the remarkable simplicity of this fundamental model of communication,
the capacity of the deletion channel; i.e., the maximum 
achievable transmission rate, remains unknown. 
Apart from the obvious significance in information, coding, and communications
theory, the problem has attracted significant attention from the theoretical
computer science community (e.g., \cite{ref:Mit09,ref:SZ99,ref:GNW12,ref:GW17,ref:BGZ16,ref:BGH17}). 
This is due to the problem's rich combinatorial structure
and its fundamental connection with the understanding of the distribution of
long subsequences in bit-strings, which, in turn, is of significance to such theory
problems as pattern matching, edit distance, longest common subsequence,
communication complexity problems involving the edit distance,  
the document exchange problem (cf.\ \cite{ref:BZ16}) 
or secure sketching in cryptography \cite{ref:DORS08}, to name a few.
It is also closely related to the algorithmic {trace reconstruction problem}
which, in turn, is of significance to real world applications 
ranging from sensor networks to computational biology \cite{ref:Mit09}.

\subsection{Previous work}

There is already a relatively vast literature on the deletion channel problem,
and we are only able to touch upon some of the major results most
relevant to this work.
Qualitatively, it is known that  
\begin{enumerate*}[label=\roman*)]
\item the capacity curve for this channel
is continuous, 
\item the capacity is positive for all $d \in [0,1)$ \cite{ref:DG01,ref:DG06},
\item the capacity is $1-\Theta(d \log d)$ (as in the binary
symmetric channel) when $d \to 0$ \cite{ref:KMS10,ref:KM13} and $\Theta(1-d)$
when $d \to 1$ (as in the binary erasure channel) \cite{ref:DM07,ref:MD06,ref:Dal11}.
\end{enumerate*}
Trivially, the capacity is at most the capacity of the binary
erasure channel; i.e., $1-d$. Nevertheless, the exact capacity
of the channel remains elusive. A related problem is to identify
the best achievable rate against adversarial, or oblivious, deletions;
for which significant progress has been recently made
\cite{ref:BGH17,ref:GR18}.
However, in this work we focus on the Shannon-type
capacity over random deletions (and, more generally, repetitions).
Much of the major known results on the subject as
well as the significance to theoretical computer science are discussed in 
Mitzenmacher's excellent survey \cite{ref:Mit09}. 

On the achievability side, Diggavi and Grossglauser \cite{ref:DG01,ref:DG06} were the first
to show that the capacity of the deletion channel is nonzero for all $d \in [0,1)$.
A more explicit capacity lower bound of (slightly better than) $(1-d)/9$, for all $d$, was proved by 
Drinea and Mitzenmacher \cite{ref:DM07,ref:MD06}, where in the latter work they also introduced
and motivated the \emph{Poisson-repeat channel}. This channel not only deletes
bits but may also insert replicated bits in the stream. More precisely, the channel 
is defined by a parameter $\lam$ and, given a bit, replicates the bit by
a number sampled from a Poisson distribution with mean $\lam>0$ (the bit
is deleted if the number of repetitions is zero). In \cite{ref:MD06}, the
authors establish a connection between the Poisson-repeat and the deletion channel.
Namely, they show that any \emph{lower bound} on the capacity of any Poisson-repeat channel
translates into a capacity lower bound for any deletion channel. Using
a first-order Markov chain for generating the input distribution,
and numerical computations
on the resulting capacity bound expressions for various choices of $\lam$,
they 
derive the claimed capacity lower bound of $(1-d)/9$ for the deletion channel.

For small deletion probability $d \to 0$, several results show that the deletion
capacity behaves similar to the symmetric channel. 
Combined with \cite{ref:Gal61,ref:Zig69,ref:DG06}, Kalai, Mitzenmacher, and Sudan~\cite{ref:KMS10}
show that in this regime that capacity is $1-h(d)(1-o(1))$, where $h(\cdot)$ is
the binary entropy function. Independently of this work, and 
based on a parameter continuity argument, 
Kanoria and Montanari \cite{ref:KM13} obtain a more refined asymptotic estimate 
in this regime that is correct up to the $O(d)$ term\footnote{We
remark that the constant behind the residual $O(d)$ term is not specified or bounded in 
\cite{ref:KM13}. Therefore, while this result sharply characterizes the limiting
behavior of the capacity curve, it cannot be used to obtain concrete,
numerical, bounds on the channel capacity for any nonzero value of $d$.
}.

Diggavi, Mitzenmacher and Pfister \cite{ref:MDP07} obtained capacity upper bounds
for all $d$, including the first nontrivial upper bound, 
of $0.7918(1-d)$, for $d \to 1$. To show the upper bounds, 
they consider a genie-aided decoder with access to side information about the
deletion process, and then upper bound the capacity of the channel
with side information (which is higher than the original capacity)
using a combination of classical information theoretic tools
and a computer-based distribution-optimization component.
Different sets of numerical capacity upper bounds were obtained in
\cite{ref:FD10} (and for more general channels in \cite{ref:FDE11}) 
based on several, carefully designed, genie-aided
decoders. These constructions essentially reduce the problem to upper bounding
the capacity of a finite variation of the deletion channel problem,
whose capacity is in turn numerically computed using the Arimoto-Blahut 
algorithm (which runs in exponential time in the finite number of bits).
Both \cite{ref:MDP07} and \cite{ref:FD10} thus cleverly identify a finite-domain
capacity problem, that is solved numerically, and then upper bound
the deletion capacity using the numerical results for the finite
problem. Such techniques cannot be readily extended to such problems as the
Poisson-repeat channel problem which are inherently infinite.

\subsection{Our main contributions}

Roughly speaking, the techniques of \cite{ref:MDP07} and \cite{ref:FD10}
pursue the following recipe: 
\begin{enumerate*}[label=\roman*)]
\item ``enhance'' the channel to one with a higher capacity
by carefully considering a ``genie-aided'' decoder that
receives auxiliary information from the channel; 
\item Heuristically extract a finite optimization problem
to upper bound the capacity of the enhanced (and thus the original)
channel;
\item Numerically solve the finite optimization problem
by a computer.
\end{enumerate*}
While the above general method results in very
strong capacity upper bounds, much of the mathematical
structure of the problem is pushed into the 
computationally intensive numerical optimization 
problem in the third step. It is thus unclear to what extent
can the methods be further developed 
towards a complete understanding of the channel capacity.
In this work, rather than setting our goal to improving the best known numerical capacity
upper bounds for the deletion channel,
we focus
on gaining deeper insights about the \emph{analytic structure} of the
problem (nevertheless, as a proof of concept, 
we are able to improve the best reported numerical
upper bounds for small deletion probability; e.g., for $d \leq 0.02$). 
We develop several tools that further the existing
intuitions on the deletion channel problem and may potentially serve as key steps
towards a full characterization of the capacity.
As a result, we are able, for the first time,
to develop a {single and systematic method} that results in
a capacity upper bound curve for the deletion channel which is smooth, convex-shaped,
non-trivial for all $d$, 
and simultaneously exhibits the correct behavior of $c(1-d)$, for a constant $c<1$,
at $d \to 1$ and $1-\Theta(h(d))$ at $d \to 0$ (see Figure~\ref{fig:DelPlot}).
The fact that our approach obtains the above features in a natural and
organic way suggests that the true capacity of the deletion channel might have
the same qualitative shape as what we obtain\footnote{In particular, we believe
that this observation further supports a conjecture of Dalai \cite{ref:Dal11} that the capacity 
curve is convex, towards which significant progress has already been made in \cite{ref:RD15}.
}.

As discussed above,
the best known reported capacity upper bounds for the deletion channel
\cite{ref:FD10,ref:MDP07,ref:RD15},
are based on identifying a finite, but as large as possible,
sub-problem (possibly by adding side information) and then 
searching for the optimum solution for the finite problem by a computer.
In contrast,
a key focus in our work is to avoid any computer-assisted components in the
proofs as much as possible, so as to gain as much intuition
about the mathematical structure of the problem as possible.
Our results, including all the involved distributions in the proofs,
are indeed fully analytical. Namely,
we upper bound the capacity of the deletion channel as
the maximum of a univariate real function, which is concave
and smooth, over the interval $(0,1)$ 
(depicted in Figure~\ref{fig:DelSlopes}). The function
to be maximized is explicitly defined in terms of exponentially decaying
sums, and is thus computable in polynomial time in the desired
accuracy. If desired, computation of the involved sums can
be avoided by using the sharp upper and lower bound estimates
on the function that we provide in terms of both elementary
and standard special functions (Figure~\ref{fig:DelSlopesUB}). 
The only numerical computation would thus involve finding 
the maximum value of an explicitly defined concave function over $(0,1)$.
Even this can be avoided in some cases, leading us to the
first fully explicit capacity upper bound
for the deletion channel that is nontrivial for all
deletion probabilities $d \in (0,1)$ and proved
without any numerical computation: \emph{The deletion
capacity is at most $(1-d) \log \varphi$ for $d \geq 1/2$, and,
under the plausible conjecture that the capacity function 
is convex \cite{ref:Dal11}, at most
$1-d \log(4/\varphi)$ for $d<1/2$, where $\varphi=(1+\sqrt{5})/2$
is the golden ratio.}
We remark that this, itself, is better than the bounds reported
in \cite{ref:MDP07} for all $d \geq 0.70$, while 
our numerical bounds improve those of \cite{ref:MDP07}
for all $d \geq 0.35$.

In addition to the classical deletion channel,
our methods are generally applicable to \emph{any} channel 
with independent insertions and deletions defined
by any given repetition rule. Namely, given an 
arbitrary (possibly infinite) distribution $\cD$ on non-negative integers,
our methods can be applied to upper bound the
capacity of a channel---what we call a \emph{$\cD$-repeat channel},
that replaces each input bit
independently with a number of repetitions sampled
from $\cD$ (where the outcome zero would cause
deletion of the bit). For the deletion channel, 
$\cD$ would be a known Bernoulli distribution.

For such problems as the Poisson-repeat channel problem,
introduced by Mitzenmacher and Drinea \cite{ref:MD06},
that are inherently infinite (even if only one bit is supplied
at the input), the known methods \cite{ref:FD10,ref:MDP07}
cannot be readily used, since it is not clear how to identify a
finite sub-problem that can be optimized
by a computer-based search.
In contrast, we show that our method easily applies 
to the Poisson-repeat channel
(where $\cD$ is Poisson with a known mean $\lam$),
and thus we obtain the first set of capacity upper bounds
for this channel. 
Our methods demonstrate striking connections between
the analytical structure of this channel and the deletion channel,
and suggest that understanding the Poisson-repeat channel
may be the key towards the ultimate characterization of
the capacity of the deletion channel. Even though the
Poisson-repeat channel may appear more complex than
the deletion channel (since it not only deletes, but
also inserts bits), our results suggest that the Poisson-repeat
channel may be simpler to analyze. This is mainly due to the
fact that an $x$-fold convolution of $\cD$ with itself
is a binomial distribution for the deletion channel, 
which is a more complex distribution than Poisson,
and indeed, contains the latter as a limiting special case.
In fact, we study the Poisson-repeat channel first, which
then naturally guides us towards our results for the deletion channel.
Our obtained bounds for both channels are
plotted in Figures \ref{fig:PoiPlot}~and~\ref{fig:DelPlot}
and tabulated in Tables \ref{tab:PoiCapData}~and~\ref{tab:DelCapData}.

To obtain our results, we develop a number of techniques along the way that 
may be of independent interest. We motivate a systematic
study of what we call \emph{general mean-limited} channels,
and their special case of \emph{convolution channels}.
These are channels, with input and output alphabets over the
reals, defined by a known probability transition rule
and a mean-constraint on their output distributions.
Special cases include the mean-limited binomial and Poisson
channels, that model how the deletion and Poisson-repeat
channels shrink consecutive runs of bits. 
The notion and our techniques can
be used to model physical channels studied outside the
context of deletion-type channels as well, a notable example
being the well-known Poisson channel that is of central importance
to optical communications systems \cite{ref:Sha90}.
Indeed, a subsequent work by the author \cite{ref:CR18}
successfully applies the techniques developed in this work to obtain improved
upper bounds on the capacity of the discrete-time Poisson channel.
Furthermore, our contributions in probability theory include motivating
novel distributions over non-negative
integers and a first study of them,
which may be of use in other contexts as well. 
This includes what we define as an ``inverse binomial''
distribution, as well as distributions obtained by
multiplying the probability mass function of the 
Poisson distribution $p(y)$ by $y^y$ or $\exp(y H_{y-1})$,
where $H$ denotes harmonic numbers (see \eqref{eqn:PoiYfirst}
and \eqref{eqn:PoiYre}). We introduce novel
special functions to study such distributions (e.g., 
generalizations of the log-gamma function; see \eqref{eqn:Lam})
which may be of independent interest to mathematical analysis.

\newcommand{\SecOrganization}{%
The rest of the article is organized as follows: 
\iftoggle{full}{%
Section~\ref{sec:highlevel} gives a high-level
exposition of the entire work, explaining the developed
techniques with a focus on intuitions and underlying 
insights rather than the technical details. %
}{}%
In Section~\ref{sec:meanlim}, we formally define the notion
of general mean-limited channels, as well as the special
case of convolution channels. We prove a duality-based
necessary and sufficient condition for achieving
the capacity of such channels, as well as the dual feasibility
criteria that certify upper bounds on the capacity.
Section~\ref{sec:general:repeat} formally defines
the general notion of $\cD$-repeat channels, and proves
a capacity upper bound for such channels based on 
the capacity of a mean-limited channel defined according
to $\cD$. We use the obtained tools in Section~\ref{sec:Poi}
to obtain our capacity upper bounds for the Poisson-repeat
channel. Towards this end, we construct two dual-feasible
solutions for the corresponding mean-limited channel
and estimate their parameters in terms of elementary
and standard special functions. 
In Section~\ref{sec:del}, guided by the result
obtained for the Poisson-repeat channel,
we prove our capacity upper bounds for the deletion
channel. We first introduce the notion of inverse binomial
distributions, and then show that it is dual feasible for
the mean-limited binomial channel. We estimate the parameters
of this distribution in terms of elementary and standard
spacial functions. We then apply a truncation technique,
that we first develop for Poisson-repeat channels,
to refine the dual feasible solution and obtain
improved capacity upper bounds.
}

\newcommand{\SecNotation}{%
Unless otherwise stated, all logarithms are taken to base $e$, and
the measure of information is converted from nats to bits only for the final
numerical estimates.
We denote the set of non-negative real numbers by $\R^{\geq 0}$
and the set of non-negative integers by $\NN$.
As is standard in information theory, we generally use capital letters
for random variables. When there is no risk of confusion, we may use
the same symbol for a random variable and its underlying distribution.
Support of a random variable $X$, denoted by $\supp(X)$, is the
set of the possible outcomes of $X$.
Calligraphic letters are used for several purposes; alphabets, distributions,
and probability transition rules. 
The entropy of a random variable $X$ is denoted by $H(X)$, and $h(p)$ denotes
the binary entropy function: \InLine{h(p):=-p \log p -(1-p) \log(1-p).}
The Kullback-Leibler (KL) divergence between the underlying distributions of random variables
$X$ and $Y$, denoted by $\KL{X}{Y}$, is defined as
\InLine{\KL{X}{Y} := \sum_{x \in \supp(X)} \Pr[X=x] \log(\Pr[X=x]/\Pr[Y=x]).} 
We use $I(X;Y)$ for the mutual information
between jointly distributed random variables $X$ and $Y$, and 
$I(X;Y|Z)$ for the conditional mutual information given a third variable $Z$.
We use the asymptotic notation $f(x) \sim g(x)$ to mean 
$\lim_{x \to \infty} f(x)/g(x) = 1$.
The binomial distribution with $x$ trials and success probability $p$
is denoted by $\Bin_{x,p}$, and the Bernoulli distribution with
mean $p$ is denoted by $\ber_p$. The capacity of a channel
$\ch$ is denoted by $\capa(\ch)$. We may use the binomial coefficient
$\binom{x}{y}$ over non-integers, in which case the definition
$\binom{x}{y} := \Gamma(1+x)/(\Gamma(1+y)\Gamma(1+x-y))$
should be used.
}

\iftoggle{full}{\subsection*{Organization} \SecOrganization}{}
\iftoggle{full}{\subsection*{Notation} \SecNotation}{}
 
\section{High-level exposition of the techniques and results}
\label{sec:highlevel}

In this work, we formalize and study the 
notion of ``general repeat channels'', 
which are binary input channels characterized 
by a given probability distribution $\cD$ over
non-negative integers. 
A \emph{$\cD$-repeat channel}, 
given a bit, draws an independent sample $D \geq 0$ from
$\cD$ and outputs $D$ copies of the received bit.
We define the deletion probability $d$ of the channel to
be the probability that $\cD$ assigns to zero, and
let $p = p(\cD) := 1-d$ be the \emph{retention probability}.
Thus, what we call the deletion channel 
corresponds to the case
where $\cD$ is the Bernoulli distribution
with mean $p$. On the other hand, if $\cD$ is
a Poisson distribution with mean $\lam$, we get
a Poisson-repeat channel with deletion probability
$d=e^{-\lam}$. We note that, in general, $\cD$
need not be uniquely determined by its deletion
parameter $d$, albeit this is the case for the class of 
deletion and Poisson-repeat channels.

\subsection{Reduction to mean-limited and convolution channels}
\label{sec:expo:reduction}

Suppose the input to a $\cD$-repeat channel
is a bit-sequence $X=\Dots{B_1}{B_2}{B_n}$ 
with $Y$ being the output bit-sequence. If $\mu:=\E[\cD]$,
the expected output length would be $\mu n$. By Shannon's
theorem, it can be seen (as in \cite{ref:Dob67}) that the capacity of the channel, that we denote
below by $\capa(\cD)$, is the supremum
of the normalized mutual information between $X$ and $Y$; i.e., 
$\capa(\cD) = \TextStyle \lim_{n \to \infty} {I(X;Y)}/n$.

A common technique in analyzing the deletion channel is
to consider how it acts on \emph{runs of bits}, rather than the
individual bits. Given a run of $x>0$ bits, the deletion 
channel outputs a run of $\Bin_{x,p}$ bits; i.e., a sample
from the binomial distribution with $x$ trials and success
probability $p$. For the Poisson-repeat channel, this would
be a run of length given by a Poisson sample with mean $\lam x$.
In general, the distribution of the output run-length would
be the $x$-fold convolution of $\cD$ with itself, that we 
denote by $\cD^{\oplus x}$.

Since $n$ grows to infinity, without loss of generality,
we can assume that the first input bit $B_1$ is zero.
This allows us to unambiguously think of $X$ as its
run-length encoding $X=\Dots{X_1}{X_2}{X_t}$, where each $X_i$
is a positive integer. Similarly, we may also think of $Y$
by its run-length encoding $Y=\Dots{Y_1}{Y_2}{Y_m}$, where
each $Y_i$ is a positive integer. This will identify
$Y$ up to a negation of all bits. Since the channel
has memory, the random variables
$Y_i$ (unconditionally)
are not necessarily independent. However, this would
be the case if the $X_i$ are independent and identically
distributed, in which case the $Y_i$ also become
identically distributed. Note that, given $X$,
the bit-length of $Y$ and the parameter $m$ are random variables determined
by the channel, and indeed the randomness of $m$ causes
technical difficulties that should be rigorously handled by a
careful analysis. However, in the informal exposition below we pretend that
$m$ is known and fixed a priori. 
One may now attempt to use the chain rule
and write 
\InLine{I(X;Y)
=I(X;Y_1)+I(X;Y_2|Y_1)+\cdots+I(X;Y_m|Y_1\ldots Y_{m-1}) 
\leq 
 \sum_{i=1}^m (H(Y_i)-H(Y_i|X,Y_1 \ldots Y_{m-1})).} 
A major difficulty in deriving the capacity of the
deletion channel is the fact that, unlike channels
with no synchronization errors, a certain $Y_i$ does
not only depend on the corresponding $X_i$, but rather, potentially any part of
$X$. Given $X$, we know that $Y_i$ has a binomial 
distribution with mean depending on the summation
$X_J + X_{J+2} + \cdots + X_{J+2K}$, for some
random variables $J$ and $K$ that are in general
difficult to analyze. Furthermore, even given a fixed
$X$, the random variables $Y_1, Y_2, \ldots$ 
do not become conditionally independent. 
Therefore, the above result of the chain rule 
cannot be upper bounded
by a simpler, single-letter, expression.
A natural idea, that has been pursued previously 
(e.g., \cite{ref:MDP07,ref:FD10}), 
is to consider ``genie-aided'' decoders that receive enhanced
side information by the channel. The side information is carefully
designed so as to reduce the problem to an i.i.d.\ channel
problem that can be analyzed more conveniently.
However, this approach generally comes at the expense
of effectively turning the channel into one with
strictly higher capacity, and consequently, obtaining
inherently sub-optimal capacity upper bounds.

\paragraph{The result of Diggavi, Mitzenmacher, and Pfister \cite{ref:MDP07}.} 
An elegant execution of the above idea, that in fact
inspires the starting point of our work, 
has been done in \cite{ref:MDP07}
which we briefly explain here. 
This result is based on the simple idea that, if
we imagine that the channel places a ``comma'' after each run
of bits, such that the commas are never deleted by 
the channel, this can only increase the capacity\footnote{%
For the limiting case $d \to 1$, the authors establish 
a different channel enhancement argument using 
carefully placed markers, that we do not discuss here.
}. Furthermore,
the enhanced channel is equivalent to an i.i.d.\ channel that
receives a stream of positive integers (i.e., the run-length
encoding of $X$) and passes each integer independently over a 
``run-length'' channel. The effect of the run-length channel,
given an input $x$, is to output a sample from $\Bin_{x,1-d}$;
i.e., the binomial distribution
with $x$ trials and success probability $p=1-d$.
Now, the capacity of the deletion channel can be upper bounded
by the capacity of the run-length channel, normalized by 
the number of input channel uses (i.e., the length of $X$).
Since the run-length channel is i.i.d., its capacity
is equivalent to the single-letter capacity. Thus, letting
$U$ and $V$ denote the input and output of a single use
of the run-length channel, 
capacity of the deletion channel is upper bounded by
$\sup_U I(U;V)/\E[U]$, where the supremum is over the distribution
of $U$ over positive integers\footnote{It is important to 
note that $U$ is defined over non-zero integers. Without this
consideration, the capacity upper bound would be infinite.
This can be easily seen by considering a distribution $U$ that
puts $1-\eps$ of the probability mass on the zero outcome, for
some $\eps>0$, and has mean $\Theta(\eps)$. 
A straightforward calculation shows that, in this case,
$I(U;V) = \Omega(\eps \log(1/\eps))$, and that the capacity
upper bound becomes $\Omega(\log(1/\eps))$, which can 
be made arbitrarily large by choosing a sufficiently small $\eps$.
}. At this point, a fundamental result of Abdel-Ghaffar \cite{ref:AG93}
on per-unit-cost capacity
can be used to, in turn, upper bound the resulting expression.

In the per-unit-cost capacity problem, a cost function $c(u)$ is
defined over the input domain $\cU$ of a channel with transition
rule $\cP(v|u)$ and the capacity is defined as $\sup I(U;V)/\E[c(U)]$,
where the supremum is over the distribution of $U$. For the above
application, the input domain is the positive integers and the
cost function is identity: $c(u)=u$. As proved in \cite{ref:AG93},
a necessary and sufficient condition for a pair $(U,V)$ to be
capacity achieving is the following: Letting
\InLine{C(V) := \sup_{u \in \cU}(\KL{V_u}{V}/c(u)),} where 
$V_u$ is the output distribution corresponding to a fixed input
$u$ and $\KL{\cdot}{\cdot}$ is the 
Kullback-Leibler (KL) divergence, the supremum is
attained for all $u$ on the support of $U$.
In this case, $C(V)$
is the per-unit-cost capacity.
Furthermore, for any distribution $V$ on the output domain 
(whether or not it corresponds to an input distribution),
the capacity is upper bounded by $C(V)$ as defined above.

There are two drawbacks with the above approach taken by \cite{ref:MDP07}.
First, the side information in fact genuinely increases the channel's
capacity, and therefore, the upper bounds resulting from this method
are inherently sub-optimal. One way to see this is to consider the case
where $p=1-d$ is small. Consider the run-length channel $(U,V)$ 
and a choice of $U$ that assigns a $1-p$
of the probability mass to $U=1$, and the rest to $U=2$.
In this case, one can see that $I(U;V)=\Omega(p \log(1/p))$,
and thus, the per-unit-cost capacity is also at least $\Omega(p \log(1/p))$.
This is while, by the trivial erasure channel upper bound,
the deletion capacity must be at most $p$.
Indeed, the numerical upper bounds reported by \cite{ref:MDP07}
exhibit this phenomenon 
at $d$ close to one (notice the distinctive concavity in this area
in Figure~\ref{fig:DelPlot}). 
The second drawback is that, while the result of
Abdel-Ghaffar is in principle powerful enough to characterize the
true per-unit-cost capacity upper bound, it may be
extremely difficult to work with this result analytically.
A way around this issue, undertaken in \cite{ref:MDP07}, 
is to employ a computer-based search. In order to do so, 
first a finite-domain distribution (supported up to an integer
$M$) for $U$ that maximizes the capacity
is constructed by a computer-based search,
and the corresponding KL divergence supremum, up to $M$,
is numerically computed.
Then, the resulting $V$ is truncated and the tail is
geometrically redistributed over the remaining (infinite)
input domain. The KL divergence at large values of $u$
can be accurately approximated by a linear function of $u$,
at which point \cite{ref:AG93} can be applied with the modified
choice of $V$, allowing the resulting capacity upper bound 
to be numerically computed.
While this approach indeed
achieves very strong numerical capacity upper bounds (especially
for small to moderate values of $d$), much of
the analytic structure of the problem is absorbed by
the computer-based search, 
making further progress elusive.
In this work, we develop a systematic, albeit technically demanding,
approach to overcome the barriers encountered by \cite{ref:MDP07}.

%

\paragraph{Equivalent reformulation of the channel.}

Consider any $\cD$-repeat channel $\ch$ with deletion probability $d=1-p$.
Rather than introducing side information, which, as demonstrated above, may result in a channel with
strictly larger capacity, we use a careful analysis to decompose the
action of the channel into two steps, forming a Markov chain $X-Z-Y$, such
that the resulting $Y$ from the two-step process has the \emph{exact} distribution
as the run-length encoding of the output of $\ch$. 
Then, we upper bound the capacity by $I(Z;Y)$ divided by the number
of channel uses in $X$. 
Assume, without loss of generality, that the first input bit in $X$ is
zero and this bit is promised not to be deleted by the channel (in particular,
the first bit ever output by the channel is also zero). 
Given the input $X$, the channel considers the run-length encoding 
$\Dots{X_1}{X_2}{X_t}$ of $X$. In order to produce the first run in the output,
the channel starts deleting bits from the even runs $\DotsII{X_2}{X_4}$ until
the first non-deletion event occurs, say at run $X_{2j}$. The odd runs
are then combined as $Z_1 = X_1 + X_3 + \cdots + X_{2j-1}$, and the process
continues from the first survived bit in the even runs until the input is fully
scanned. The resulting sequence $\Dots{Z_1}{Z_2}{Z_m}$ is then passed component-wise
through a channel $\ch'$, with integer input and output, defined
according to $\cD$, 
to produce the output sequence
$\Dots{Y_1}{Y_2}{Y_m}$ (for the case of the deletion channel, each $Y_i$ is formed simply by 
passing $Z_i-1$ through a binomial channel\footnote{A binomial channel,
given a non-negative integer input $x$, outputs a sample from the binomial
distribution with $x$ trials and success probability $p$.}, 
much similar to \cite{ref:MDP07}). 
We show that the resulting sequence has \emph{precisely} the same distribution as the 
run-length encoding of the output of $\ch$. By a delicate
analysis of this process, which is depicted in Figure~\ref{fig:processor}, we are able
to show that the capacity of $\ch$ is upper bounded as
\begin{equation}
\label{eqn:expo:capa:c}
\iftoggle{full}{%
\capa(\ch) \leq \sup_{U} \frac{I(U;V)}{1/p+\E[U]},
}{%
\capa(\ch) \leq \TextStyle \sup_{U} I(U;V)/(1/p+\E[U]),}
\end{equation}
where $U$ and $V$ are the input and output distributions of $\ch'$.
This constitutes the first technical building block in our capacity upper bound proofs.
The bound \eqref{eqn:expo:capa:c} has a similar flavor in form, 
but is strictly stronger, than the upper bound
expression in \cite{ref:MDP07} (especially for larger $d$). 

\paragraph{Mean-limited and convolution channels.}
Once \eqref{eqn:expo:capa:c} is available, one may attempt to
apply Abdel-Ghaffar's result \cite{ref:AG93} to obtain a
capacity upper bound, by using the cost function
$c(u)=u+1/p$. Indeed, any upper bound may in principle
be obtained by this result as it provides necessary 
and sufficient conditions for characterizing the quantity
on the right hand side of \eqref{eqn:expo:capa:c}.
However, as demonstrated by \cite{ref:MDP07}, an analytic
approach for obtaining a distribution $V$ that minimizes
the divergence fraction $\sup_{u \in \cU} \KL{V_u}{V}/(u+1/p)$,
or even gets sharply close to the minimum, may be extremely challenging.
Instead, \cite{ref:MDP07} uses a computer-assisted optimization
subroutine to construct a satisfactory $V$ and numerically upper bound
the capacity. 
While this exact numerical optimization subroutine applied on
\eqref{eqn:expo:capa:c} will strictly improve the numerical capacity upper
bounds reported in \cite{ref:MDP07} (since the cost function
$c(u)=u+1/p$ resulting from \eqref{eqn:expo:capa:c}
is strictly larger than the cost function $c(u)=u$ that is
used in \cite{ref:MDP07}), our aim is to obtain an analytic improvement
that avoids extensive numerical computations and provides 
deeper insights into the structure of the problem.
To overcome this difficulty, we observe that it 
would be much more
natural to break down the task of finding the best distribution for $V$
into two steps. First, we restrict the mean of $V$ to a fixed
parameter $\mu$ and optimize only over those $U$ such that $\E[V]=\mu$.
This fixes the denominator of the divergence fraction
to a constant and allows us to focus on optimizing the non-fractional
quantity $I(U;V)$ with respect to the fixed mean constraint.
Then, we take the supremum of the resulting bounds over the
choice of $\mu$ to upper bound \eqref{eqn:expo:capa:c}.
Note that the optimal $V$ for the two-step optimization
must satisfy the (necessary and sufficient) conditions
of \cite{ref:AG93} as well, so the two methods for characterizing
the capacity-achieving pair $(U,V)$ are technically
equivalent. However, factoring out the mean allows for a much
more natural and systematic derivation of the right distribution
for $V$, and is what allows us to achieve the desired analytic breakthroughs.

We thus obtain the abstraction of what we call a
``mean-limited channel''. Such a channel is defined with respect to certain
input and output domains over non-negative reals,
and a transition rule $\cP(y|x)$ for producing an 
output distribution over the output domain given an input distribution over the
input domain. Furthermore, the channel is given a mean parameter $\mu>0$
and only accepts those input distributions for which the corresponding
output distributions have mean $\mu$. The capacity of the channel is
determined in the standard sense of maximal mutual information
between admissible input and output distributions.

The abstraction is of general and independent interest to the study of communications
channels in presence of mean ``power constraints'', such as
the classical Poisson channel.
However, for our applications, it suffices to consider discrete
domains (in particular, non-negative integers) and 
the special case of transition rules that are defined by convolutions
of distributions, resulting in a special case that we call a 
``convolution'' channel.  
A convolution channel is defined with respect to a distribution
$\cD$ over non-negative reals, and is denoted by $\ch_\mu(\cD)$,
where $\mu>0$ is the output mean constraint.
Given an input $x$, the channel produces a sample from $\cD^{\oplus x}$
(i.e., the distribution defined by the $x$th power of the characteristic
function of $\cD$) in the output.
One may extend the notion of mean-limited channels to allow for
$m$ uses, and for a total mean constraint $\mu m$. Namely,
the channel now accepts an $m$-dimensional sequence
$U_1, \ldots, U_m$ at input and passes each $U_i$ through
an independent, identical, mean-limited channel to generate 
the output sequence $V_1, \ldots, V_m$. The mean-constraint
in this case would enforce the condition $\E[V_1 + \cdots + V_m]=\mu m$.
It is straightforward to show that the capacity of this channel is
achieved by a product distribution.

\subsection{Upper bounding the capacity of general mean-limited channels}

The appeal in reducing the capacity upper bound problems for 
$\cD$-repeat channels to that of general mean-limited channels is that,
for the latter, one may naturally use powerful tools from convex optimization to
obtain strong capacity upper bounds in a systematic and completely analytic fashion.
To this end, we prove an analogue of Abdel-Ghaffar's result \cite{ref:AG93}
for general mean-limited channels. However, we use a different, direct, proof. 
Namely, in Section~\ref{sec:capa:meanlim} we directly
write the mutual information maximization problem as a convex program,
form its dual and observe that strong duality holds. Hence, we may
write down the Karush-Kuhn-Tucker (KKT) conditions that provide the 
necessary and sufficient conditions for optimality. 

Characterization of channel capacity in terms of the optimum of a convex
program and the use of duality is a standard technique in information theory
(cf.\ \cite{ref:CK11}). 
Variations of this technique has been used,
for example, towards understanding the capacity of
multiple-antenna systems \cite{ref:LM03} and the
discrete time Poisson channel \cite{ref:Mar07,ref:LM09}. In this section,
we derive a variation tailored to our applications for
upper bounding the capacity of mean-limited channels\footnote{%
The variation developed in this section can be generally
applied to any mean-limited discrete or continuous channel.
In a subsequent work by the author \cite{ref:CR18}, the technique
has been successfully applied to obtain simple and improved upper bounds on
the capacity of the discrete-time Poisson channel. 
}.

Consider a general mean-limited channel $\ch$ with transition rule $\cP$,
input domain $\cX$, and 
mean-constraint $\mu$. 
With a slight overload of the notation, in this section consider an input
distribution $X$ for $\ch$ and let $Y$ denote the corresponding output distribution.
For any fixed input $x$, denote by $Y_x$ the output random variable when
the input is fixed to $x$. The KKT conditions imply that $X$ is capacity
achieving if and only if, for some real parameters $\nu_0$ and $\nu_1$,
we have
\begin{equation} \label{eqn:expo:dualKL}
(\forall x \in \cX)\colon\ \KL{Y_x}{Y} \leq \nu_1 \E[Y_x]+\nu_0,
\end{equation}
with equality for all $x \in \supp(X)$.
In this case,
the capacity is equal to $\nu_1 \mu+\nu_0$. 
Furthermore, if there is a distribution $Y$ over the output alphabet for which
\eqref{eqn:expo:dualKL} holds (and we call the distribution
``dual feasible''), then $\capa(\ch) \leq \nu_1 \mu+\nu_0$.
These results are summarized in Theorem~\ref{thm:meanlimited}.

Perhaps the most technically demanding aspect of this work is to obtain
fully analytical dual feasible solutions $Y$ that provably 
provide sharp, and explicit, or
at least efficiently computable, upper bound estimates on the
capacity of the mean-limited Poisson and binomial channels.
These, in turn, lead to capacity upper bounds
for the Poisson-repeat and deletion channels.
In both cases, we carefully construct dual feasible distributions parameterized
by a parameter $q \in (0,1)$ that controls the mean $\mu$ to
any arbitrary positive value.
Once the feasibility of these distributions are proved, we explicitly
write down the corresponding real parameters $\nu_0, \nu_1$, as well as
the resulting capacity upper bound $\nu_1 \mu + \nu_0$ (which requires writing
$\mu$ as a function of $q$), and plug in the resulting upper bound in 
\eqref{eqn:expo:capa:c}. This, in turn, results in an upper bound expression
for the capacity of the original $\cD$-repeat problem
(e.g., either the Poisson-repeat or deletion channel) as the maximum of
a uni-variate real function in $q$ (which turns out to be concave in $q$).
In all cases, this function is efficiently computable. 
In turn, we upper bound this function in terms of either explicit elementary functions,
or more sharply, in terms of the standard special functions. Thus, in particular,
we are able to reduce the problem of upper bounding the capacity of a
Poisson-repeat or deletion channel to finding the maximum of an elementary, 
concave, function of $q$. Numerical computation is then only applied,
if necessary, at
the very last stage for computing the maximizing value of $q$ for this
function, and the corresponding capacity upper bound. 

\paragraph{The Poisson-repeat channel.}
In order to obtain a capacity upper bound for the Poisson-repeat
channel, we use \eqref{eqn:expo:capa:c} combined with a capacity
upper bound for the corresponding mean-limited Poisson channel
using the KKT conditions described above. Suppose that the Poisson-repeat
channel replaces each bit with a number of bits sampled from a 
Poisson distribution with mean $\lam$. Therefore, the deletion 
probability of this channel (i.e., probability that a bit is replaced
by zero copies) is $d=\Exp{-\lam} =: 1-p$. The corresponding
mean-limited Poisson channel $\ch$ with mean constraint $\mu$ takes
a non-negative integer $X$ at input and outputs a fresh sample from
the Poisson distribution with mean $\lam X$. 
We use a convexity argument to show that the following distribution
over the non-negative integers,
parameterized by $q \in (0,1)$,
satisfies \eqref{eqn:expo:dualKL}
with $\nu_0=-\log y_0$ and $\nu_1=-\log q$: 
\begin{equation} \label{eqn:expo:PoiYfirst}
\Pr[Y=y] = y_0 {y ^y} (q/e)^y/y!,
\end{equation}
where $y_0$ is the normalizing constant, and $0^0:=1$. 
The values of $y_0$ and $\mu:=\E[Y]$ would in
turn be functions of 
$q$, and can numerically be computed
in polynomial time in the desired accuracy, given the exponential decay of \eqref{eqn:expo:PoiYfirst}
(see Table~\ref{tab:PoiMeanY}).
This results in a capacity
upper bound of $\sup_{q \in (0,1)}(-\mu \log q-\log y_0)$ for the mean-limited channel
(Theorem~\ref{thm:PoiChUpper}).
Furthermore, combined with \eqref{eqn:expo:capa:c}, we get the first
set of capacity upper bounds for the Poisson-repeat channel with deletion
probability $d$ (Theorem~\ref{thm:Poi:capa}):
\begin{equation} \label{eqn:expo:Poi:capa}
\iftoggle{full}{%
\capa \leq \TextStyle \sup_{q \in (0,1)} \frac{-\mu(q) \log q - \log y_0(q)}{-\mu(q)/\log d+1/(1-d)}.
}{%
\capa \leq \TextStyle \sup_{q \in (0,1)} (-\mu(q) \log q - \log y_0(q))/(-\mu(q)/\log d+1/(1-d)).}
\end{equation}
The function inside the supremum turns out to be concave, and the maximum
can efficiently be found by a simple search (Figure~\ref{fig:PoiSLope}). 
However, it is desirable to have sharp 
upper bound estimates on the function (in $q$) to be maximized. Note that,
from Stirling's approximation, the asymptotic behavior of \eqref{eqn:expo:PoiYfirst}
is $\Theta(q^y/\sqrt{y})$. Therefore, intuitively, it should be possible to 
estimate $y_0$ and $\mu$ in terms of 
the summations $\sum_{y=1}^\infty q^y/\sqrt{y}$ and 
$\sum_{y=1}^\infty \sqrt{y} q^y$, which may be 
expressed by the polylogarithm function \eqref{eqn:polylog:def}, a well studied
special function (however, obtaining upper and lower bounds requires more work).
This allows us to provide a remarkably sharp upper estimate on
the function in \eqref{eqn:expo:Poi:capa} in terms of the standard
special functions. This is made precise in Theorem~\ref{thm:PoiChUpperAnalytic}
and the quality of the approximation is depicted in Figure~\ref{fig:PoiMuLerch}.

We observe that the gap in \eqref{eqn:expo:dualKL} achieved by \eqref{eqn:expo:PoiYfirst}
is zero at $x=0$ but converges to an absolute constant (namely, $1/2$) as $x \to \infty$.
This results in sub-optimal capacity upper bounds. We rectify this issue
by replacing the $y^y=\exp(y \log y)$ term in \eqref{eqn:expo:PoiYfirst}
with $\exp(y \psi(y))$, where $\psi(y)$ is the digamma function (essentially
we are replacing the $\log y$ in the exponent with harmonic numbers, which have
the same asymptotic behavior); namely, we now use
what we call the \emph{digamma distribution}
\begin{equation} \label{eqn:expo:PoiYre}
\Pr[Y=y] = y_0 {\exp(y \psi(y))} (q/e)^y/y!,
\end{equation}
with $\Pr[Y=0]=y_0$. 
 Using the Newton series expansion of harmonic
numbers as well as the factorial moments of the Poisson distribution, 
we show that this alternative choice is also dual feasible, and in fact,
the gap in \eqref{eqn:expo:dualKL} offered by this choice is
precisely $\lam x E_1(\lam x)$, where $E_1(\cdot)$ is the
exponential integral function \eqref{eqn:Ei:def}. Thus the gap is zero at $x=0$ and 
exponentially vanishes as $x$ grows (Figure~\ref{fig:Ei}). This leads to a significant 
improvement in the resulting capacity upper bounds (Figure~\ref{fig:PoiPlot}).
We note that the digamma distribution 
\eqref{eqn:expo:PoiYre} still exhibits the same asymptotic behavior
as \eqref{eqn:expo:PoiYfirst}, and thus its parameters can be similarly approximated.
However, we show that the same asymptotic behavior is exhibited by the
well-studied negative binomial distribution \eqref{eqn:negbin} of order $r=1/2$. 
Since the parameters of a negative binomial distribution take remarkably simple forms,
we are able to obtain excellent upper and lower bound estimates on 
the parameters $y_0$ and $\mu$ of the digamma distribution \eqref{eqn:expo:PoiYre}, and in turn
the function inside the supremum in \eqref{eqn:expo:Poi:capa}, in terms of 
elementary functions. This is made precise in Section~\ref{sec:Poi:estimate:negBin}
(Corollary~\ref{coro:Poi:negBinApprox}). As a result, we obtain several upper bound
estimates on the capacity of the Poisson-repeat channel (either using 
\eqref{eqn:expo:PoiYfirst} or the digamma distribution 
\eqref{eqn:expo:PoiYre} or their upper bound
estimates), which are depicted in Figure~\ref{fig:PoiPlot}
and listed in Table~\ref{tab:PoiCapData}.

\paragraph{The deletion channel.}
At a first glance, it is natural to get the impression that understanding
the capacity of a Poisson-repeat channel may be a more complex problem
than that of a deletion channel. After all, a deletion channel only deletes bits
whereas a Poisson-repeat channel may cause insertions  (repetitions)
\emph{in addition to} deletions. However, our work indicates that, counter-intuitively,
the deletion channel is analytically more complex than the Poisson-repeat
channel. In fact, we use the above results for the Poisson-repeat channel
as a guiding tool towards attacking the deletion channel problem (which is
why the Poisson-repeat channel is discussed first). The mean-limited
channel corresponding to a deletion channel is a binomial channel, 
which maps an input $x$ to the binomial distribution $\Bin_{x,1-d}$ over $x$
trials. On the other hand, the Poisson-repeat channel corresponds to a
mean-limited Poisson channel, which maps $x$ to a Poisson random
variable with mean $\lam x = -x \log d$. A Poisson distribution,
being a one-parameter distribution, is analytically simpler than
a two-parameter binomial distribution. Indeed, the Poisson distribution 
is a limiting special case of the binomial distribution. We use
this intuition to extend our results for the Poisson-repeat channel to
the deletion channel.  As in the Poisson case, we invoke 
\eqref{eqn:expo:capa:c} to reduce the capacity upper bound problem
for the deletion channel to that of the mean-limited binomial channel
with output mean constraint $\mu$. Then, the task of finding a dual feasible
distribution $Y$ naturally leads us to a novel distribution that we call
an ``inverse binomial'' distribution (discussed in Section~\ref{sec:invbin}), 
which is defined, for 
$q \in (0,1)$, by
\begin{equation} \label{eqn:expo:invbin}
\Pr[Y=y] = \TextStyle y_0 \binom{y/p}{y} q^y \exp(-y h(p)/p),
\end{equation}
where $h(\cdot)$ is the binary entropy function. The parameter $q$
uniquely determines the normalizing constant $y_0$ and the mean
$\mu=\E[Y]$ (and the mean can be adjusted to any desired positive value).
We use a convexity argument to show  (in Theorem~\ref{thm:BinChUpper})
that the above distribution indeed satisfies \eqref{eqn:expo:dualKL} with $\nu_0 = -\log y_0$
and $\nu_1 = -\log q$, thus resulting in a capacity upper bound of
$\sup_{q \in (0,1)}(-\mu \log q - \log y_0)$ for the mean-limited binomial channel and
a capacity upper bound of
\begin{equation} \label{eqn:expo:del:capa}
\iftoggle{full}{%
\capa \leq p \sup_{q \in (0,1)} \frac{-\mu(q) \log q - \log y_0(q)}{1+\mu(q)}.
}{%
\capa \leq \TextStyle p \sup_{q \in (0,1)} (-\mu(q) \log q - \log y_0(q))/(1+\mu(q)).}
\end{equation}
for the deletion channel. As in the Poisson case, it is desirable to obtain sharp
upper bound estimates on the term inside the supremum in \eqref{eqn:expo:del:capa},
which turns out to be a concave function of $q$, in terms of elementary or
common special functions. This is a technical task, and in 
Section~\ref{sec:invbin:spacial} (in particular, Theorem~\ref{thm:ibin:Lerch}),
we obtain sandwiching bounds for the parameters of an inverse binomial
distribution in terms of the Lerch transcendent \eqref{eqn:Lerch}, a well studied generalization
of the Riemann zeta function. Furthermore, we observe that an inverse binomial
distribution exhibits the same asymptotic growth as a negative binomial distribution
of order $r=1/2$. Using this, we are able to obtain upper and lower bound estimates
on the parameters of an inverse binomial distribution in terms of elementary functions
(Corollary~\ref{coro:invbin:negBinApprox} in Section~\ref{sec:invbin:negbin}).
The estimates are excellent if the deletion probability is not too small 
(Figure~\ref{fig:BinEstimates}).

Interestingly, we show that for $p=d=1/2$, the inverse binomial distribution
is \emph{exactly} a negative binomial distribution. Thus, in this case, we can
write down the \emph{exact} parameters of the distribution in terms of elementary functions
and show that the term inside the supremum in \eqref{eqn:expo:del:capa}
is simply $h(q)/(2-q)$, which is maximized at $q=1-\varphi$, where $\varphi=(\sqrt{5}+1)/2$
is the golden ratio, which results in the \emph{fully explicit} capacity upper bound of $(\log \varphi)/2
\approx 0.347120 \text{ (bits per channel use)}$
for the deletion channel with $d=1/2$. We may then interpolate between this bound
and the trivial values at $d=0,1$ using a convexity technique\footnote{%
For extending the bounds to 
$d<1/2$, the results of \cite{ref:RD15}
are not tight, so in this regime
we rely on the plausible conjecture that the
capacity function is convex \cite{ref:Dal11}.} of 
\cite{ref:RD15}, 
thereby obtaining fully explicit capacity upper bounds for general $d$ that
are proved without any need for numerical computation
(Corollary~\ref{coro:del:convexify}).

As in the Poisson case with \eqref{eqn:expo:PoiYfirst}, the inverse binomial
distribution suffers from a constant asymptotic gap of $1/2$ in the KKT conditions
\eqref{eqn:expo:dualKL} (Figure~\ref{fig:BinRpX}). By examining the connection between 
\eqref{eqn:expo:PoiYfirst} and the digamma distribution \eqref{eqn:expo:PoiYre}, we develop a
systematic ``truncation technique'' (made precise in Section~\ref{sec:Poi:truncation})
that allows us to refine \eqref{eqn:expo:invbin} to sharply
eliminate the gap for the binomial case as well. 
To begin with, we prove that enforcing the KKT conditions \eqref{eqn:expo:dualKL} with
equality \emph{for all} $x \geq 0$ results in a \emph{unique} class of
solutions for the distribution of $Y$, which is exactly
\begin{equation} \label{eqn:expo:KKToptimal}
\Pr[Y=y]=y_0 q^y \exp\left(\TextStyle -\sum_{k=0}^y \binom{y}{k} \frac{1}{p^k} \left(1-\frac{1}{p}\right)^{y-k} H(\Bin_{k,p})\right)
=: y_0 q^y \exp\Big(g(y)-y h(p)/p\Big)/y!.
\end{equation}
Therefore, if such a distribution feasibly exists, it would necessarily be capacity achieving. 
However, we observe that the term inside the exponent 
(what we have labeled as $g(y)$ in \eqref{eqn:expo:KKToptimal}) 
exponentially grows in $y$, and therefore, there
is no normalizing constant\footnote{%
Interestingly, using \eqref{eqn:expo:dualKL}, 
this shows that while the optimal input distribution must have infinite 
support, it cannot have a full support. See Remark~\ref{rem:bin:noFull},
and, similarly for the Poisson case, Remark~\ref{rem:Poi:noFull}.
} $y_0$ that would make \eqref{eqn:expo:KKToptimal} a valid
distribution for any $q>0$. Our proposed truncation technique adjusts the exponent of this
alleged optimal solution so as to make its growth rate manageable, while still
satisfying the KKT conditions \eqref{eqn:expo:dualKL} with a potentially
nonzero gap that exponentially decays in $x$. In order to do so, we first prove
the following integral expression for $g(y)$ in \eqref{eqn:expo:KKToptimal}: 
\begin{equation} \label{eqn:expo:Enp}
g(y)=\cE_{1/p}(y)-\cE_{1/p-1}(y)\text{, where } \cE_\eps(y)= \TextStyle
 \int_0^1 \frac{ 1-\eps t y-(1-\eps t)^y}{t \log(1-t)}\,dt.
\end{equation}
We show that $\cE_\eps(y)$ exponentially grows in $y$ when $\eps > 1$,
and grows as $\eps y(\log(\eps y)-1)+o(y)$ 
when $\eps \leq 1$ (Claim~\ref{claim:Enp:asymptotic}).
The truncation technique would involve the truncation of the upper bound
of the integral that defines $\cE_\eps(y)$ (when $\eps>1$)
to $1/\eps$. The resulting function, that we call $\Lam_{1/\eps}(y)$,
may be written, after a change of variables, as 
\InLine{\Lam_\eps(y):=
 \int_0^1 \frac{ 1-t y-(1-t)^y}{t \log(1-\eps t)}\,dt.}
We remark that $\Lam_1(y)=\cE_1(y)=\log\Gamma(1+y)$.
When $\eps \leq 1$, the growth rate of $\Lam_\eps(y)$ is
\InLine{(y/\eps) (\log(y/\eps)+\li(1-\eps)-1)+o(y),} where $\li(\cdot)$
is the logarithmic integral \eqref{eqn:Li:def}. 
Using the factorial moments of the binomial distribution,
we show that $\E[\cE_\eps(Y_x)]=\cE_{\eps p}(x)$ (Proposition~\ref{prop:Enp}),
and furthermore, that
\[
\E[\Lam_\eps(Y_x)]=
\Ex_{p/\eps}(x)+xp\li(1-\eps)/\eps-\eta(1-\eps)+ \TextStyle
\int_\eps^1 \frac{ (1-tp/\eps)^x }{t \log(1-t)}\,dt\text{, where } 
\eta(z) := \int_0^z \frac{dt}{(1-t) \log t}.
\]
From the above results, it follows that, letting
\[
g_p(y):=\begin{cases}
\Lam_p(y)-\Lam_{p/(1-p)}(y)+
\frac{y}{p}\big( (1-p)\li(\frac{1-2p}{1-p})-\li(1-p) \big)+
\eta(1-p)-\eta(\frac{1-2p}{1-p})
& p \in (0,1/2), \\
\Lam_p(y)-\Ex_{1/p-1}(y)-y \li(1-p)/p+\eta(1-p) & p \in [1/2,1],
\end{cases}
\]
and replacing the $g(y)$ in \eqref{eqn:expo:KKToptimal} with
$g_p(y)$ (except for the special case $g(0)=0$) results in a refined distribution
for $Y$ that sharply satisfies \eqref{eqn:expo:dualKL}.
Despite the complex-looking expression defining the above distribution, as we see in
Section~\ref{sec:del:general:trunc}, the distribution converges
pointwise to the dual-feasible solution \eqref{eqn:expo:PoiYre} (the digamma distribution)
for the Poisson case as $p \to 0$; therefore, it is indeed a generalization
of \eqref{eqn:expo:PoiYre} to arbitrary values of $p$.
The gap to equality in the KKT conditions can be explicitly
computed in integral form (using the above results for the expectations
of $\cE_\eps$ and $\Lam_\eps$), which we show to be
\InLine{%
R_p(x) = \TextStyle \int_{p}^1 \frac{(1-t)^x - (1-p)^x}{t \log(1-t)}\, dt
-\int_{\frac{p}{1-p}}^1 \frac{(1-(1-p)t)^x - (1-p)^x}{t \log(1-t)}\, dt,}
with the second integral understood to be zero for $p \geq 1/2$.
This gap is zero at $x=0$, exponentially decays as $x$ grows,
and converges to $xp E_1(xp)$ (as in the Poisson case) 
for $p \to 0$ (see Figure~\ref{fig:BinRpX}).
Hence, we obtain several capacity upper bounds, of varying complexities, 
for the deletion channel. This depends on the chosen 
dual feasible distribution for $Y$, that is, either the truncated
variation of \eqref{eqn:expo:KKToptimal} or the inverse
binomial distribution \eqref{eqn:expo:invbin}, or in the latter
case, whether the function in \eqref{eqn:expo:del:capa}
is numerically computed or upper bounded by either elementary
or standard special functions 
(Figures \ref{fig:DelSlopes}~and~\ref{fig:DelSlopesUB}). The resulting bounds are depicted
in Figure~\ref{fig:DelPlot} and are listed in Table~\ref{tab:DelCapData}.
For the limiting case $d \to 1$, we see that our best upper bound estimate
is $0.4644(1-d)$ which comes quite close to the computer-assisted upper
bound $0.4143(1-d)$ reported in \cite{ref:RD15}, and substantially
improves the $0.7918(1-d)$ in \cite{ref:MDP07} (which is also computer-assisted).
Our fully explicit upper bound of $(1-d) \log \varphi \approx 0.694242(1-d)$ 
is also better than what reported in \cite{ref:MDP07}.
Since our methods for the Poisson-repeat channel converge
to what we obtain for the deletion channel in the limit $d \to 1$, 
we obtain the same upper bound estimate of $0.4644(1-d)$
for the capacity of the Poisson-repeat channel with deletion
probability $d \to 1$ as well.
Finally, we analyze our results for the limiting case $d\to 0$
(Section~\ref{sec:del:smallD}). We prove that, in this regime,
our upper bounds exhibit the asymptotic behavior of
$1-\Theta(h(d))$ which is known to be the case for the
true capacity of the deletion channel \cite{ref:KMS10,ref:KM13}.

\newcommand{\SecDiscussion}{%
\section{Discussion and open problems}


We introduced a number of new techniques that
leave plenty of room for improvement in execution and
lead to intriguing problems for future investigation.
The first is to understand the loss in 
the capacity upper bound \eqref{eqn:expo:capa:c}.
Recall that the Markov chain representation $X-Z-Y$ of a $\cD$-repeat
channel in Section~\ref{sec:expo:reduction} (Figure~\ref{fig:processor}) 
is exact. However, 
as shown in \eqref{eqn:IXYupper}, the potential loss
would correspond to the term $I(Y;Z|X)$.
Developing techniques for lower bounding this conditional
mutual information for the optimal input
distribution $X$ would readily yield an improvement
in the capacity upper bound (assuming that one can show
a general $\Omega(n)$ lower bound on this conditional mutual
information).

%
%

Another intriguing problem is to further improve the quality of the
dual feasible distributions that we introduced for the Poisson-repeat
and deletion channels (i.e., the digamma distribution \eqref{eqn:expo:PoiYre}
and the truncated variation of \eqref{eqn:expo:KKToptimal}). 
Can the truncation technique of 
Sections \ref{sec:Poi:truncation}~and~\ref{sec:del:general:trunc}
be further refined to result in
even better capacity upper bounds for either channel?
As we observe in Remarks \ref{rem:Poi:noFull}~and~\ref{rem:bin:noFull},
the optimal input distributions for mean-limited Poisson and binomial
channels cannot have full support on
all non-negative integers, although they must have infinite supports. 
Our intuition based on the variances suggests that the optimal input distribution must 
actually be quite sparse; e.g., supported on points $\Theta(i^2/\lam)$
for the Poisson case and $\Theta(i^2(1-p)/p)$ for the binomial case
($i=0,1,\ldots$). A further intuition is that the KKT gap in \eqref{eqn:expo:dualKL}
attained by the optimal output distribution should take the general form of
the gaps attained by our dual-feasible distributions (Figures \ref{fig:Ei} and \ref{fig:BinRpX}), 
but additionally, oscillate back and forth to zero (while remaining positive)
and reach zero exactly at the sparse set of points supported by the
optimal input distribution. 

In the $X-Z-Y$ Markov chain representation of the deletion channel (Figure~\ref{fig:processor}), 
observe that each $Z_i$ is the sum of a geometric number of the
entries in the run-length representation $\Dots{X_1}{X_2}{X_t}$ of $X$,
and that $Y_i$ is obtained by passing $Z_i-1$ through a binomial channel.
Let $Z^\star$ be the optimal input distribution for a single use of the binomial channel,
and $\chi$ be the characteristic function of the distribution of $1+Z^\star$.
We say the distribution is \emph{geometrically infinitely divisible} if
$\Exp{1-1/\chi(t)}$ is an infinitely divisible characteristic function \cite{ref:KMM85}. 
In this case, for all $r \in (0,1]$, one can identify a random variable
$Z'$ such that $1+Z^\star$ is the sum of a geometric number (with mean $1/r$) 
of independent copies of $Z'$. Then, one may hope to set up the run-length
distribution of the input sequence $X$ to be i.i.d.\ from a
distribution $X'$ such that,
for an appropriate $r$, sum of a geometric (with mean $1/r$) copies of
$X'$ gives the distribution of $1+Z^\star$.
If $r$ is chosen appropriately (i.e., such that it coincides with the deletion probability
of an entire run in $X$), the distribution of $Z_1-1,Z_2-1,\ldots$
would form a sequence of i.i.d.\ copies of $Z^\star$; i.e., the optimal 
input distribution for the $Z-Y$ link. In this case, one may hope to show
that the resulting input distribution $X$ would be capacity achieving for the
deletion channel. We note, however, that currently there is no general consensus 
on whether the capacity achieving input distribution for the deletion channel
must consist of i.i.d.\ run-lengths. This is known to be the case for $d \to 0$
\cite{ref:KM13}. The optimality of i.i.d.\ run-lengths has not been
ruled out for any $d$, and indeed the above intuition on geometric infinite
divisibility may suggest this as a possibility.

We obtain sharp estimates on the functions inside the supremums 
in \eqref{eqn:expo:Poi:capa} and \eqref{eqn:expo:del:capa}
in terms of elementary or standard special functions. 
Can the supremums themselves (i.e., the capacity upper bounds)
be upper bounded in terms of 
such explicit functions? An effort towards this goal is
demonstrated in Appendix~\ref{app:claim:cber:analytical}.
Furthermore, our work motivates the study of several novel discrete probability
distributions that are worth further consideration. 



}
\iftoggle{full}{}{\SecDiscussion}

\iftoggle{full}{}{\part{Details of the results}}

\iftoggle{full}{}{%
\subsection*{Organization}
\SecOrganization
}

\iftoggle{full}{}{%
\subsection*{Notation}
\SecNotation
}

\section{General mean-limited and convolution channels}
\label{sec:meanlim}

In this section, we consider general classes of channels that we call
``general mean-limited'' and ``convolution channels''. 
The input and output alphabet for these
channels is the set $\R^{\geq 0}$ of non-negative reals.
In general, any such channel can be described by a probability
transition rule  
$\cP(y|x)$ over the non-negative reals, 
where $x, y \in \R^{\geq 0}$. 
The channel takes an input $X \in \R^{\geq 0}$
and outputs the output random variable $Y$ according to the rule
$\cP(y|x)$.
Since the capacity of this channel may in general be infinite,
we restrict the set of possible input distributions by defining
a mean constraint $\E[Y]=\mu$, 
for a parameter $\mu > 0$, and 
use the notation $\ch_\mu(\cP)$ and the terminology
\emph{general mean-limited channel} for the channel 
defined with respect to the transition rule $\cP(y|x)$ and
mean constraint $\mu$. 
The rate achieved by an input distribution for this channel
is defined as $I(X;Y)/\E[X]$, and naturally, the capacity
is the supremum of the achievable rates subject to the given
mean constraint.

A natural choice for
the transition rule $\cP$, that results in what we call a
convolution channel, is via a multiplicative noise
distribution $\cD$ over non-negative reals. For an $x \in \R^{\geq 0}$, let 
$\cD^{\oplus x}$ denote the distribution attained by raising
the characteristic function of $\cD$ to the power $x$. When $x$
is a positive integer, this would correspond to adding together $x$
independent samples from $\cD$, or equivalently, the $x$-fold
convolution of the distribution $\cD$ with itself (hence the terminology
``convolution channel''). The convolution channel defined
with respect to $\cD$, that we use the overloaded notation $\ch_\mu(\cD)$ for,
takes an input $X \in \R^{\geq 0}$ and outputs a sample from $\cD^{\oplus X}$.

Let $\lam := \E[\cD]$.
Note that, since $\E[Y]=\lam \E[X]$, 
the rate achieved by an input distribution $X$
would be $\lam I(X;Y)/\mu$, and the capacity is simply
$(\lam/\mu) \sup I(X;Y)$, where the supremum is over the
input distributions $X$ satisfying $\E[X]=\mu/\lam$.

A convolution channel, or more generally any channel $\ch_\mu(\cP)$,
can be defined over continuous or discrete distributions.
In this work, for the sake of concreteness and a unified notation,
we focus on the discrete case (and in fact, the integer case). 
However,
the results can be readily extended to continuous distributions
if differential entropy and mutual information are used
(rather than the discrete Shannon entropy) to measure information,
and summations are replaced by the the analogous integrals.

 \subsection{Capacity of general mean-limited channels} 
\label{sec:capa:meanlim}

In this section, we characterize the capacity of a general mean-limited
channel $\ch_\mu(\cP)$, defined with respect to a probability transition rule
$\cP(y|x)$ and output mean constraint $\mu$, as the optimum solution of
a convex program. This particularly provides the technical tool for analyzing 
the capacity of convolution channels, and subsequently, general repeat channels.

Recall that the capacity of $\ch_\mu(\cP)$ is the supremum mutual information
$I(X;Y)$ between the input and output of the channel, where the supremum
is taken with respect to all input distributions $X$ whose corresponding
output distribution $Y$ satisfies the given mean constraint $\E[Y]=\mu$.
Although our methods are general and apply to any (continuous or discrete)
transition rule $\cP(y|x)$, in this section we focus on discrete distributions.
In particular, we assume that the input alphabet is a discrete set
$\cX \subseteq \R^{\geq 0}$ and so is the output alphabet
$\cY \subseteq \R^{\geq 0}$ (for the purpose of this work, we 
may think of both $\cX$ and $\cY$ as the set of non-negative integers).

For each fixed $x \in \cX$, let $Y_x$ denote the random variable
$Y$ conditioned on $X=x$ (i.e., $Y_x$ is the output of
the channel when a fixed input $x$ is given). 
The problem of maximizing the mutual information 
\begin{equation} \nonumber 
I(X;Y)=H(Y)-H(Y|X)=\sum_{x\in\cX} \Pr[X=x] \KL{Y_x}{Y}, 
\end{equation}
where $\KL{\cdot}{\cdot}$ denotes the 
Kullback-Leibler (KL) divergence,
can naturally
be written as a convex minimization program as follows:

\begin{align}
\underset{X,Y}{\text{minimize}} \qquad 
 & -I(X;Y) =\sum_{y \in \cY} Y(y) \log Y(y) + \sum_{x \in \cX} X(x) H(Y_x)  
 \label{eqn:cvx1} \\
\text{subject to} \qquad 
 & X \succeq 0 \label{eqn:cvx1:pos} \\ 
& \innr{X, \bOne}=1 \label{eqn:cvx1:sum} \\
& \sum_{y \in \cY} y \cdot Y(y) = \mu \label{eqn:cvx1:exp} \\
& \bP  X = Y, \label{eqn:cvx1:trans}
\end{align}
where we have used the following notation: $X(x)$ (resp., $Y(y)$)
in \eqref{eqn:cvx1} and \eqref{eqn:cvx1:exp}
denotes the probability assigned by the distribution of $X$ (resp.,
the distribution of $Y$) to the outcome $x$ (resp., $y$). Moreover,
with a slight abuse of notation, we may think of $X$ and $Y$ as 
vectors of probabilities assigned to the possible outcomes by the distributions that define
the underlying random variables, so that for example 
$\innr{X, \bOne}$ in \eqref{eqn:cvx1:sum}, where $\bOne$ is the all-ones vector,
is a shorthand for the summation of probabilities that define $X$
(which should be equal to one). 
Note that for each fixed $x$, the entropy $H(Y_x)$ in
\eqref{eqn:cvx1} is a constant value 
defined by the transition rule $\cP$. 
In \eqref{eqn:cvx1:trans},
$\bP$ denotes the (infinite dimensional) transition matrix from $X$ to $Y$ whose
entry at row $y$ and column $x$ is equal to $\cP(y|x)$.

Following the standard approach in convex optimization, we define 
slack variables $Y'(y)$, $y \in \cY$, for each constraint in
\eqref{eqn:cvx1:trans}, $X'(x)$, $x \in \cX$, for each non-negativity constraint
in \eqref{eqn:cvx1:pos}, $\nu_0$ for \eqref{eqn:cvx1:sum} and
$\nu_1$ for \eqref{eqn:cvx1:exp}. Now, the Lagrangian $L(X,Y;X',Y',\nu_0,\nu_1)$ for the
program \eqref{eqn:cvx1} may be written as
\begin{align} \nonumber %
L(X,Y;X',Y',\nu_0,\nu_1) &= 
\sum_{y \in \cY} Y(y) \log Y(y) + \sum_{x \in \cX} X(x) H(Y|x)  \\
&
-\innr{X', X} + \nu_0 (\innr{X, \bOne}-1) + \nu_1 \left(\sum_{y \in \cY} y \cdot Y(y) - \mu\right) 
+\innr{Y', \bP X - Y}. \label{eqn:Lag}
\end{align}
%
From \eqref{eqn:Lag},
the derivatives of $L$ with respect to each variable $X(x)$ and $Y(y)$ 
can be written as
\begin{align}
\del{L}{X(x)}&=
H(Y|x)-X'(x)+\nu_0+\innr{Y', \bP^{(x)}}, \label{eqn:cvx1:delX}
\\
\del{L}{Y(y)}&=
1+\log Y(y)+y \nu_1 - Y'(y), \nonumber
\end{align}
where $\bP^{(x)}$ denotes the column of $\bP$ indexed by $x$.

Observe that one can trivially construct a strictly feasible solution
$X \succ 0$ for \eqref{eqn:cvx1}. 
Therefore, Slater's condition holds
and the duality gap for this program is zero.
The dual objective function 
\[
g(X',Y',\nu_0,\nu_1) := \inf_{X,Y} L(X,Y;X',Y',\nu_0,\nu_1)
\]
can now be written by analytically
optimizing $L$ with respect to $X$ and $Y$.
Setting $\del{L}{X(x)}=0$ implies that
\begin{equation} \label{eqn:cvx1:delXzero}
\E_{Y_x} [Y'(Y_x)]=
-H(Y_x)+X'(x)-\nu_0,
\end{equation}
where expectation is taken over the random variable $Y_x$.
This can be deduced from \eqref{eqn:cvx1:delX}
by observing that the product $\innr{Y', \bP^{(x)}}$
is precisely the average of the values defined by
$Y'(y)$, $y \in \cY$, with respect to the measure
defined by the $x$th column of $\bP$; i.e., the
distribution of $Y_x$.
In other words,
\[
\innr{Y', \bP^{(x)}} = \sum_{y \in \cY} \Pr[Y_x=y] \cdot Y'(y) =
\E_{Y_x} [Y'(Y_x)].
\]
Setting $\del{L}{Y(y)}=0$, on the other hand, gives us
\begin{equation} \label{eqn:cvx1:Ysolve}
Y'(y)=1+\log Y(y)+y \nu_1,
\end{equation}
and thus,
\begin{equation} \label{eqn:cvx1:YsolveB}
Y(y) = \exp(-1+Y'(y)-y \nu_1).
\end{equation}

Since $L$ linearly depends on variables $X(x)$ and
a linear function has bounded infimum only when the
function is zero, $g(X',Y',\nu_0,\nu_1)$ is only
finite when \eqref{eqn:cvx1:delXzero} holds for all 
$x$, which we will assume in the sequel.
In this case, we deduce
\begin{align}
g(X',Y',\nu_0,\nu_1) &= 
\sum_{y \in \cY} Y(y) \log Y(y)  
-\nu_0 + \nu_1 \left(\sum_{y \in \cY} y \cdot Y(y) - \mu\right) 
- \innr{Y', Y}\nonumber \\
&\stackrel{\eqref{eqn:cvx1:Ysolve}}{=}
\sum_{y \in \cY} Y(y) ((Y'(y)-1-y\nu_1)+y \nu_1-Y'(y))-\nu_0-\mu \nu_1 \nonumber \\
&\stackrel{\eqref{eqn:cvx1:YsolveB}}{=}
-\sum_{y \in \cY} \exp(-1+Y'(y)-y \nu_1)-\nu_0-\mu \nu_1.
\label{eqn:cvx1:g}
\end{align}
The dual program to \eqref{eqn:cvx1} can now be written,
recalling the constraints \eqref{eqn:cvx1:delXzero}, as
\begin{align*}
\underset{X',Y',\nu_0,\nu_1}{\text{maximize}} \qquad 
 & g(X',Y',\nu_0,\nu_1) \\
\text{subject to} \qquad & (\forall x \in \cX)\colon 
\E_{Y_x} [Y'(Y_x)]=
-H(Y_x)+X'(x)-\nu_0 \\
& X' \succeq \bZero,
\end{align*}
which we rewrite, using \eqref{eqn:cvx1:g}, as a convex minimization problem
\begin{align}
\underset{Y',\nu_0,\nu_1}{\text{minimize}} \qquad 
 & \sum_{y \in \cY} \exp(-1+Y'(y)-y \nu_1)+\nu_0+\mu \nu_1
 \label{eqn:cvx2} \\
\text{subject to} \qquad & (\forall x \in \cX)\colon 
\E_{Y_x} [Y'(Y_x)] \geq
-H(Y_x)-\nu_0. \label{eqn:cvx2:constr}
\end{align}

Now, the objective value achieved by any feasible solution 
to the above dual formulation gives an upper bound on the maximum
attainable mutual information $I(X; Y)$ and thus a capacity
upper bound. 
Furthermore, since \eqref{eqn:cvx1} is convex and satisfies strong duality,
Karush-Kuhn-Tucker (KKT) conditions imply that a primal feasible 
solution $\sX, \sY$ for \eqref{eqn:cvx1} is optimal \Iff\ 
there is a dual feasible solution $(Y',\nu_0,\nu_1)$ such that,
recalling \eqref{eqn:cvx1:Ysolve},
\begin{equation} \label{eqn:cvx1:YsolveC}
Y'(y)=1+\log \sY(y)+y \nu_1,
\end{equation}
for all $y \in \cX$, and moreover,
\eqref{eqn:cvx2:constr} holds with equality for all $x \in \cX$
such that $\sX(x)>0$. In this case, using \eqref{eqn:cvx1:YsolveC},
the dual objective function simplifies to
\begin{align*}
g(X',Y',\nu_0,\nu_1) &=
\sum_{y \in \cY} \exp(-1+Y'(y)-y \nu_1)+\nu_0+\mu \nu_1 \\
&\stackrel{\eqref{eqn:cvx1:YsolveC}}{=}  \sum_{y \in \cY} \sY(y) +\nu_0+\mu \nu_1 \\
&= 1+\nu_0+\mu \nu_1.
\end{align*}
On the other hand, we can write 
\begin{align*}
\E_{Y_x} [Y'(Y_x)] &\stackrel{\eqref{eqn:cvx1:YsolveC}}{=}
\sum_{y \in \cY} \Pr[Y_x=y] (\log \sY(y)+1+y \nu_1) \\&=
1+\nu_1 \E[Y_x] + \sum_{y \in \cY} \Pr[Y_x=y] \log \sY(y) \\
&= 1+\nu_1 \E[Y_x] + \sum_{y \in \cY} \Pr[Y_x=y] \log (\sY(y)/\Pr[Y_x=y]) - H(Y_x) \\
&=  1+\nu_1 \mu_x - \KL{Y_x}{\sY} - H(Y_x), 
\end{align*}
where 
we define $\mu_x := \E[Y_x]$,
so that \eqref{eqn:cvx2:constr} can be rewritten as 
\begin{equation*}
(\forall x \in \cX)\colon \KL{Y_x}{\sY}\leq 1+\nu_1 \mu_x+\nu_0.
\end{equation*}
We have thus proven the following result
(written with a trivial change of the variable $\nu_0$):
\begin{thm} \label{thm:meanlimited}
Consider a general mean-limited channel $\ch_\mu(\cP)$ with
output alphabet $\cY$, and let $Y$ be any distribution 
over $\cY$. 
Denote by the random variable $Y_x$ the output of the channel given $x$
as the input, and
let $\nu_0$ and $\nu_1$ be any real parameters such that\footnote{It is
worthwhile to note that having the term $\nu_1 \E[Y_x]$ on
the right hand side of \eqref{eqn:cvx1:KKT}, as opposed
to $\nu_1 x$, is quite useful for finding natural dual feasible distributions,
and causes a mean-adjusting term of the form $q^y$ in the distribution of $Y$
to be naturally absorbed in the constant $\nu_1$. However, this would
not necessarily be the case if, on the right hand side, we had $\nu_1 x$
(as would be the case if the techniques of \cite{ref:AG93} were applied). 
For the special case of the convolution channels, 
$\E[Y_x]$ is a linear function of $x$ and the distinction disappears.
However, our best upper bounds (Theorem~\ref{thm:Cupper}) consider
non-convolution channels as well.
}
\begin{align} \label{eqn:cvx1:dualF} 
(\forall x \in \cX)\colon &\KL{Y_x}{Y} \leq 
\nu_1 \E[Y_x] +\nu_0. 
\end{align} 
\noindent Then, capacity of $\ch_\mu(\cP)$ is at most
$\nu_1 \mu + \nu_0$.
Furthermore, the capacity is exactly 
$\nu_1 \mu + \nu_0$ \Iff\ there is a distribution
$X$ over the input alphabet such that the corresponding
output distribution is $Y$ and moreover,
%
\begin{align}
(\forall x \in \supp(X))\colon &\KL{Y_x}{Y}= 
\nu_1 \E[Y_x] +\nu_0. 
\label{eqn:cvx1:KKT} 
\end{align}
\end{thm}


\subsection{Extension to multiple uses}
\label{sec:meanlimited:multiple}

We now extend the notion of 
general mean-limited channels to multiple uses and
observe that the capacity is achieved by
a product input distribution.

Let $\ch_\mu^m(\cP)$ be the $m$-fold concatenation
of the channel $\ch(\cP)$; i.e., the channel takes an 
$m$-dimensional input vector $X=(X_1, \ldots, X_m)$ 
and applies the transition rule on each $X_i$ independently,
resulting in an output vector $Y=(Y_1, \ldots, Y_m)$. 
In this case, the input distribution is allowed to be
arbitrary subject to a \emph{total mean constraint}
$\E[Y_1+\cdots+Y_m]=\mu$. Naturally, the achieved
rate by an input distribution $X$ is 
$I(X;Y)/\E[X_1+\cdots+X_m]$ and the capacity of the
$m$-fold channel is the supremum of the achievable rates.

It is straightforward to argue that the capacity of
$\ch_\mu^m(\cP)$ is achieved by an independent input
distribution. To see this, consider any input distribution
$X$ satisfying the total mean constraint. Let 
$X'=(X'_1, \ldots, X'_m)$ denote a distribution of $m$
independent real numbers such that the marginal distribution
of $X'_i$ is identical to that of $X_i$ for all $i \in [m]$. Clearly, 
this would mean that the output distribution $Y'=(Y'_1,\ldots,Y'_m)$
corresponding to $X'$ consists of $m$ independent entries,
where each $Y'_i$ has the same marginal distribution as $Y_i$.
Therefore, $\E[Y'_1+\cdots+Y'_m]=\E[Y_1+\cdots+Y_m]=\mu$ and thus,
$X$ also satisfies the total mean constraint. However, we
may show that the rate achieved by $X'$ is no less than that of
$X$, as follows:
\begin{align*}
I(X;Y) &= \sum_{i=1}^m I(X_i;Y_i|X_1,\ldots,X_{i-1}) \\
&= \sum_{i=1}^m (H(Y_i|X_1,\ldots,X_{i-1})-H(Y_i|X_1,\ldots,X_{i})) \\
&\leq \sum_{i=1}^m (H(Y_i)-H(Y_i|X_1,\ldots,X_{i})) \\
&= \sum_{i=1}^m (H(Y_i)-H(Y_i|X_i)) \\
&= \sum_{i=1}^m I(X_i;Y_i) 
= \sum_{i=1}^m I(X'_i;Y'_i) = I(X'; Y').
\end{align*}
Therefore, we have proved the following:
\begin{lem} \label{lem:product}
Let $\ch_\mu^m(\cP)$ be an $m$-use general mean-limited
channel. Then, there is a capacity achieving input distribution
$(X_1, X_2, \ldots, X_m)$ which is a product distribution.
\end{lem}


\section{General repeat channels}
\label{sec:general:repeat}

A natural model to generalize both deletion and Poisson-repeat
channels is the following: Let $\cD$ be a distribution on 
non-negative integers. The $\cD$-repeat channel is defined to
receive a (possibly infinite) stream of bits and
replaces each bit independently with a number of copies of 
the bit distributed according to $\cD$. We call such a channel
a \emph{general repeat channel} with respect to the repetition rule $\cD$.
The binary deletion
channel and Poisson-repeat channels are 
repeat channels with respect to the Bernoulli and
Poisson repetition rules, respectively. 

To characterize the capacity of a 
$\cD$-repeat channel, without loss of generality, one may
assume that the first-ever bit given to the channel is not
deleted by the channel (i.e., it will be replicated at least
once). The effect of this assumption on the capacity is amortized
down to zero as the number of channel uses tends to infinity,
and thus we shall make this assumption in the sequel.
Similarly, we may assume that the first input bit ever given to the channel
is a $0$, as this assumption will also have no effect on the
capacity of the channel.

Consider any input distribution on $n$-bit sequences
$X=\Dots{B_1}{B_2}{B_n}$, where we think of $n$ as growing to infinity.
If we know the first bit of $X$, we may equivalently think of it 
as its run-length
encoding which we denote by a sequence of positive integers
$\Dots{X_1}{X_2}{X_t}$, where
we also think of $t$ (where $1 \leq t \leq n$) as growing to infinity.
Let the ``deletion probability'' 
$d(\cD)$ denote the probability assigned to the zero outcome by $\cD$,
and $p(\cD):=1-d(\cD)$.

We now aim to model the behavior of the $\cD$-repeat in a way
that can be analyzed more conveniently. Towards this goal, consider
the following \emph{pre-processor} procedure on a bit sequence $X$:

\begin{enumerate}
\item Let $p := p(\cD)$.
Given the input sequence $X$, draw a geometrically
distributed random variable $G\geq 1$ with mean $1/p$.

\item Let the bit sequence $X'=\Dots{B'_1}{B'_2}{B'_{n'}}$ be the sequence of
``even runs'' in $X$; i.e., the bits of $X$ corresponding
to the runs $\DotsIII{X_2}{X_4}{X_6}$.
Let $i$ be so that $B_i$, the $i$th bit of $X$,
corresponds to the $G$th bit in $X'$ 
(if $G > n'$, output $n-n'$ and terminate).
Suppose that the bit $B_i$ corresponds to the run $X_{2j}$ in $X$.

\item Output the integer $Z:=X_1+X_3+\cdots+X_{2j-1}$.
Repeat the procedure with $X=\Dots{B_i}{B_{i+1}}{B_n}$.

\end{enumerate}
Note that, since each iteration of the above procedure
eliminates at least one of the $X_i$, the number of integers
$M$ output by the procedure (which are all positive) must necessarily
satisfy $1 \leq M \leq n$.

Denote by $\overline{\cD}$ the distribution $\cD$ conditioned on the
outcome being nonzero.
We now define an auxiliary, \emph{run-processor} channel as follows:
The channel receives, at input, a sequence of positive integers
$\DotsAZ{Z_1}{Z_M}$. For each $i = 1,\ldots, M$, the channel
independently computes $Y_i$ as will be described next and then outputs the
sequence $\DotsAZ{Y_1}{Y_M}$. To compute $Y_i$, the channel
outputs a sample from $\overline{\cD}\oplus \cD^{\oplus (Z_i-1)}$,
where $\oplus$ denotes convolution of distributions (i.e.,
the distribution of independent random samples
from each component added together).
A schematic diagram of the above pre-processor and
run-processor procedures appears in Figure~\ref{fig:processor}.
We now show that the combination of the pre-processor and the
run-processor is statistically equivalent to action of the $\cD$-repeat channel.

\begin{figure}
\begin{center}
\resizebox{\columnwidth}{!}{
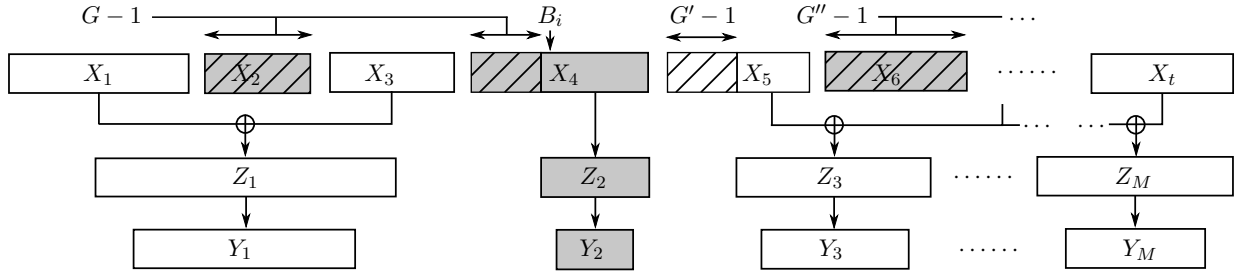
}
\end{center}
\caption{A diagram of the pre-processor and the run-processor.
The white blocks are runs of zeros and the gray blocks are runs of ones,
and the hatched portions represent the bits that are deleted by
the channel. 
The random variables $G$, $G'$ and $G''$ respectively denote the random choices of $G$
in the first three iterations of the pre-processor.
In this example, a total of $G-1$ bits of $X_2$ and $X_4$ are deleted,
as well as $G'-1$ bits of $X_5$ and $G''-1$ bits of $X_6,X_8,\ldots$.
Consequently, $Z_1$ consists of the concatenation of the
runs corresponding to $X_1$ and $X_3$, the block $Z_2$ represents
the remaining portion of $X_4$, and so forth.
Finally, each $Y_i$ is obtained from the corresponding $Z_i$ according to
the repetition rule $\cD$.
}
\label{fig:processor}
\end{figure}

\begin{lem}
\label{lem:processing}
Let $X=\DotsAZ{B_1}{B_n}$ be a bit-sequence and
$Z=\DotsAZ{Z_1}{Z_M}$ be the output of the pre-processor
on $X$ with respect to a distribution $\cD$. 
Let $Y=\DotsAZ{Y_1}{Y_M}$ be the output of the
run-processor channel on $Z$ (with respect to $\cD$). 
Then, the distribution of
$Y$ is identical to the run-length encoding of the
output of the $\cD$-repeat channel on $X$.
\end{lem}

\begin{proof}
Recall the convention that the first bit of the input sequence
$X$ is assumed to be zero, and that this bit is not deleted by the
channel.
Given $X$, the $\cD$-repeat channel replaces each $B_i$ independently
with $D_i \geq 0$ copies sampled from $\cD$ (except for $D_1$, which
is sampled from the conditional distribution $\overline{\cD}$). 
Since the decision is
made independently for each $B_i$, the $D_i$ can be sampled in any
order without a change in the resulting output distribution. Suppose, therefore,
that the channel first decides on the even runs corresponding to 
$X_2, X_4, \ldots$. Namely, the channel can be thought of as performing
the following procedure to produce the output distribution:
\begin{enumerate}
\item For each bit in the input sequence represented by 
the even runs $X_2, X_4, \ldots$, decide, in order, whether the bit 
is deleted (i.e., replaced by zero copies), until the first
non-deletion is found. Suppose the first non-deletion
occurs at some $B_i$ located in block $2j$.

\item Consider now the odd runs $X_1, \ldots, X_{2j-1}$,
up to the non-deletion position. 
Replace $B_1$ (which ought not be deleted) by one
or more copies as determined by a fresh sample from $\overline{\cD}$.
Furthermore, independently replace
each remaining bit in the odd runs with a number of copies
determined by a fresh sample from $\cD$.

\item Restart the process on the remaining bits
$B_i, B_{i+1}, \ldots$, until the input sequence is exhausted.
\end{enumerate}
Note that, by the end of each execution of the above procedure,
we have revealed that $B_i$ is a non-deletion position.
In the subsequent round, $B_i$ will be the first
bit in the remaining input sequence, and this is consistent with the 
assumption that the first input bit in each round is 
replaced by a number of copies sampled from 
$\overline{\cD}$ (i.e., $\cD$ conditioned on the 
nonzero outcomes), as opposed to $\cD$.

We see that the number of zeros output by the first
iteration of the above process determines
the distribution of the first run-length of the output of
the $\cD$-repeat channel given $X$, which is $Y_1$. Moreover, 
conditioned on the outcome of the first iteration,
the second iteration determines the distribution
of the second run of the output of the $\cD$-repeat
channel, i.e., $Y_2$, given $X$ and $Y_1$. Inductively,
we see that the $\ell$th iteration of the above procedure,
conditioned on the outcome of the first $\ell-1$ iterations,
produces a number of bits distributed according to the
conditional distribution $Y_\ell$ given $X, Y_1,\ldots, Y_{\ell-1}$.
Therefore, the entire procedure samples $Y_1, \ldots, Y_M$
according to the exact run-length encoding distribution
of the output of the $\cD$-repeat channel, given the input $X$.

We now observe that the first step of the above procedure
corresponds to the first step of the pre-processor; i.e.,
the geometric random variable $G$ determines the position
of the first non-deletion $B_i$ within the even runs $\DotsII{X_2}{X_4}$
(recall that $1-p(\cD)$ is the deletion probability of
a bit). Furthermore, observe that the 
combined length of the odd runs $X_1+X_3+\cdots+X_{2j-1}$
is exactly the random variable $Z$ output by each iteration
of the pre-processor.

Now, the run-processor takes each $Z_\ell$ produced
by the $\ell$th execution of the pre-processor and 
replaces it with $Y_\ell$, which is an independent
sample from $\overline{\cD} \oplus \cD^{\oplus Z_\ell-1}$.
Observe that a sample from the above convolution corresponds
to the number of bits generated by taking a run of bits
of length $Z_\ell$, replacing the first bit with a non-zero
number of copies sampled from $\overline{\cD}$ and replacing
each remaining bit by a number of copies independently
sampled from $\cD$. This is precisely what the above
procedure does with the combined bits from the odd
runs $\Dots{X_1}{X_3}{X_{2j-1}}$ in each round. That is,
the integer $Y_\ell$ output by the run-processor
(conditioned on $\Dots{Y_1}{Y_2}{Y_{\ell-1}}$) has the same
distribution as the number of bits output by the above
procedure in round $\ell$ (conditioned on the transcript
of the execution the procedure for the first $\ell-1$ rounds).

We conclude that the pre-processor followed by the run-processor
generate an integer sequence that has the same distribution
as the run-length encoding of the above procedure, which in turn
has the same distribution as the run-length encoding of the output of the $\cD$-repeat
channel given input $X$.
\end{proof}

In light of Lemma~\ref{lem:processing}, characterizing
the capacity of a $\cD$-repeat channel (that we denote by
$\capa(\cD)$) is equivalent
to characterizing capacity of the cascade of the pre-processor and the
run-processor.
Let this cascade channel be denoted by the Markov chain $X-Z-Y$.
We are interested in
\[
\capa(\cD) = \lim_{n\to\infty} \sup \frac{I(X;Y)}{n},
\]
where the supremum is taken over all input distributions $X_1 \ldots X_n$.
Unlike channels with no synchronization errors (e.g., binary
symmetric or erasure channels), it is not trivial to show that 
the above limit exists, and is equal to the capacity. However, 
the identity follows from \cite{ref:Dob67}.
Note that
\begin{equation} \label{eqn:IXYupper}
I(X;Y) = H(Y)-H(Y|X)=H(Y)-(H(Y|X,Z)+I(Y;Z|X))\leq H(Y)-H(Y|Z)=I(Y;Z),
\end{equation}
so that \begin{equation} \label{eqn:IYZ}
\capa(\cD) \leq \lim_{n\to\infty} \sup \frac{I(Y;Z)}{n},
\end{equation}
where
the supremum is still over the distribution of $X$.
A major technical tool that we introduce is the following theorem, which reduces the
task of upper bounding the capacity of a $\cD$-repeat channel to a capacity
upper bound problem for a related mean-limited channel. A proof of this
result appears in Section~\ref{sec:thm:Cupper:proof}.

\begin{thm} \label{thm:Cupper}
Consider a $\cD$-repeat channel $\ch$ and let $\ch_\mu(\cP)$ be the general
mean-limited channel with respect to the transition rule $\cP$ over positive
integer inputs that, given an integer $1+x$, outputs a sample from the
convolution $\overline{\cD}\oplus \cD^{\oplus x}$, where 
$\overline{\cD}$ is the distribution $\cD$ conditioned on the
outcome being nonzero. Let $\lam := \E[\cD]$, 
$\overline{\lam}:=\E[\overline{\cD}]$, and
$1-p$ be the probability assigned
to the outcome zero by $\cD$. Then,
\begin{equation}
\label{eqn:Cupper}
\capa(\ch) \leq \sup_{\mu\geq \overline{\lam}} \frac{\capa(\ch_\mu(\cP))}
{1/p+(\mu-\overline{\lam})/\lam}.
\end{equation}
\end{thm}

A slightly simpler result to apply is the following corollary of 
Theorem~\ref{thm:Cupper}:

\begin{coro} \label{coro:Cupper}
Consider a $\cD$-repeat channel $\ch$ and let $\ch_\mu(\cD)$ be a
convolution channel corresponding to $\cD$ and restricted to
non-negative integer inputs. 
Let $1-p$ be the probability assigned
to the outcome zero by $\cD$, and $\lam := \E[\cD]$. Then,
the capacity of the $\cD$-repeat channel can be upper bounded as
\begin{equation}
\label{eqn:coro:Cupper}
\capa(\ch) \leq \sup_{\mu\geq 0} \frac{\capa(\ch_\mu(\cD))}{1/p+\mu/\lam}.
\end{equation}
\end{coro}

\begin{proof}
\newcommand{\mus}{{\mu^{\star}}}
Let $\ch_\mu(\cP)$ be the channel defined in the statement 
of Theorem~\ref{thm:Cupper}. The mean-limited channel corresponding
to the transition rule $\cP$ receives a positive integer $x$,
and outputs a summation $Y=Y_0+Y_1+\cdots+Y_{x-1}$ of independent
random variables where $Y_0$ is sampled from $\overline{\cD}$
(i.e., $\cD$ conditioned on the outcome being nonzero) and
the rest are sampled from $\cD$. 
Let $\ch'_\mu(\cP)$ be 
a modification of $\ch_\mu(\cP)$ with side information, in which
the receiver also receives the exact value of $Y_0$. This
side information can only increase the capacity of the channel
for the corresponding parameter $\mu$. 
Let $\overline{\lam} := \E[\overline{\cD}]$. Since the input
to $\ch_\mu(\cP)$ is a positive integer, the modified channel
$\ch'_\mu(\cP)$
is equivalent to, and has the same capacity as, the convolution channel 
$\ch_{\mu-\overline{\lam}}(\cD)$ (with the input restricted to non-negative integers).
This is due to the fact that the receiver may simply subtract the
given value of $Y_0$, which is independent of the input and thus
bears no information about the input, from $Y$
and thereby simulate a convolution channel with the matching
mean constraint, which is $\ch_{\mu-\overline{\lam}}(\cD)$.
This means that $\capa(\ch_\mu(\cP)) \leq \capa(\ch_{\mu-\overline{\lam}}(\cD))$.

Let $\mus > \mu_0$ be
a value of $\mu$
that attains the supremum in \eqref{eqn:Cupper}.
We now have that 
\[\frac{\capa(\ch_{\mus}(\cP))}{1/p+(\mus-\overline{\lam})/\lam} \leq 
\frac{\capa(\ch_{\mus-\overline{\lam}}(\cD))}{1/p+(\mus-\overline{\lam})/\lam} 
\leq \sup_{\mu\geq 0} \frac{\capa(\ch_\mu(\cD))}{1/p+\mu/\lam},
\]
proving the claim.
\end{proof}

\begin{remark}
Compared with Theorem~\ref{thm:Cupper}, Corollary~\ref{coro:Cupper}
is in general more convenient to work with. This is due to the 
fact that the normalizing constant in the probability mass function of the
conditional distribution $\overline{\cD}$ in Theorem~\ref{thm:Cupper}
incurs an additive factor in the entropy expression for $\overline{\cD}$ 
that, in general,
may be of little effect but nevertheless cause significant technical difficulties. 
However, this convenience comes at cost of potentially obtaining worse capacity
upper bounds than what Theorem~\ref{thm:Cupper} would give. For the 
case of the deletion channel, $\cD$ is a Bernoulli distribution and
$\overline{\cD}$ becomes a trivial, singleton, distribution. Therefore,
in this case, Corollary~\ref{coro:Cupper} can obtain the same result
as Theorem~\ref{thm:Cupper}. However, for channels for which
$\overline{\cD}$ contains substantial entropy; e.g., the Poisson-repeat
channel where $\cD$ has a large mean,
the loss incurred by applying Corollary~\ref{coro:Cupper} 
rather than Theorem~\ref{thm:Cupper} may be noticeable and even potentially trivialize the resulting upper bounds.  
\end{remark}

\subsection{Proof of Theorem~\ref{thm:Cupper}}
\label{sec:thm:Cupper:proof}

In order to prove Theorem~\ref{thm:Cupper}, we first recall \eqref{eqn:IYZ}; i.e.,
\[
\capa(\cD) \leq \lim_{n\to \infty}\sup \frac{I(Y;Z)}{n},
\]
where the supremum is over the distribution of the $n$-bit
input sequence $X$ (and $Y=\DotsAZ{Y_1}{Y_M}$ and $Z=\DotsAZ{Z_1}{Z_M}$ being the corresponding
distributions of the outputs of the pre-processor and
run-processor, respectively). 
Assume that the capacity is not zero (otherwise, the claimed upper bound would be trivial).
In order to avoid introducing excessive notation for the various error terms involved,
in the sequel we use asymptotic notation as $n$ grows to infinity
(with hidden constants possibly depending on $\cD$); so that
a $o(1)$ term can be made arbitrarily small as $n$ grows; an $\omega(1)$ term grows
with $n$, and so forth. Consider a large enough $n$ (that we will
tend to infinity in the end) and a choice for $X$
that approaches the corresponding supremum, so that,
for the $Y$ and $Z$ defined by $X$, we have
\[
\capa(\cD) \leq \frac{I(Y;Z)}{n}+o(1).
\]
The minimum possible $n$ would depend on the desired
magnitude of the added $o(1)$ term.
We note that the length $M$ of $Y$ and $Z$ is itself a random variable
jointly correlated with $X$, $Y$ and $Z$. 
This causes technical difficulties that we first handle by  
showing below that we may essentially assume that the length $M$
is large, but fixed and known. In order to do so rigorously, first
recall that, denoting $Y_1^M := \DotsAZ{Y_1}{Y_M}$
and $Z_1^M := \DotsAZ{Z_1}{Z_M}$, we may write (since the knowledge of 
either $Y$ or $Z$ uniquely reveals $M$ as well),
\begin{align*}
I(Y;Z) &= I(M,Y_1^M;M,Z_1^M)\\
&= H(M,Y_1^M)-H(M,Y_1^M|M,Z_1^M)\\
&= H(M)+H(Y_1^M|M)-H(Y_1^M|M,Z_1^M)\\
&= H(M)+I(Y_1^M;Z_1^M|M)\\
&\leq \log n + I(Y_1^M;Z_1^M|M),
\end{align*}
where for the last inequality we are using the fact
that $M$ is always an integer between $1$ and $n$.
Therefore, conditioning on $M$ has no asymptotic effect on the capacity upper bound
and we may write
\[
\capa(\cD) \leq \frac{I(Y;Z|M)}{n}+o(1).
\]

Without loss of generality, in the sequel we assume that $X$ is entirely supported
on $n$-bit sequences that consist of $\Omega(n/\log n)$ runs
(much lower estimates would also suffice).
The contribution of all other sequences to the entropy of $X$ would be
$o(n)$, which would have no asymptotic effect on the achieved rate.
For any such input sequence (and consequently, for the distribution defined
by $X$), it is straightforward to show (e.g., using
Azuma-Hoeffding inequality) that with overwhelming probability
$1-1/n^{\omega(1)}$,
the resulting choice of $M$ will also be large; 
particularly, that $M \geq m_0$ for some $m_0=\Omega(n/\log n)$.
Let us now write
\begin{align}
\capa(\cD) &\leq \sum_{m \geq 1} \Pr[M=m] \frac{I(Y;Z|M=m)}{n}+o(1) \nonumber \\
&\leq \sum_{m \geq m_0} \Pr[M=m] \frac{I(Y;Z|M=m)}{n}+o(1), \label{eqn:capD:m0}
\end{align}
where in the second inequality, we have used the fact that
$M \geq m_0$ with probability $1-o(1)$, and have used the trivial
upper bound of $1$ for $I(Y;Z|M=m)/n$ when $m < m_0$ (recall that
$Z$ is always the run-length encoding of a bit-string of length at
most $n$, and thus its entropy is at most $n$).

Consider an alternative, but equivalent, realization of the pre-processor
that, given the input $X$, first draws an infinite sequence of i.i.d., geometrically
distributed random variables $G_1, G_2, \ldots$ (each with mean $1/p$),
and sets $G=G_\ell$ in the first step of the $\ell$th iteration (thus
the variables $G_{M+1},G_{M+2},\ldots$ are never looked at). 

Note that the total bit-length of $X$ consists of the summation of the produced
values of $Z_i$ by the pre-processor plus the corresponding $G_i$ (which represent 
the deleted bits by the pre-processor; i.e., hatched part in Figure~\ref{fig:processor}),
except for the final $G_M$ which may extend beyond the length of $X$.
More formally, it is always the case that
\[
\sum_{i=1}^M (Z_i + G_i) - G_M \leq n \leq \sum_{i=1}^M (Z_i + G_i),
\]
or, in other words,
\begin{equation} \label{eqn:capD:n}
n = \sum_{i=1}^M (Z_i + G_i) - \Delta,
\end{equation}
where $0 \leq \Delta \leq G_M$.

We recall that, for \emph{all} $m \geq 1$, we simultaneously have
$\E[\sum_{i=1}^m G_i]= m/p$. Furthermore,
by a Chernoff-Hoeffding inequality, the summation
highly concentrates around its expectation due to the
$G_i$ being independent; namely we may observe that
with probability $1-1/n^{\omega(1)}$, it is the
case that for all $m \geq m_0$
we have $\sum_{i=1}^m G_i = m(1/p + o(1))$.
Furthermore, the value of $G_M/m_0$ is $o(1)$
with probability $1-o(1)$ by Markov's inequality.
Overall, combined with \eqref{eqn:capD:n}, it follows that with probability $1-o(1)$,
we have
\begin{equation} \label{eqn:capD:n2}
\sum_{i=1}^M (Z_i + 1/p \pm o(1))=n.
\end{equation}
Note also that, the left hand side of \eqref{eqn:capD:n2} is $O(n)$ (treating $p$
as a constant) with probability $1$.
For an integer $m$, denote by $Z_{i,m}$ the random variable $Z_i$ 
conditioned on the event $M=m$. Given the input $X$, if we condition 
the output of the pre-processor on the event $M=m$, the joint
distribution of $G_1, G_2, \ldots$ obviously changes, to possibly 
even a non-product distribution. However, we may still apply an
averaging argument on \eqref{eqn:capD:n2} to show that\footnote{%
Some care is needed for the averaging argument. Particularly, we may take advantage of
the fact that, with high probability, the concentration bound
$\sum_{i=1}^m G_i = m(1/p+o(1))$ holds simultaneously \emph{for all} $m \geq m_0$.
Therefore, it just suffices to construct $S$ so that, for all $m \in S$, the
random variable $G_m$ conditioned on the event $M=m$ is upper bounded by $o(m)$,
which in turn follows by a simple averaging using Markov's inequality.
},
for some set $S \subseteq \N \setminus [m_0]$ such that $\Pr[M \in S] = 1-o(1)$,
the following holds: For all $m \in S$, we have
\begin{equation*} 
\sum_{i=1}^m (Z_{i,m} + 1/p \pm o(1))=n,
\end{equation*}
with probability $1-o(1)$ over the distribution of
$Z_{1,m}^m := (Z_{1,m},Z_{2,m}, \ldots, Z_{m,m})$. This, in turn, implies that,
for all $m \in S$,
\begin{equation} \label{eqn:capD:S}
\sum_{i=1}^m (\E[Z_{i,m}] + 1/p \pm o(1))=n(1+o(1)).
\end{equation}
Similar to $Z_{1,m}^m$, define $Y_{1,m}^m := (Y_{1,m},\ldots, Y_{m,m})$,
where $Y_{i,m}$ is the random variable $Y_i$ conditioned on the
event $M=m$. In other words, $Y_{1,m}^m$ is the output of the
run-processor when given $Z_{1,m}^m$ at input.
We may now rewrite \eqref{eqn:capD:m0} as
\begin{align}
\capa(\cD) &\leq 
\sum_{m \in S} \Pr[M=m] \frac{I(Y;Z|M=m)}{n}+o(1) \nonumber \\
&\stackrel{\eqref{eqn:capD:S}}{\leq} 
\sum_{m \in S} \Pr[M=m] \frac{I(Y_{1,m}^m;Z_{1,m}^m) (1+o(1))}{\sum_{i=1}^m (\E[Z_{i,m}] + 1/p \pm o(1))}+o(1).
 \label{eqn:capD:m0b}
\end{align}
Consider any fixed $m$. 
Note that the effect of the run-processor on $Z_{1,m}^m$ is 
precisely the same as an $m$-use mean-limited channel, 
as defined in Section~\ref{sec:meanlimited:multiple} (albeit 
without the mean constraint), where
the transition rule $\cP$ is given by the conditional distribution
$Y_{i,m}^m|Z_{i,m}^m$. 
That is, given an integer input $1+x$, the transition rule
$\cP$ outputs a sample from 
$\overline{\cD}\oplus \cD^{\oplus x}$.
Therefore, as in the proof of Lemma~\ref{lem:product}, since
each $Y_{i,m}$ only depends on the corresponding
random variable $Z_{i,m}$, we may write
\begin{equation} \label{eqn:capD:product}
I(Y_{1,m}^m;Z_{1,m}^m) \leq \sum_{i=1}^m I(Y_{i,m};Z_{i,m}).
\end{equation}
Furthermore, 
recall $\lam := \E[\cD]$ and $\overline{\lam} := \E[\overline{\cD}]$,
and observe that, for all $i$,
\begin{equation} \label{eqn:capD:EYZ}
\E[Y_{m,i}]=\overline{\lam}+(\E[Z_{m,i}]-1) \lam.
\end{equation}
Using \eqref{eqn:capD:product} and \eqref{eqn:capD:EYZ}, we may now rewrite 
 \eqref{eqn:capD:m0b} as
 \begin{align}
\capa(\cD) &\leq  
\sum_{m \in S} \Pr[M=m] \frac{\sum_{i=1}^m I(Y_{i,m};Z_{i,m}) (1+o(1))}
{\sum_{i=1}^m (1/p+(\E[Y_i]-\overline{\lam})/\lam \pm o(1))}+o(1)\nonumber \\
&\leq  
\sum_{m \in S} \Pr[M=m] \max_{i \in [m]}\frac{I(Y_{i,m};Z_{i,m}) (1+o(1))}
{1/p+(\E[Y_i]-\overline{\lam})/\lam \pm o(1)}+o(1) 
 \label{eqn:capD:m0c}
\end{align}
where the second inequality is due to the following simple result:
\begin{prop}
For positive real numbers $a_1, \ldots, a_m$ and
$b_1, \ldots, b_m$, we have
\[
\frac{a_1 +\cdots+a_m}{b_1 +\cdots+b_m} \leq \max_{i=1,\ldots,m} \frac{a_i}{b_i}.
\]
\end{prop}

\begin{proof}
Without loss of generality suppose the right hand side is $a_1/b_1$. Then,
the inequality is equivalent to 
\[
\sum_{i=1}^m b_1 a_i \leq \sum_{i=1}^m a_1 b_i, 
\]
which is true since for each $i$, we have assumed
$a_i b_1 \leq b_i a_1$.
\end{proof}
Now, observe that for any $i$ and $m$, we have
\[
\frac{I(Y_{i,m};Z_{i,m})}
{1/p+(\E[Y_i]-\overline{\lam})/\lam \pm o(1)}
\leq \sup_{\mu\geq \overline{\lam}} \frac{\capa(\ch_\mu(\cP))}
{1/p+(\mu-\overline{\lam})/\lam \pm o(1)},
\]
since, assuming that $\mu = \E[Y_{i,m}]$, the random variable
$Y_{i,m}$ is sampled by transmitting $Z_{i,m}$ over
a mean-limited channel with transition rule $\cP$ and 
mean constraint $\mu$. Therefore, the mutual information
$I(Y_{i,m};Z_{i,m})$ would be no more than the capacity of
this channel. Using this, \eqref{eqn:capD:m0c}
further simplifies to 
 \begin{align*}
\capa(\cD) &\leq  
\sum_{m \in S} \Pr[M=m] \max_{i \in [m]}
\sup_{\mu\geq \overline{\lam}} \frac{\capa(\ch_\mu(\cP))(1+o(1))}
{1/p+(\mu-\overline{\lam})/\lam \pm o(1)} + o(1)\\
&\leq \sup_{\mu\geq \overline{\lam}} \frac{\capa(\ch_\mu(\cP)) (1+o(1))}
{1/p+(\mu-\overline{\lam})/\lam \pm o(1)} + o(1) \\
&= \sup_{\mu\geq \overline{\lam}} \frac{\capa(\ch_\mu(\cP))}
{1/p+(\mu-\overline{\lam})/\lam},
\end{align*}
where the last equality is attained from the fact that,
by taking the limit $n \to \infty$, the $o(1)$ terms vanish.
This completes the proof of Theorem~\ref{thm:Cupper}.



\section{Upper bounds on the capacity of the Poisson-repeat channel}
\label{sec:Poi}

\subsection{Upper bounds on the capacity of a mean-limited Poisson channel}
\label{sec:Poi:PoiMeanLim}

Let $\cD$ be a Poisson distribution with mean $\lam$, and 
$\ch := \ch_\mu(\cD)$ be the convolution channel defined with respect
to the distribution\footnote{We note that, in 
the context of optical communications, 
this channel was also considered in \cite{ref:AW12}.
A related channel is the standard, additive, discrete-time Poisson channel \cite{ref:Sha90}
that has been extensively studied in information theory. 
} 
$\cD$ and mean constraint $\mu$. 
Let $\cP$ be the probability 
transition rule corresponding to $\ch_\mu(\cD)$
when seen as a general mean-limited channel.
Recall that the input and output alphabets for this channel
are both the set of non-negative integers.
By Theorem~\ref{thm:meanlimited}, in order to upper
bound the capacity of $\ch$, it suffices to exhibit
a distribution over non-negative integers
and real parameters $\nu_0$ and $\nu_1$, so that
the corresponding random variable $Y \in \N^{\geq 0}$ drawn from  
this distribution satisfies \eqref{eqn:cvx1:dualF}.

Let $Y_x$ be the output of the channel when the input is fixed
to $x$. Explicitly, $Y_x$ has a Poisson distribution with
mean $\E[Y_x]=\lam x$, so that \eqref{eqn:cvx1:dualF} can be
rewritten as
\begin{equation} \label{eqn:cvx1:dualFPoi} 
\KL{Y_x}{Y} \leq 
\lam \nu_1 x +\nu_0, \qquad x=0,1,\ldots.
\end{equation} 
By the conclusion of Theorem~\ref{thm:meanlimited}, 
exhibiting any such distribution $Y$ and parameters
$\nu_0$ and $\nu_1$ would imply
\begin{equation}
\capa(\ch_\mu(\cD)) \leq \nu_1 \mu + \nu_0.
\end{equation}
We consider the following general form for the distribution
of $Y$:
\begin{equation} \label{eqn:PoiY}
\Pr[Y=y]=y_0 \exp(f(y)) (q/e)^y,\quad y=0,1,\ldots,
\end{equation}
for some function $f\colon \N^{\geq 0} \to \R$,
real parameter $q>0$, and normalizing constant
\[
y_0 = \left(\sum_{y=0}^\infty \exp(f(y)) (q/e)^y\right)^{-1},
\] 
assuming that the summation is convergent. Given any
function $f$ that grows linearly in $y$ or slower,
it is always possible to choose $q$ small enough so
that the distribution is well defined. Moreover, by
varying the choice of $q$ it is possible to set the
expectation of $Y$ to match the chosen parameter\footnote{%
By varying $q$, the mean $\mu$ may be adjusted to any arbitrary positive 
value so long as,
for some fixed $q_0>0$, the summation defining $y_0$
diverges to infinity with $q=q_0$ but, on the other hand, converges for all $q<q_0$.
} $\mu$.
%
It turns out to be more convenient to set
$f(y):=g(y)-\log y!$, for some function $g\colon \N^{\geq 0} \to \R$,
and our goal would be to obtain an appropriate choice for $g$.

Recall that, for any choice of positive integers $x$ and $y$,
\[
\Pr[Y_x = y] = \frac{e^{-\lam x} (\lam x)^y}{y!}.
\]
The KL divergence $\KL{Y_x}{Y}$ can now be written as
\begin{align} \nonumber
\KL{Y_x}{Y} &= \sum_{y=0}^\infty \Pr[Y_x = y] \log \frac{\Pr[Y_x = y]}{\Pr[Y=y]}\\
&\stackrel{\eqref{eqn:PoiY}}{=} 
-\log y_0-\lam x (\log q) + \lam x \log(\lam x) - \E[g(Y_x)].
\label{eqn:KLPoiGeneral}
\end{align}
Note that the only nonlinear term (in $x$) in the above is
$\lam x \log(\lam x) - \E[g(Y_x)]$, so achieving
\eqref{eqn:cvx1:dualFPoi} is equivalent to having
real coefficients $a, b$ such that
\begin{equation} \label{eqn:cvx1:dualFPoiab} 
\lam x \log(\lam x) - \E[g(Y_x)] \leq 
ax+b, \qquad x=0,1,\ldots.
\end{equation} 
At this point, the following feasible choice 
$g(y)=g_0(y)$ is immediate:
\[
g_0(y)=\begin{cases}
0 & \text{if $y=0$,} \\
y \log y & \text{if $y>0$,}
\end{cases}
\]
which results in 
\begin{equation} \label{eqn:PoiYfirst}
\Pr[Y=y] = y_0 \frac{y ^y}{y!} (q/e)^y,
\end{equation}
where $0^0$ is to be understood as $1$.
Numerical estimates on the mean and normalizing constants
of this distribution for various choices of $q$ are listed in
Table~\ref{tab:PoiMeanY}.
To see that this choice satisfies 
\eqref{eqn:cvx1:dualFPoiab}, it suffices to
note that the function $g(y)$ defined above
is convex, and thus, by Jensen's inequality,
\[
\E[g_0(Y_x)] \geq g_0(E[Y_x]) = g_0(\lam x)=\lam x \log(\lam x),
\] 
so \eqref{eqn:cvx1:dualFPoiab} is satisfied for $a=b=0$.
One can, however, observe that the inequality
is strict by a constant gap as $x$ grows (as we show in Section~\ref{sec:Poi:truncation}). 

Using Stirling's approximation $y! \sim \sqrt{2\pi y} (y/e)^y$,
we may write the asymptotic behavior of \eqref{eqn:PoiYfirst} as
\begin{equation} \label{eqn:Poi:StirApprox}
\Pr[Y=y] \sim y_0 \frac{y^y (q/e)^y}{\sqrt{2\pi y} (y/e)^y}
= \frac{y_0}{\sqrt{2\pi y}}\, q^y,
\end{equation}
so we see that \eqref{eqn:PoiYfirst} can be normalized to a valid
distribution \Iff\ $q<1$.

We now present
a different choice for $g$ that more closely estimates
the linear upper bound $ax+b$, and in particular, converges to
it as $x$ grows. This alternative choice results
in a better capacity upper bound than the immediate choice
above. It is obtained by replacing $\log y$ in $g_0(y)$
with harmonic numbers that asymptotically behave
like $\log y$ 
but provide a more refined result. 
Explicitly, consider
\begin{equation} \label{eqn:gPoi}
g(y)=\begin{cases}
0 & \text{if $y=0$,} \\
y \psi(y)=y (H_{y-1} -\gamma) & \text{if $y>0$,}
\end{cases}
\end{equation}
where $\psi(y)=\Gamma'(y)/\Gamma(y)$ is the digamma
function, $H_n=\psi(n+1)+\gamma$ denotes the $n$th harmonic number
(where $H_0=0$, and, for a positive integer $n$, $H_n=\sum_{k=1}^n 1/k$),
 and $\gamma \approx 0.57721$ is the Euler-Mascheroni constant.
It is known that \cite[p.~259]{ref:Ab72}
\[
\psi(y) = \log y - \frac{1}{2y} +O\left(\frac{1}{y^2} \right),
\] 
so we have $\lim_{y\to \infty} (y \psi(y)-y \log y)=-1/2$,
and thus, combined with the Stirling approximation, we see that
with this alternative choice of $g$, we have a slightly different
asymptotic behavior than \eqref{eqn:Poi:StirApprox}, namely,
\begin{equation} \label{eqn:Poi:StirApprox:alternative}
\Pr[Y=y] \sim y_0 \frac{y^y (q/e)^y}{\sqrt{2\pi e y} (y/e)^y}
= \frac{y_0}{\sqrt{2\pi e y}}\, q^y,
\end{equation}
%
To verify that this choice of $g$ is feasible, we first
prove a series expansion for the function $g$.

\begin{lem} \label{lem:newton}
For any $y \geq 0$, the function $g$ in \eqref{eqn:gPoi}
can be represented as
\[
g(y)=-\gamma y + \sum_{j=2}^\infty \frac{(-1)^j j}{j-1} \binom{y}{j}.
\]
\end{lem}

\begin{proof}
The proof is similar to the derivation for the well-known Newton series
expansion of the digamma function \eqref{eqn:digamma:newton}. The function $g$ exhibits
a discontinuity at $y=0$, where the limit is $-1$. Let us
correct this discontinuity by defining a function $g_1$ which
is the same as $g$ except at point $y=0$, where we define $g_1(0)=-1$.
The function $g_1$ is continuous and well-defined at all $y \geq 0$,
which we now express as a Newton series expansion.
Recall that the Newton series expansion of (any function) $g_1$ 
around zero can be written as
\begin{equation} \label{eqn:newton}
g_1(y)=\sum_{j=0}^\infty c_j \binom{y}{j},
\end{equation}
where the coefficient $c_j$ is defined to be the $j$th 
forward difference of the function at zero, namely,
\begin{equation} \label{eqn:newtonCoeff}
c_j = \sum_{k=0}^j (-1)^{j-k} \binom{j}{k} g_1(k).
\end{equation}
For our particular choice of $g_1$,
the forward difference at point $y$ is, 
understanding $y H_{y-1}$ at $y=0$ by its limit $-1$,
\[
\Delta[g_1](y):=g_1(y+1)-g_1(y)=-\gamma+(y+1) H_y-yH_{y-1}=
-\gamma+y (H_{y}-H_{y-1})+H_y=1-\gamma+H_y.
\]
Taking the second forward difference, we then obtain
\[
\Delta^{(2)}[g_1](y)=\Delta[g_1](y+1)-\Delta[g_1](y)=
H_{y+1}-H_y=\frac{1}{y+1}.
\]
Thus, for any $j\geq 2$, the $j$th forward difference of
the function $g_1$ is the $(j-2)$nd forward difference of
the function $\frac{1}{y+1}$, which can in turn be written as
\[
\Delta^{(j)}[g_1](y)=\sum_{k=0}^{j-2} (-1)^{j-k} \binom{j-2}{k} \frac{1}{y+k+1}.
\]
One can verify by induction on $j$ that the right hand side is equal to
\[
(-1)^{j} (j-2)! \prod_{k=0}^{j-2} \frac{1}{y+k+1},
\]
which, at $y=0$ and for $j\geq 2$, 
simplifies to $(-1)^j/(j-1)$ and gives the value of $c_j$.
Plugging this result in \eqref{eqn:newton}, we conclude that
\[
g_1(y)=-1+(1-\gamma)y+\sum_{j=2}^\infty \frac{(-1)^j}{j-1} \binom{y}{j}.
\]
Since the functions $g$ and $g_1$ only differ at $y=0$, 
from \eqref{eqn:newtonCoeff} we see that the $j$th Newton series
coefficients for the function $g$ is $(-1)^j+c_j$, which, for $j \geq 2$,
is equal to $(-1)^j j/(j-1)$. 
The function can now be expanded as
\[
g(y)=-\gamma y + \sum_{j=2}^\infty \frac{(-1)^j j}{j-1} \binom{y}{j},
\]
as desired.
\end{proof}

As a corollary of the above lemma, we may derive the following.

\begin{coro} \label{coro:EgY}
Let $g$ be the function in \eqref{eqn:gPoi},
and $Y$ be a Poisson random variable with mean $\lam$. Then,
\[
\E[g(Y)]=\lam (E_1(\lam) +\log \lam), 
\]
where
\begin{equation} \label{eqn:Ei:def}
E_1(\lam)=\int_{1}^\infty \frac{e^{-\lam t}}{t}\, dt
\end{equation}
is the exponential integral function.
\end{coro}

\begin{proof}
The main ingredient to use is the simple fact that the $j$th
factorial moment of $Y$ is given by
\[
\E\left[ j! \binom{Y}{j} \right] = \lam^j.
\]
Combining this with the result of Lemma~\ref{lem:newton} immediately
gives
\begin{equation} \label{eqn:EgY}
\E[g(Y)]= -\gamma \lam+\sum_{j=2}^\infty \frac{(-\lam)^j j}{(j-1)j!}
= -\gamma \lam+\sum_{j=1}^\infty \frac{(-\lam)^{j+1}}{j j!}
\end{equation}
We now recall the following basic power series expansion for the
exponential integral function \cite[p.~229]{ref:Ab72}: For all $x > 0$,
\[
E_1(x) = -\gamma -\log x -\sum_{j=1}^\infty \frac{(-x)^j}{j j!},
\]
where $\gamma$ is the Euler-Mascheroni constant.
Using this expansion in \eqref{eqn:EgY} yields
\[
\lam E_1(\lam) + \gamma \lam = -\lam \log \lam + \E[g(Y)] + \gamma \lam, 
\]
which completes the proof.
\end{proof}

The result of Corollary~\ref{coro:EgY} immediately implies that
the choice of $g$ in \eqref{eqn:gPoi} satisfies 
\eqref{eqn:cvx1:dualFPoiab} with $a=b=0$, as we have
\begin{equation} \label{eqn:EgYx}
\lam x \log(\lam x)-\E[g(Y_x)]=-\lam x E_1(\lam x) \leq 0,
\end{equation}
from the fact that the exponential integral function $E_1(x)$ is,
by its integral definition, positive for all $x>0$.
The value of $\lam x E_1(\lam x)$ exponentially decays down to zero
(see Figure~\ref{fig:Ei}),
and therefore, \eqref{eqn:cvx1:dualFPoiab} sharply holds
for the choice of $g$ in \eqref{eqn:gPoi}.
\begin{figure}[t!]
\begin{center}
\includegraphics[height=2in]{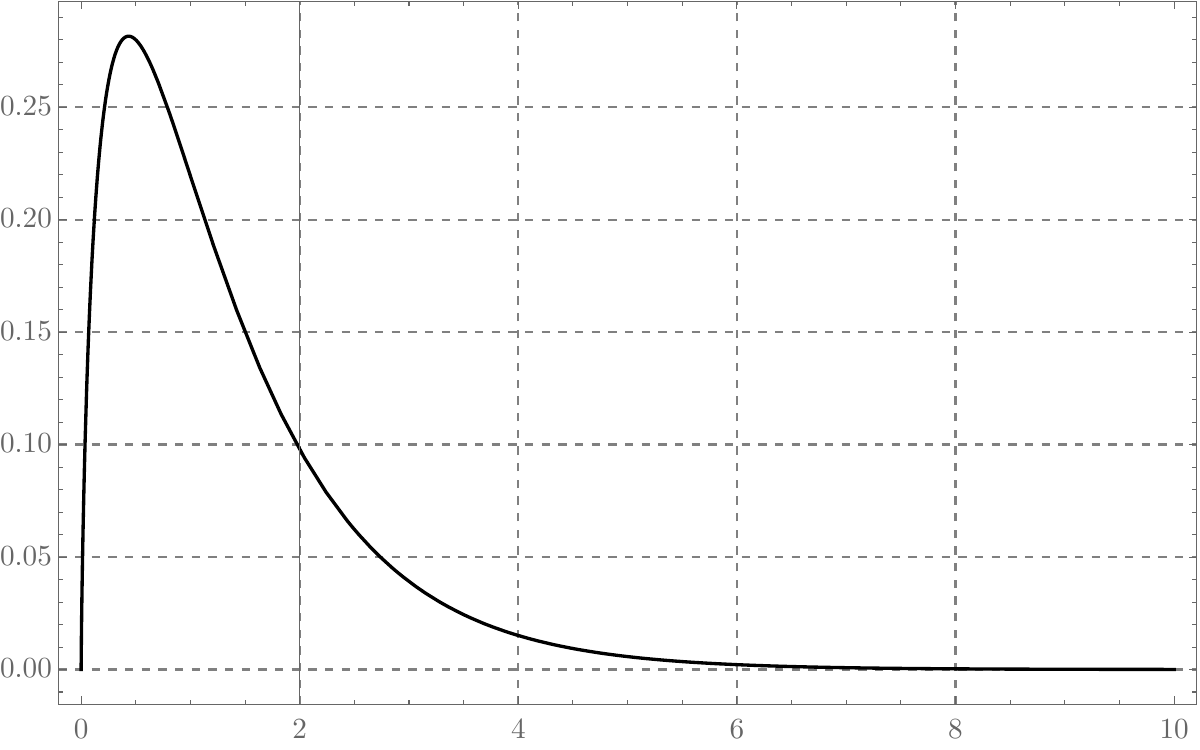}
\end{center}
\caption{The function $x E_1(x)$.}
\label{fig:Ei}
\end{figure}
The resulting distribution $Y$ can now be rewritten,
from \eqref{eqn:PoiY}, as
\begin{equation} \label{eqn:PoiYre}
\Pr[Y=y]=\begin{cases}
y_0 & \text{if $y=0$,} \\
y_0\, {\exp(y\psi(y)) (q/e)^y}/{y!} & \text{if $y>0$.}
\end{cases}
\end{equation}
We call the distribution defined by \eqref{eqn:PoiYre} the
\emph{digamma distribution} due to the digamma term in the exponent.
Combining \eqref{eqn:KLPoiGeneral} and 
\eqref{eqn:EgYx} gives us
\begin{equation} \label{eqn:PoiKL}
\KL{Y_x}{Y}\leq -\log y_0-\lam x \log q= -\log y_0- \E[Y_x] \log q, 
\end{equation}
and thus, Theorem~\ref{thm:meanlimited} gives the upper bound
\begin{equation}
\capa(\ch_\mu(\cD)) \leq -\mu \log q-\log y_0,
\end{equation}
where $q$ (and accordingly, the normalizing constant $y_0$)
must be chosen so that the mean constraint $\E[Y]=\mu$ is
satisfied (note that this choice is unique for any given $\mu>0$).
Therefore, understanding the upper bound requires a characterization 
of the relationship between $\mu, q$, and $y_0$.
Since the probability mass function defining $Y$ exhibits
an exponential decay, the values $\mu$ and $y_0$, as a function of
$q$, can be numerically computed efficiently to any
desired accuracy. Plots of these functions, for both distributions
\eqref{eqn:PoiYfirst} and \eqref{eqn:PoiYre}, are depicted
in Figure~\ref{fig:PoiMu}. Moreover, their numerical estimates
are listed, for various choices of $q \in (0,1)$, in Table~\ref{tab:PoiMeanY}.
\begin{figure}[t!]
\begin{center}
\includegraphics[height=2in]{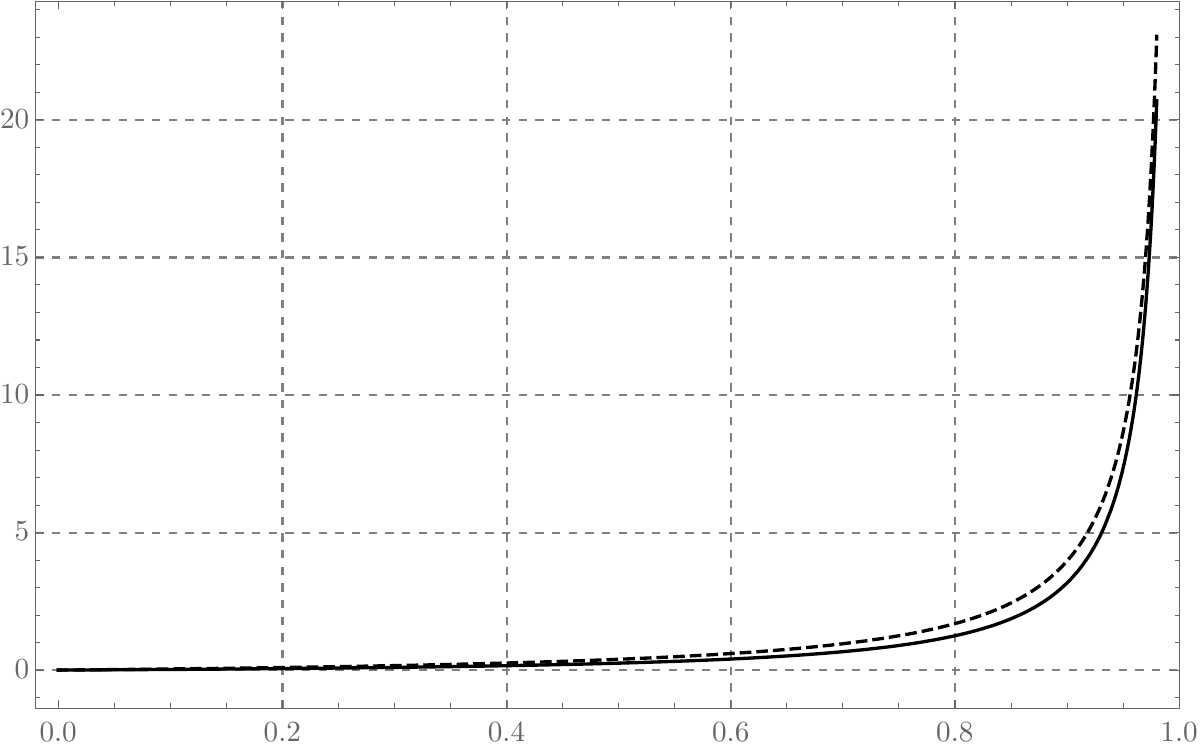}
\includegraphics[height=2in]{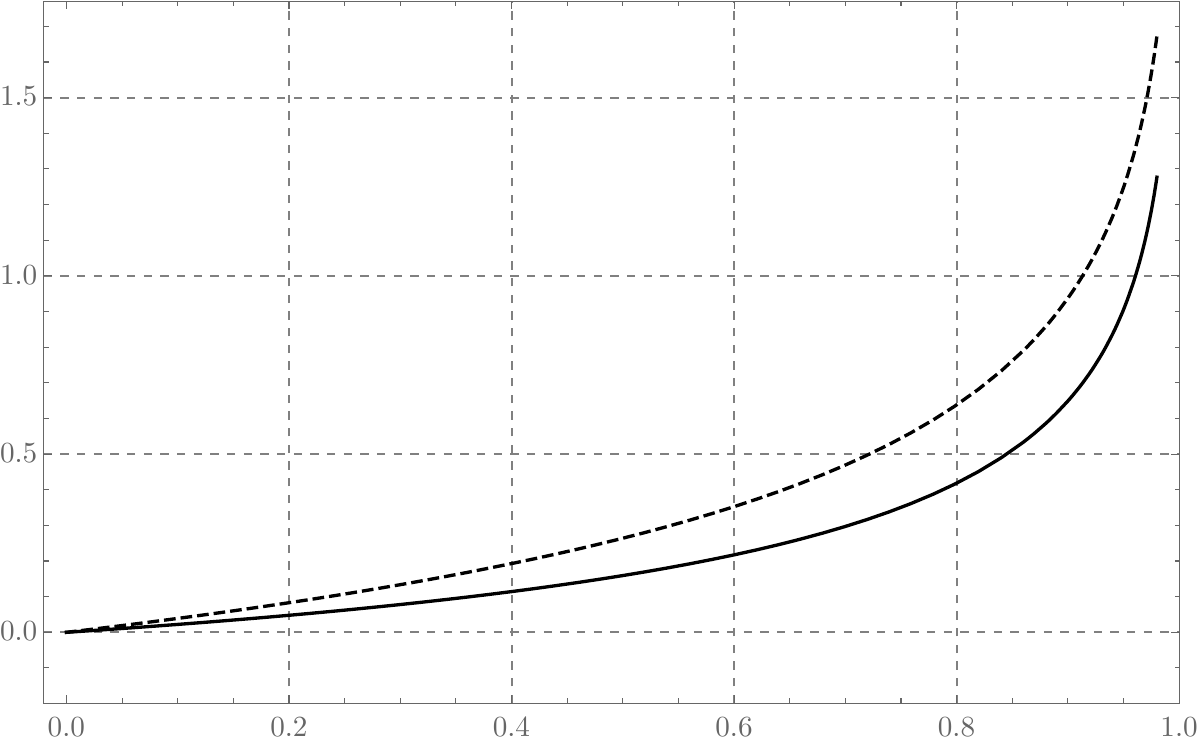}
\end{center}
\caption{Plots of $\mu=\E[Y]$ (left) and $\log(1/y_0)$ (right),
for the distribution of $Y$ 
in \eqref{eqn:PoiYfirst} (dashed) and the digamma distribution
\eqref{eqn:PoiYre} (solid), as
functions of $q$.}
\label{fig:PoiMu}
\end{figure}
We summarize the result of this section in the following\footnote{%
An appealing aspect of assigning the mean constraint to the output,
rather than the input, distribution is that such results as 
Theorem~\ref{thm:PoiChUpper} become independent of the
channel parameter $\lam$. Therefore, when we apply this result to
obtain capacity upper bounds for the Poisson-repeat channel, 
the deletion probability $p=1-\Exp{-\lam}$ appears only in the final expression 
to be optimized  \eqref{eqn:Poi:capa:general}.
}:
\begin{thm} \label{thm:PoiChUpper}
Let $q\in(0,1)$ be a given parameter and $Y$ be a random
variable distributed according to the digamma distribution \eqref{eqn:PoiYre},
for an appropriate normalizing constant $y_0$. Let
$\mu := \E[Y]$, and $\cD$ denote any Poisson distribution with positive mean.
Then, capacity of the mean-limited Poisson channel $\ch_\mu(\cD)$
satisfies 
\begin{equation} \label{eqn:PoiMeanLimCapUpperA}
\capa(\ch_{\mu}(\cD)) \leq -\mu \log q-\log y_0.
\end{equation}
\end{thm} 

Figure~\ref{fig:PoiMeanLimCap} depicts the capacity upper bounds attained by the
above result, as well as a similar result when the dual-feasible distribution for $Y$
is defined by \eqref{eqn:PoiYfirst}.

\begin{figure}[t!]
\begin{center}
\includegraphics[height=2in]{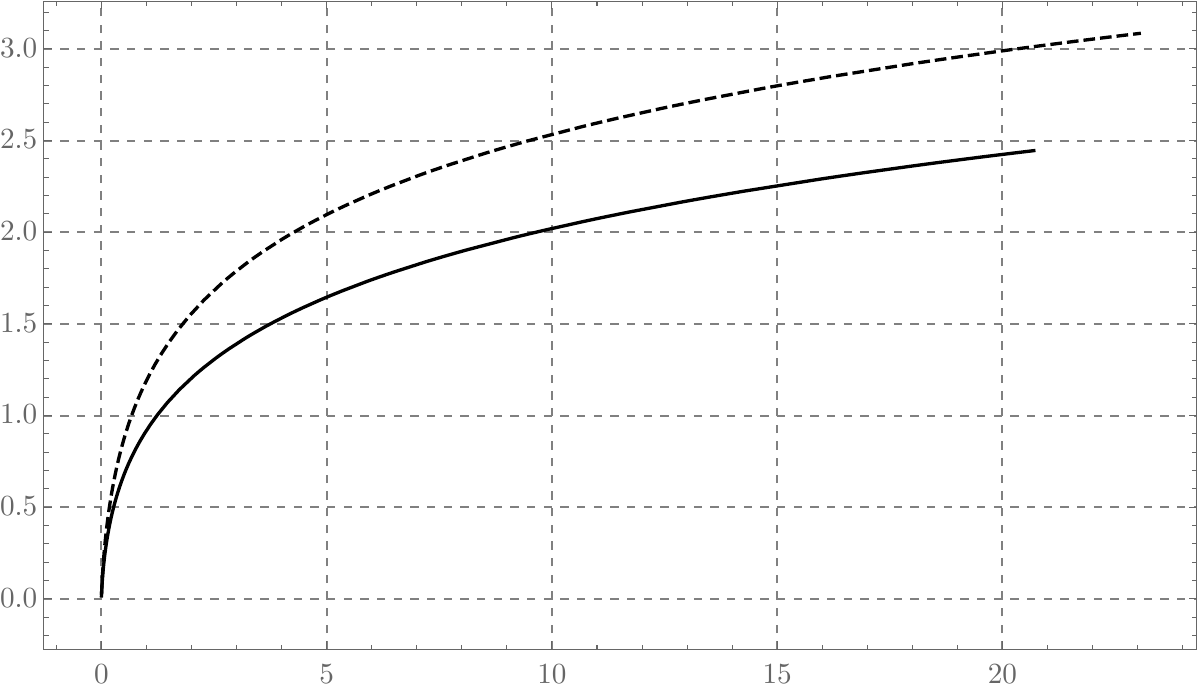}
\end{center}
\caption{Capacity upper bounds (measured in bits) 
of Theorems \ref{thm:PoiChUpper} 
and \ref{thm:PoiChUpperAnalytic} 
for the mean-limited Poisson channel in terms of $\mu=\E[Y]$.
The solid plot is the upper bound given by \eqref{eqn:PoiMeanLimCapUpperA}
(when the digamma distribution \eqref{eqn:PoiYre} is used for $Y$).
The dashed plot is given by \eqref{eqn:PoiMeanLimCapUpperB} 
(when \eqref{eqn:PoiYfirst} is used for $Y$).
The first
inequality in \eqref{eqn:PoiMeanLimCapUpperB} and its analytic upper bound
estimate would completely overlap and be indistinguishable in this plot.
Likewise, the upper bound estimates of Corollary~\ref{coro:Poi:negBinApprox}
for the digamma distribution \eqref{eqn:PoiYre} would result in essentially the same plot as the exact one above.
}
\label{fig:PoiMeanLimCap}
\end{figure}

\subsection{The truncation effect of replacing logarithm with harmonic numbers}
\label{sec:Poi:truncation}

The aim of this section is to provide an intuitive explanation of why the choice of
the digamma distribution \eqref{eqn:PoiYre} for the distribution of $Y$, that essentially replaces the logarithmic term $\log y$
in the exponent of $\exp(y \log y)$ in \eqref{eqn:PoiYfirst} with harmonic
numbers (equivalently, $\psi(y)$), results in improved capacity upper bounds.

For any analytic choice of $g(y)$ in  \eqref{eqn:cvx1:dualFPoiab}, we can write down 
the Taylor series expansion of $g$ around $\mu$ as
\[
g(y) = g(\mu)+(y-\mu) g'(\mu)+\frac{1}{2}(y-\mu)^2 g''(\mu)+\sum_{j=3}^\infty (y-\mu)^j \frac{g^{(j)}(\mu)}{j!}.
\]
Let $Y_x$ be Poisson-distributed with mean $\lam x$. 
Assuming that the above series converges\footnote{Convergence issues may be disregarded for large values of $y$
that have negligible contribution to the probability mass function of $Y_x$.}, we may take
the expectation of the above and, and noting
that the variance of $Y_x$ is equal to $\lam x$, and letting $\mu := \lam x$, write
\[
\E[g(Y_x)] = g(\lam x)+\frac{1}{2} \lam x g''(\lam x)+\sum_{j=3}^\infty \mu_j \frac{g^{(j)}(\lam x)}{j!},
\]
where $\mu_j := \E[(Y_x-\mu)^j]$. Now, with $g(y)=y \log y$ as in \eqref{eqn:PoiYfirst},
we would have $g''(\lam x) = 1/(\lam x)$, and for $j \geq 3$, $g^{(j)}(\lam x)=O(1/x^{j-1})$ so that,
for large $x$, we have the asymptotic behavior
\[
\E[g(Y_x)] = g(\lam x)+\frac{1}{2}+o(1).
\]
\begin{figure}[t!]
\begin{center}
\includegraphics[height=2in]{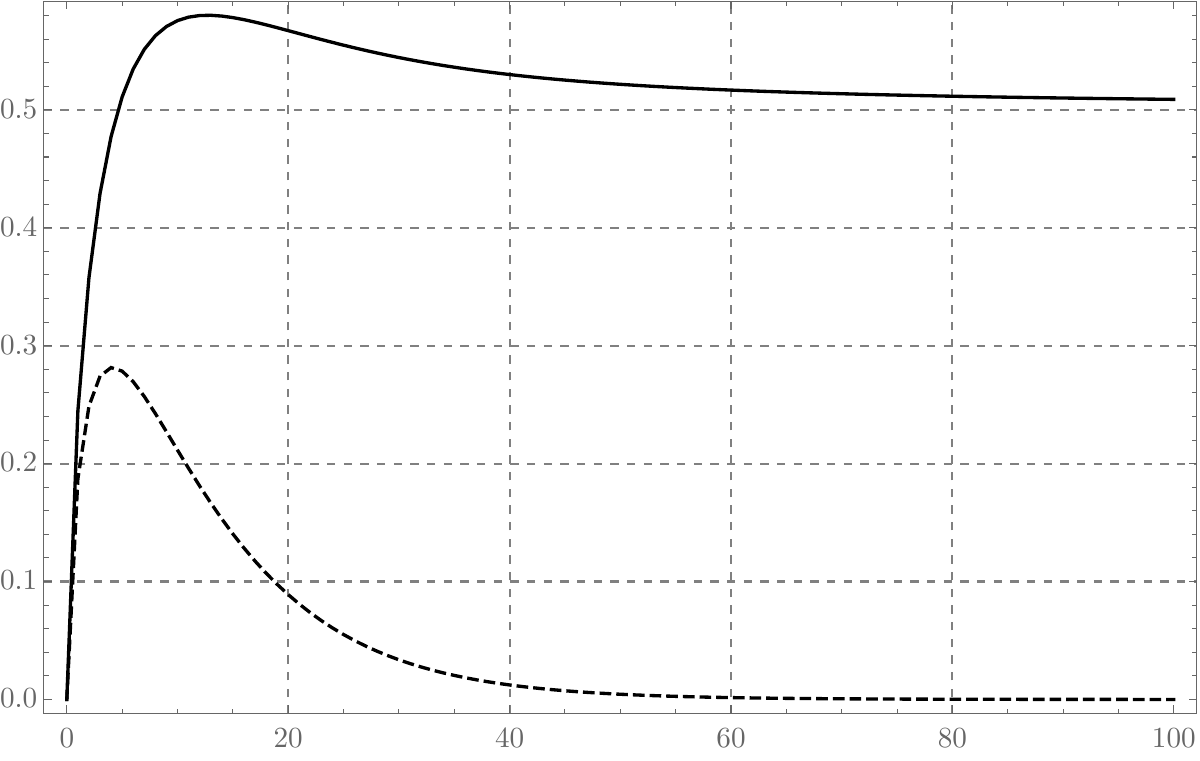}
\end{center}
\caption{Plot of the KKT gap to equality in \eqref{eqn:cvx1:dualFPoiab},
as a function of $x$, attained by the distribution \eqref{eqn:PoiYfirst} (solid) 
and the digamma distribution \eqref{eqn:PoiYre} (dashed) for deletion probability $d=0.9$ (the plots would
simply be scaled in $x$ for other deletion parameters). 
By Corollary~\ref{coro:EgY}, the second plot
coincides with the function $\lam x E_1(\lam x)$, where $\lam=-\log d$.
}
\label{fig:PoiKLGap}
\end{figure}
Therefore, while $g(y)$ satisfies \eqref{eqn:cvx1:dualFPoiab} with $a=b=0$ for all $x$
(and with equality for $x=0$),
the inequality exhibits an asymptotic constant gap of $1/2$ for large $x$
(see Figure~\ref{fig:PoiKLGap} for a depiction).
As we saw in Section~\ref{sec:Poi:PoiMeanLim}, specifically \eqref{eqn:EgYx}, this gap is eliminated by choosing
$g(y)$  according to \eqref{eqn:gPoi}. While this fact is verified in \eqref{eqn:EgYx},
it is worthwhile to provide a systematic way of deriving a choice of $g$ that exhibits no asymptotic gap.
In order to do so, recall that the ``ultimate goal'' in satisfying the KKT conditions of Theorem~\ref{thm:meanlimited} would be to exhibit a function $g(y)$ for which
 \eqref{eqn:EgYx} is satisfied with equality for all $x \geq 0$.
In general, this may be impossible to achieve with a choice of $g$ that 
does not grow faster than $y \log y+O(y)$ (so that the resulting expression
 \eqref{eqn:PoiY} can be normalized to a valid probability distribution).
 Nevertheless, we present a ``truncation technique'' that obtains an 
 approximate guarantee, 
 in the sense that the gap  in \eqref{eqn:cvx1:dualFPoiab}
 exponentially decays as $x$ grows,  while maintaining a controllable choice for $g$.

Assuming that  \eqref{eqn:cvx1:dualFPoiab}  holds with equality and $a=b=0$, 
we must have
$\E[g(Y_x)]=\mu \log \mu$, where $\mu :=\lam x$.
The right hand side, using the integral expression
\[
\log \mu = \int_{0}^\infty \frac{e^{-t}-e^{-\mu t}}{t}\, dt,
\]
and the Taylor expansion of the exponential function, can be written as 
\begin{align}
\mu \log \mu &= \int_{0}^\infty  \left( \frac{\mu e^{-t}}{t} + \sum_{j=1}^\infty \frac{\mu^j (-t)^{j-2}}{(j-1)!}\right)dt \label{eqn:Poi:uLogu:a} \\
&= 
\int_{0}^1  \left( \frac{\mu e^{-t}}{t} + \sum_{j=1}^\infty \frac{\mu^j (-t)^{j-2}}{(j-1)!}\right)dt +
\mu \int_{1}^\infty \frac{e^{-t}-e^{-\mu t}}{t}\, dt \nonumber \\
&= 
\int_{0}^1  \left( \frac{\mu e^{-t}}{t} + \sum_{j=1}^\infty \frac{\mu^j (-t)^{j-2}}{(j-1)!}\right)dt +
\mu E_1(1)-\mu E_1(\mu), \label{eqn:Poi:uLogu:b}
\end{align}
where $E_1(\cdot)$ denotes the exponential integral function \eqref{eqn:Ei:def}, and $E_1(1) \approx 0.219384$.
Using \eqref{eqn:Poi:uLogu:a}, and noting that the factorial moments of the Poisson distribution are given by
$\E\left[\binom{Y_x}{j}\right]=\mu^j/j!$, the following function
\begin{equation} \label{eqn:Poi:tilG}
\tilde{g}(y) := \int_{0}^\infty \left( \frac{y e^{-t}}{t} + \sum_{j=1}^\infty j \binom{y}{j} (-t)^{j-2}\right)dt
= y \int_{0}^\infty \left( \frac{ e^{-t}-1}{t} + \sum_{j=1}^\infty \binom{y-1}{j} (-t)^{j-1}\right)dt,
\end{equation}
interpreted formally, is the unique solution to the functional equation
\begin{equation} \label{eqn:Poi:functional}
(\forall x \geq 0)\ \E[\tilde{g}(y)]=\mu \log u.
\end{equation}
However, the above integral definition of $\tilde{g}(y)$ does not converge (recall that
$\int_0^\infty \Exp{-t} dt/t$ is divergent and that the inner summation in \eqref{eqn:Poi:tilG} only has a finite number
of terms for any integer $y > 0$).  To address this issue, we write down a function 
whose expectation sharply approximates the desired value \eqref{eqn:Poi:uLogu:b}
for large $\mu$. In order to do so, it suffices to note that the term $ \mu E_1(\mu)$
in \eqref{eqn:Poi:uLogu:b} is exponentially small in $\mu$ and can thus be ignored in the
approximation. Now, consider
a truncated variation of $\tilde{g}$ defined, by simply truncating the upper limit of the 
integration at $t=1$, as
\begin{align} \label{eqn:Poi:tilGtr}
\hat{g}(y) &:= 
y \int_{0}^1 \left( \frac{ e^{-t}-1}{t} + \sum_{j=1}^\infty \binom{y-1}{j} (-t)^{j-1}\right)dt  \nonumber \\
&= y(-\gamma-E_1(1))   + y \int_{0}^1 \sum_{j=1}^\infty \binom{y-1}{j} (-t)^{j-1}dt \nonumber  \\
&= y(-\gamma-E_1(1))   - y \sum_{j=1}^\infty \binom{y-1}{j} \frac{(-1)^{j}}{j} \nonumber   \\
&= -y E_1(1)   + y \psi(y), \nonumber 
\end{align}
 where, in the above, $\gamma$ is the Euler-Mascheroni constant,
 $\psi(y)$ is the digamma function (with $y \psi(y)$ to be understood as zero for $y=0$),
and we have used the Newton series expansion of the digamma function
\begin{equation} \label{eqn:digamma:newton}
\psi(y+1) = -\gamma-\sum_{j=1}^\infty \frac{(-1)^j}{j} \binom{y}{j}.
\end{equation}
From \eqref{eqn:Poi:uLogu:b}, we see that
\begin{equation} \label{eqn:Poi:trunc}
\E[\hat{g}(Y_x)]= \mu \log \mu -\mu E_1(1) + \mu E_1(\mu), 
\end{equation}
so that the function $\hat{g}(y) + yE_1(1)$ provides the desired approximation.
This is precisely the function $g$ that we defined in 
\eqref{eqn:gPoi}.

\begin{remark}
We could have truncated the integral upper limit to any constant $c \in [0,1]$
and yet obtain an exponentially sharp approximation of $\mu \log \mu$ for
large $\mu$ (choosing $c>1$, though, would result in an exponential growth of 
$\hat{g}(y)$ in $y$ and, subsequently, an expression for the probability mass function of $Y$ that
cannot be normalized to a valid distribution). Among these choices, truncation at $c=1$ provides the
closest possible approximation.
\end{remark}

\begin{remark} \label{rem:Poi:noFull}
The observation that the expression for $\tilde{g}(y)$, the
solution to the functional equation \eqref{eqn:Poi:functional},
does not converge shows that the KKT equality conditions \eqref{eqn:cvx1:KKT} of Theorem~\ref{thm:meanlimited}
cannot be simultaneously satisfied for all $x \geq 0$. In other words, for any mean-limited Poisson channel,
there is no input distribution $X$ with
full support on non-negative integers that achieves the capacity of the channel.
However, the optimal $X$ must have infinite support (since otherwise, for large enough $x$,
the KL divergence on the left hand side of \eqref{eqn:cvx1:dualF} becomes infinite and the KKT
conditions would be violated). The claim that the optimal $X$ cannot have full 
support makes intuitive sense. Intuitively, if the input distribution has nonzero
support on some $x>0$, then the corresponding channel output is expected to be $\lam x$ with a variance
of $\lam x$. Therefore, any input $x'$ for which $\lam |x-x'|$ is too close to the standard
deviation $\sqrt{\lam x}$ would cause confusion at the decoder and should be avoided. 
Roughly, this means that
if $x$ is on the support of the transmitter's input distribution $X$, 
the next symbol in the codebook should 
be picked at $x+\Omega(\sqrt{x/\lam})$.
\end{remark}

\subsection{Analytic estimates}
\label{sec:Poi:estimate}

It is desirable to provide sharp upper and lower bound estimates on the mean
and normalizing constants of the distribution \eqref{eqn:PoiYfirst}
and the digamma distribution \eqref{eqn:PoiYre} in terms of elementary or standard special functions,
and that is what we achieve in this section. 

\subsubsection{Estimates by standard special functions}
\label{sec:Poi:estimate:special}

First, we obtain sharp estimates on the parameters of the distribution \eqref{eqn:PoiYfirst}
in terms of standard special functions. Since the refined digamma distribution 
\eqref{eqn:PoiYre} achieves better capacity upper bounds than \eqref{eqn:PoiYfirst},
and, as we see in the next section,
we are able to estimate the parameters of the former sharply in terms
of elementary functions, the result of this section should be regarded
as a side result. However, the techniques presented here will be used for the more
complex problem of approximating the inverse binomial distribution, for the
deletion channel problem, in terms of standard special functions. 
It is thus natural to demonstrate the approximation techniques for the
mean-limited Poisson channel first before approaching the slightly more complex
case of the binomial channel.


We recall the standard special function 
Lerch transcendent (cf.~\cite[p.~27]{ref:EMOT53}), given by
\begin{equation} \label{eqn:Lerch}
\Phi (z,s,\alpha ):=\sum _{k=0}^{\infty }{\frac {z^{k}}{(k+\alpha )^{s}}}
= {\frac  {1}{\Gamma (s)}}\int _{0}^{\infty }{\frac  {t^{{s-1}}e^{{-\alpha t}}}{1-ze^{{-t}}}}\,dt.
\end{equation}

The approximation of \eqref{eqn:PoiYfirst} in terms of the above function 
is given by the following theorem, that we prove in Appendix~\ref{app:thm:PoiChUpperAnalytic}:
\begin{thm} \label{thm:PoiChUpperAnalytic}
Let $q\in(0,1)$ be a given parameter and $Y$ be a random
variable distributed according to \eqref{eqn:PoiYfirst},
for an appropriate normalizing constant $y_0$. Let
$\mu := \E[Y]$, and consider constants 
$\lsigma := 1/6$ and $\usigma := 0.177 \approx 16/90$.
Define special functions
$S_0(q,\sigma) := \Phi(q,1/2,1+\sigma)$ and
$S_1(q,\sigma) := \frac{q}{\sqrt{2\pi}} \Phi(q,-1/2,1+\sigma) -\sigma S_0(q,\sigma)$.
Then,
\begin{enumerate}
\item We have the bounds 
\begin{equation} \label{eqn:PoiEstimates}
\begin{aligned}
y_0 &\geq 1/(1+S_0(q,\lsigma)) =: \underline{y_0}, &
y_0 &\leq 1/(1+S_0(q,\usigma)) =: \overline{y_0}, \\
\mu &\geq S_1(q,\usigma)/(1+S_0(q,\lsigma)) =: \underline{\mu},  &
\mu &\leq S_1(q,\lsigma)/(1+S_0(q,\usigma)) =: \overline{\mu}.
\end{aligned}
\end{equation}
%

\item Let $\cD$ denote any Poisson distribution with positive mean.
Then, capacity of the mean-limited Poisson channel $\ch_\mu(\cD)$
satisfies
\begin{equation} \label{eqn:PoiMeanLimCapUpperB}
\capa(\ch_\mu(\cD)) \leq -\mu \log q-\log y_0 \leq
-\overline{\mu}\log q-\log \underline{y_0}.
\end{equation}

\end{enumerate} \qed
\end{thm}
As demonstrated in Figure~\ref{fig:PoiMuLerch}, the
expressions in \eqref{eqn:PoiEstimates} provide
remarkably sharp upper and lower bound estimates on the normalizing constant $y_0$
and the expectation $\mu$, which are accurate within a
multiplicative factor of about $1 \pm 0.004$ for all $q\in(0,1)$.

\begin{figure}[t!]
\begin{center}
\includegraphics[width=0.49 \columnwidth]{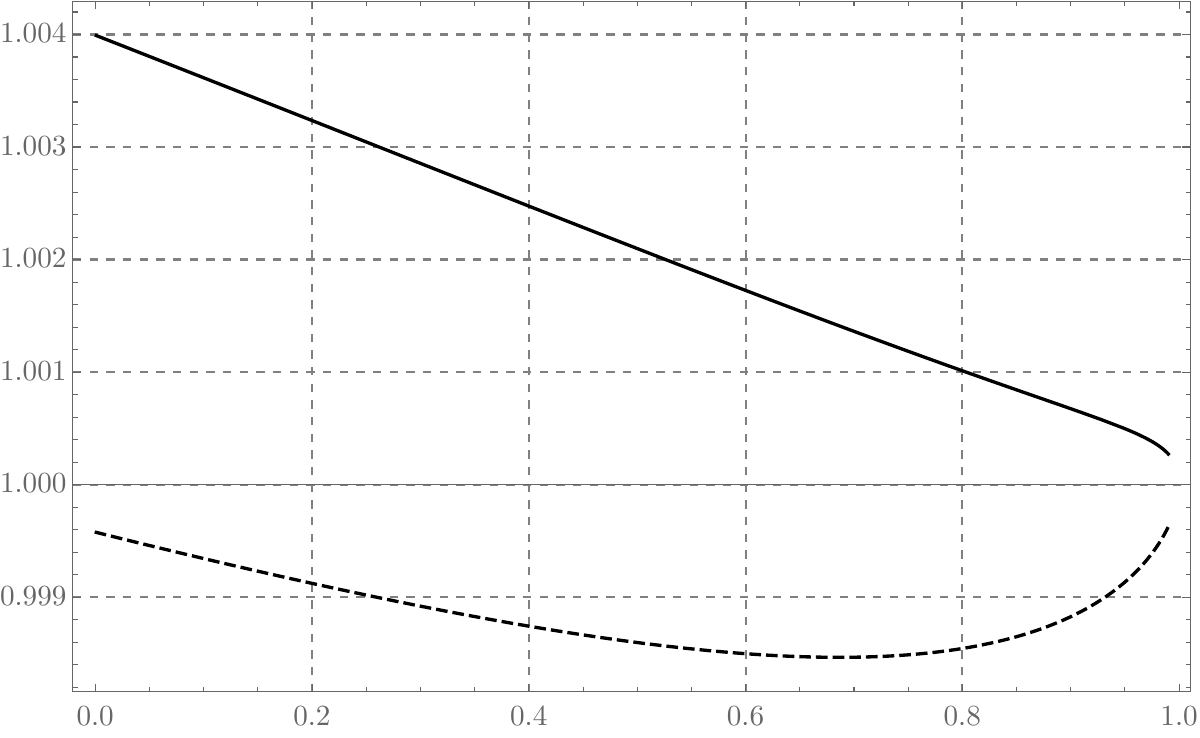} 
\includegraphics[width=0.49 \columnwidth]{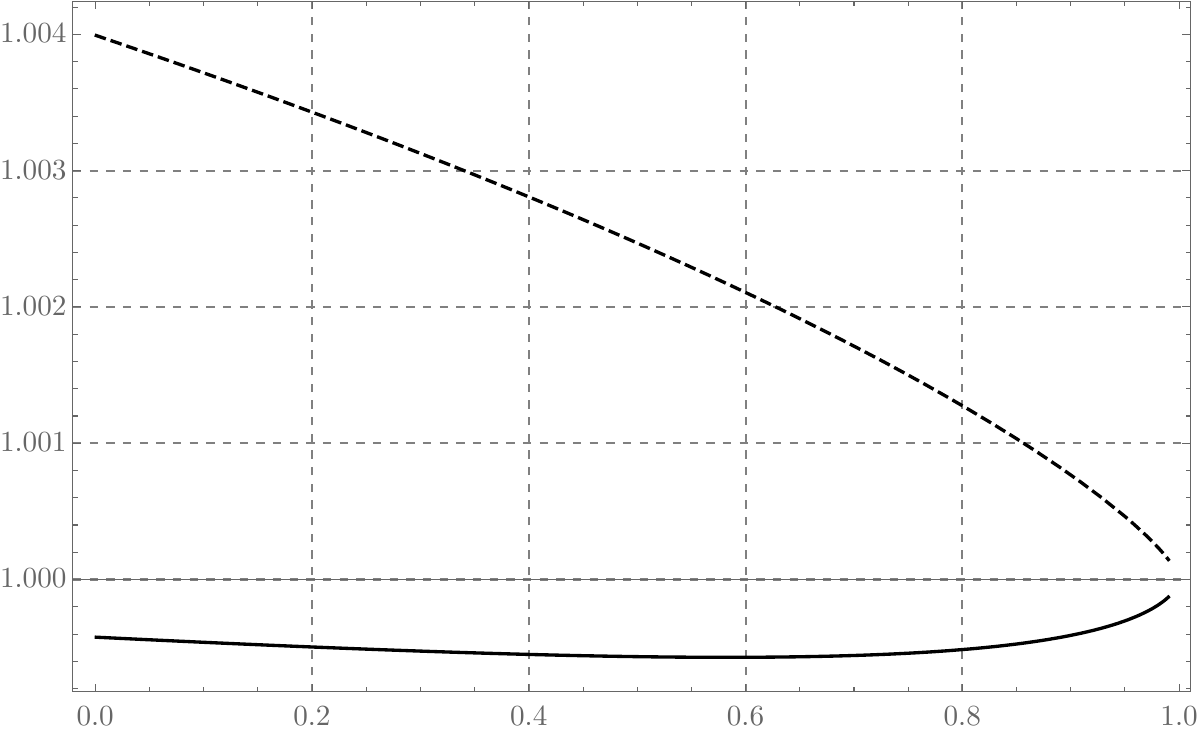}
\end{center}
\caption{Quality of the approximations in \eqref{eqn:PoiEstimates}
as a function of $q$. Left: $\overline{\mu}/\mu$ (solid) and
$\underline{\mu}/\mu$ (dashed). Right:
$(\log \overline{y_0})/(\log y_0)$ (solid) and
$(\log \underline{y_0})/(\log y_0)$ (dashed).
}
\label{fig:PoiMuLerch}
\end{figure}

\subsubsection{Estimate by the negative binomial distribution}
\label{sec:Poi:estimate:negBin}

In Section~\ref{sec:Poi:estimate:special}, we obtained
sharp estimates on the mean and the normalizing constant of the
distribution $Y$ in \eqref{eqn:PoiYfirst} in terms of standard
special functions. In this section, we provide
similar estimates, albeit not as sharp, for the
distribution \eqref{eqn:PoiYfirst} in terms of \emph{elementary} functions.

Recall that
a negative binomial distribution of order $r$ with ``success probability'' 
$q \in (0,1)$ 
is defined by the 
probability mass function
\begin{equation} \label{eqn:negbin}
\nbin_{r,q}(y)=\binom{y+r-1}{y} (1-q)^r q^y,\quad y=0,1,\ldots,
\end{equation}
and has mean $\mu=qr/(1-q)$ \cite[Section~5.3]{ref:DS12}. 
The asymptotic behavior of \eqref{eqn:negbin} at large $y$
can be understood by the Stirling approximation
$\Gamma(1+y) \sim \sqrt{2 \pi y} (y/e)^y$; namely,
\begin{align}
\nbin_{r,q}(y) &= \frac{\Gamma(y+r) (1-q)^r q^y}{\Gamma(y+1) \Gamma(r)} \nonumber\\
&\sim \frac{(1-q)^r q^y}{\Gamma(r)} \sqrt{\frac{y+r}{y+1}} e^{1-r} 
\frac{(y+r-1)^{y+r-1}}{y^y} \nonumber \\
&\sim \frac{(1-q)^r}{\Gamma(r)} e^{1-r} \left(1+\frac{r-1}{y}\right)^y q^y y^{r-1} \nonumber \\
&\sim \frac{(1-q)^r}{\Gamma(r)}\, q^y y^{r-1}  \label{eqn:negbin:asym:general}
\end{align}

\noindent Throughout this section, we focus
on the special case $r=1/2$, so
\eqref{eqn:negbin:asym:general} becomes
\begin{equation} \label{eqn:negbin:asym}
\nbin_{r,q}(y) \sim \sqrt{\frac{1-q}{\pi}} \, \frac{q^y}{\sqrt{y}}.
\end{equation}

We prove the following key estimate on the binomial coefficient $\binom{y-1/2}{y}$
that is used to provide accurate estimates on the parameters of the inverse
binomial distribution:

\newcommand{\gup}{{\overline{\gamma}}}
\newcommand{\gdown}{{\underline{\gamma}}}
\begin{lem} \label{lem:Poi:negBinApprox}
Let $\gdown := {2}/{e^{1+\gamma}} \approx 0.413099$ and 
$\gup := 1/\sqrt{2 e} \approx 0.428882$,
where $\gamma \approx 0.57721$ is the Euler-Mascheroni constant. Then,
for all $y\geq 1$,
\[
 \binom{y-1/2}{y}\, \gdown \leq  \frac{\exp(y \psi(y)-y)}{y!} \leq \binom{y-1/2}{y} \,\gup.
\]
\end{lem}

\begin{proof}
Consider the ratio
\begin{equation} \label{eqn:Poi:ratio}
g(y) = \frac{\binom{y-1/2}{y}}{\exp(y \psi(y))/(y! e^y)}=
\frac{\Gamma(y+1/2)}{\Gamma(1/2)\exp(y \psi(y)-y)}=
\frac{1}{\sqrt{\pi}}\Gamma(y+1/2) \exp(y-y \psi(y)),
\end{equation}
We use the following claim (see Appendix~\ref{app:Poi:ratio}):

\begin{claim} \label{claim:Poi:ratio}
The function $g(y)$ defined in \eqref{eqn:Poi:ratio} 
is a decreasing function of $y > 0$. 
\end{claim}
\noindent The above claim implies that, for all $y \geq 1$, we must have
\[
g(1) \geq g(y) \geq \lim_{y \to \infty} g(y),
\]
assuming that the limit exists (that we will show next). 
We have
\[
g(0)=\frac{1}{\sqrt{\pi}}\Gamma(3/2) \exp(1- \psi(1)) = \frac{e^{1+\gamma}}{2} \approx 2.420728.
\]
Recall, from \eqref{eqn:Poi:StirApprox:alternative}, that
\[
\exp(y \psi(y))/(y! e^y) \sim \frac{1}{\sqrt{2\pi e y}},
\]
and similarly, from \eqref{eqn:negbin:asym}, that
\[
\binom{y-1/2}{y} \sim 
{\frac{1}{\sqrt{\pi y}}},
\]
and therefore,
\[
\lim_{y \to \infty} g(y) = \sqrt{2e} \approx 2.331643.
\]
The result follows.
\end{proof}

\begin{figure}[t!]
\begin{center}
\includegraphics[width=0.49 \columnwidth]{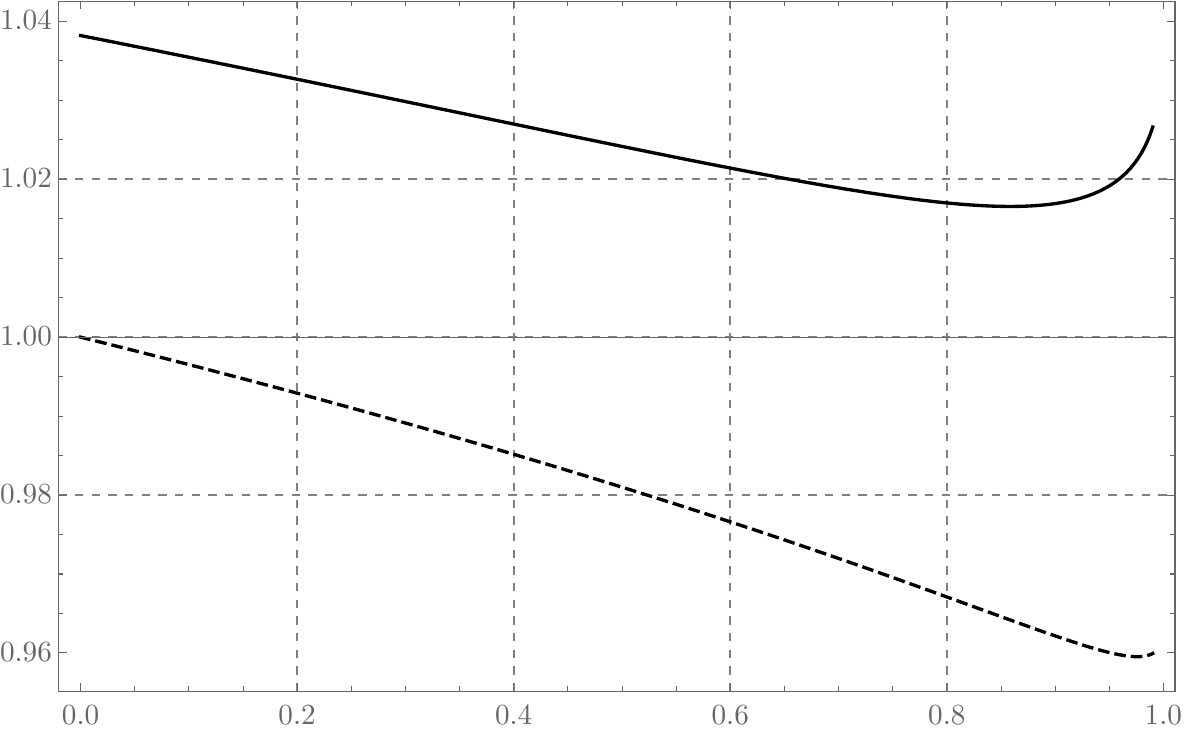} 
\includegraphics[width=0.49 \columnwidth]{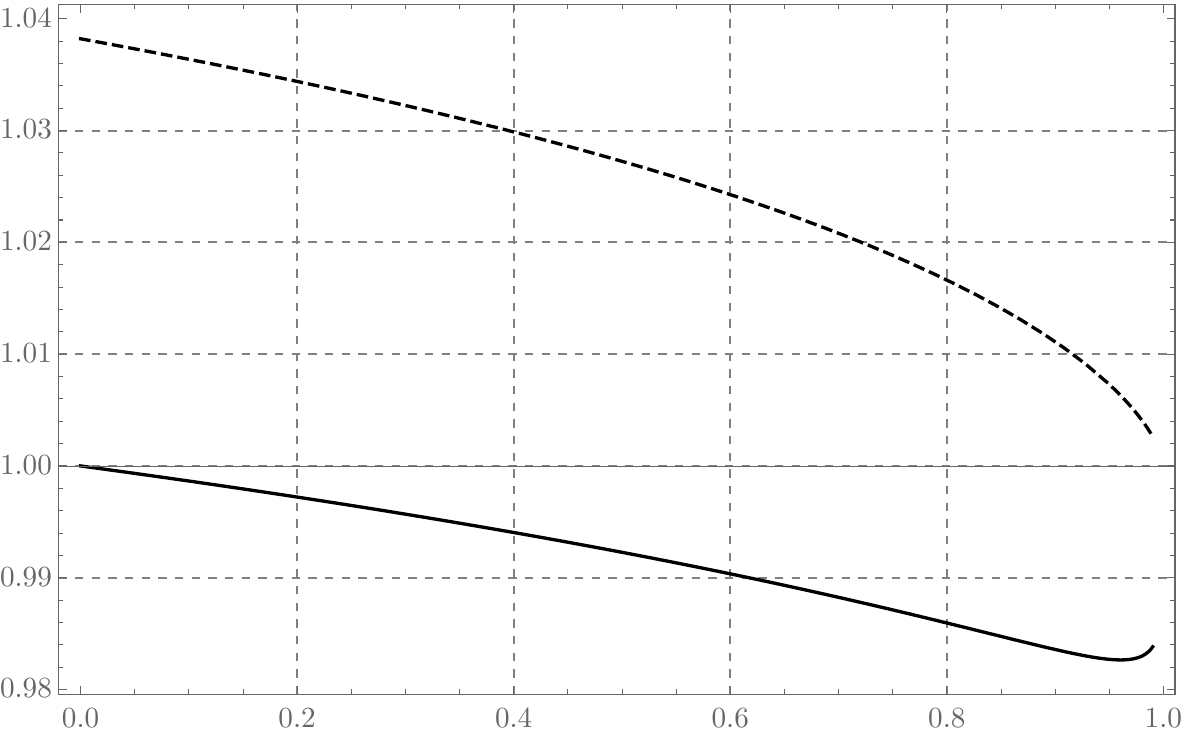}
\end{center}
\caption{Quality of the approximations by Corollary~\ref{coro:Poi:negBinApprox}
as a function of $q$. Left: $\overline{\mu}/\mu$ (solid) and
$\underline{\mu}/\mu$ (dashed). Right:
$(\log \overline{y_0})/(\log y_0)$ (solid) and
$(\log \underline{y_0})/(\log y_0)$ (dashed).
The notation $\overline{\mu}$ and $\underline{\mu}$ ($\overline{y_0}$ and $\underline{y_0}$)
respectively refer to the upper and lower bound estimates on $\mu$ ($y_0$) 
by Corollary~\ref{coro:Poi:negBinApprox}.
}
\label{fig:PoiMuNegBin}
\end{figure}

Using the above approximations, we are able to prove the following
lower and upper bound estimates on the parameters of the digamma distribution
\eqref{eqn:PoiYre} in terms of elementary functions (see Figure~\ref{fig:PoiMuNegBin}
for a depiction of the quality of the approximations).

\begin{coro} \label{coro:Poi:negBinApprox}
Let $\gdown$ and $\gup$ be as in Lemma~\ref{lem:Poi:negBinApprox},
and $Y$ be distributed according to the digamma distribution \eqref{eqn:PoiYre}.
Then,
\begin{enumerate}
\item For all $y\geq 1$,
\begin{equation} \label{eqn:Poi:PYnegbin}
\gdown\, \nbin_{1/2,q}(y) \leq \sqrt{1-q} \Pr[Y=y]/y_0 \leq \gup\, \nbin_{1/2,q}(y).
\end{equation}

\item The normalizing constant $y_0$ satisfies
\begin{equation} \label{eqn:Poi:y0:negbin}
\log\left(1+\gdown\left(\frac{1}{\sqrt{1-q}}-1\right)\right) \leq
-\log y_0 \leq \log\left(1+\gup\left(\frac{1}{\sqrt{1-q}}-1\right)\right).
\end{equation}

\item The mean $\mu = \E[Y]$ satisfies
\begin{equation} \label{eqn:Poi:mu:negbinA}
\frac{\gdown q}{2 (1-q)^{3/2}} \leq \frac{\mu}{y_0} \leq \frac{\gup q}{2 (1-q)^{3/2}},
\end{equation}
and,
\begin{equation} \label{eqn:Poi:mu:negbin}
\frac{\gdown q}{2 (1-q) (\sqrt{1-q}+\gup(1-\sqrt{1-q}))} \leq
\mu \leq \frac{\gup q}{2 (1-q) (\sqrt{1-q}+\gdown(1-\sqrt{1-q}))}.
\end{equation}

\end{enumerate}
\end{coro}

\begin{proof}
The first part is immediate from the expression of the
digamma distribution \eqref{eqn:PoiYre} combined with
the result of Lemma~\ref{lem:Poi:negBinApprox}. Now, let
$Z$ be distributed according to $\nbin_{1/2,q}$, and write, from
the definition of the normalizing constant,
\begin{align}
1/y_0&=  {1+\sum_{y=1}^\infty \Pr[Y=y]/y_0} \nonumber \\
&\stackrel{\eqref{eqn:Poi:PYnegbin}}{\leq}
{1+\left(\frac{\gup}{\sqrt{1-q}}\right)\sum_{z=1}^\infty \Pr[Z=z]} \nonumber \\
&= {1+\left(\frac{\gup}{\sqrt{1-q}}\right)\left(1-\sqrt{1-q}\right)} \nonumber\\
&= 1+\gup \left(\frac{1}{\sqrt{1-q}}-1\right), \label{eqn:Poi:y0:nbin}
\end{align}
which, after taking the logarithms of both sides, proves the
upper bound in \eqref{eqn:Poi:y0:negbin}. Proof of the lower bound
is similar. To upper bound the mean, we may write
\begin{align}
\frac{\mu}{y_0} &= \sum_{y=1}^\infty y \Pr[Y=y]/y_0 \nonumber \\
&\stackrel{\eqref{eqn:Poi:PYnegbin}}{\leq}
\sum_{z=1}^\infty \gup z \Pr[Z=z]/\sqrt{1-q} \nonumber \\
&= \gup \E[Z]/\sqrt{1-q} = \frac{\gup q}{2(1-q)^{3/2}},
\label{eqn:Poi:mu:nbin}
\end{align}
which, combined with a similar lower bound, proves \eqref{eqn:Poi:mu:negbin}.
Finally, combining this result with \eqref{eqn:Poi:y0:negbin}
yields \eqref{eqn:Poi:mu:negbin}.
\end{proof}

\subsection{Derivation of the capacity upper bound for 
the Poisson-repeat channel}

In order to obtain a capacity upper bound for the Poisson-repeat
channel, it suffices to combine 
Corollary~\ref{coro:Cupper} with either Theorem~\ref{thm:PoiChUpper}
or  Theorem~\ref{thm:PoiChUpperAnalytic}.
Let $\ch$ be a $\cD$-repeat channel, where $\cD$ 
is a Poisson distribution with mean $\lam$.
The probability $d$ assigned by $\cD$ to zero is thus
equal to $e^{-\lam}$, and its complement is $p:=1-d=1-e^{-\lam}$.
Since $d$ captures the ``deletion probability'' of the channel,
we parameterize the channel in terms of $p$ rather than the
Poisson parameter $\lam$. Note that $\lam = -\log(1-p)$.
From \eqref{eqn:coro:Cupper}, we may write
\begin{align*}
\capa(\ch) &\leq \sup_{\mu>0} \frac{\capa(\ch_\mu(\cD))}{-\mu/\log(1-p)+1/p}.
\end{align*}
We may now use Theorem~\ref{thm:PoiChUpper} with the corresponding
choice for the random variable $Y$ with parameter $q$, 
mean $\mu$ and normalizing constant $y_0$ and write, using \eqref{eqn:PoiMeanLimCapUpperA},
\begin{align} \label{eqn:Poi:capa:general}
\capa(\ch) &\leq \sup_{q\in(0,1)} 
\frac{-\mu \log q-\log y_0}{-\mu/\log(1-p)+1/p}.
\end{align}
Of course, this result would remain 
valid had we used Theorem~\ref{thm:PoiChUpperAnalytic}
and \eqref{eqn:PoiYfirst} for the distribution of $Y$.
Furthermore, the analytic approximations of Corollary~\ref{coro:Poi:negBinApprox}
may be used to upper bound the right hand side of 
\eqref{eqn:Poi:capa:general} by the supremum of an elementary function
of the channel parameter and $q$ (which can be shown to be concave in $q$). 
We summarize the above result in the following
theorem:

\begin{thm} \label{thm:Poi:capa}
Let $\ch$ be a Poisson-repeat channel with deletion probability $d \in (0,1)$
(or equivalently, repetition mean $\lam = \log(1/d)$ per bit). For a parameter
$q \in (0,1)$, let $Y$
be distributed according to either the 
distribution \eqref{eqn:PoiYfirst} or the digamma distribution \eqref{eqn:PoiYre}
 with an appropriate normalizing constant $y_0$ 
and let $\mu := \E[Y]$ denote its mean. Then,
\begin{align} \label{eqn:thm:Poi:capa}
\capa(\ch) &\leq \sup_{q\in(0,1)} 
\frac{-\mu \log q-\log y_0}{-\mu/\log d+1/(1-d)}.
\end{align}
Furthermore, let $\gdown := {2}/{e^{1+\gamma}} \approx 0.413099$ and 
$\gup := 1/\sqrt{2 e} \approx 0.428882$,
where $\gamma \approx 0.57721$ is the Euler-Mascheroni constant. 
Then, 
\begin{align} \label{eqn:thm:Poi:capa:approx}
\capa(\ch) &\leq \sup_{q\in(0,1)} 
\frac{-\overline{\mu} \log q-\log \underline{y_0}}{-\underline{\mu}/\log d +1/(1-d)},
\end{align}
where
\begin{equation} \label{eqn:thm:Poi:capa:approx:eqns}
\begin{aligned}
\overline{\mu} &:= \frac{\gup q}{2 (1-q) (\sqrt{1-q}+\gdown(1-\sqrt{1-q}))}, \\
\underline{\mu} &:= \frac{\gdown q}{2 (1-q) (\sqrt{1-q}+\gup(1-\sqrt{1-q}))}, \\
\underline{y_0} &:= 1+\gup\left(\frac{1}{\sqrt{1-q}}-1\right).
\end{aligned}
\end{equation}

\end{thm} \qed

\begin{remark}
Observe that the distributions of $Y$ defined by 
\eqref{eqn:PoiYfirst} or the digamma distribution \eqref{eqn:PoiYre}
 in Theorem~\ref{thm:Poi:capa}
only depend on the choice of the parameter $q$ and not 
the channel parameter $d$ at all. Therefore, the calculations
of mean and normalizing constant for various choices of
$q$ can be reused for the computation of the capacity upper bound
for different choices of $d$.
\end{remark}

%

\noindent Observe that, as $p\rightarrow 0$, the right hand
side of \eqref{eqn:Poi:capa:general} converges to
\begin{equation} \label{eqn:Poi:slope}
p \sup_{q\in(0,1)} 
\frac{-\mu \log q-\log y_0}{\mu+1},
\end{equation}
and that the expression under the supremum is independent of the
channel parameter $p$. The expression under the supremum is
plotted in Figure~\ref{fig:PoiSLope} and can be numerically calculated efficiently, 
which results in the following
corollary:
\begin{coro} \label{coro:Poi:slope}
Let $\sC(d)$ denote the capacity of the Poisson-repeat
channel with deletion probability $d$. Then for $d \to 1$,
\[
\sC(d) \leq 0.464421 (1-d)\cdot (1+o(1))\ %
\text{\rm bits per channel use}.
\]
\end{coro}

\begin{figure}[t!]
\begin{center}
\includegraphics[height=2in]{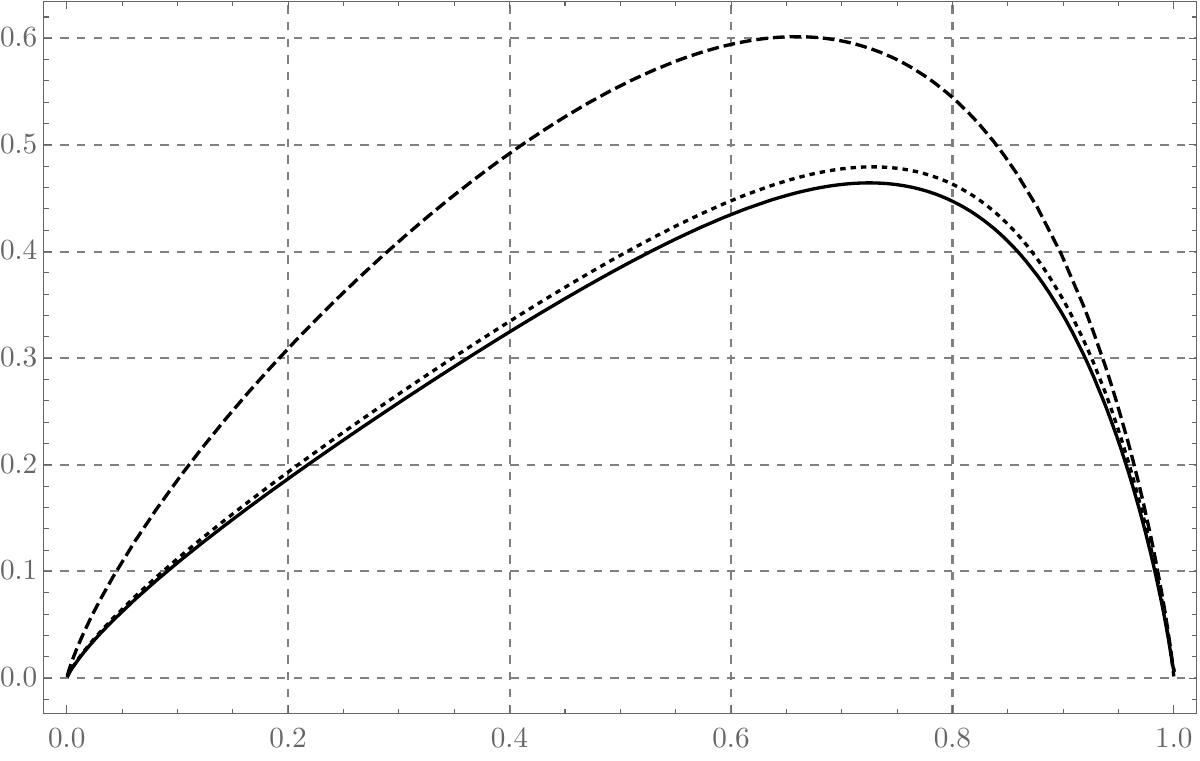}
\end{center}
\caption{The expression inside the supremum in \eqref{eqn:Poi:slope}
(measured in bits) plotted as a function of $q$ and with respect to the choices 
\eqref{eqn:PoiYre} (solid) and \eqref{eqn:PoiYfirst} (dashed) 
for the distribution of $Y$. The maximums are attained at
$q\approx 0.724762$ (solid) and $q\approx 0.659046$ (dashed),
resulting in supremums $\approx 0.464420$ (solid) and 
$\approx 0.601549$ (dashed). The analytic estimates of 
Theorem~\ref{thm:PoiChUpperAnalytic} 
result in an indistinguishable
plot from the dashed one, and a slightly higher supremum of
$\approx 0.602987$ attained at $q \approx 0.658810$.
The dotted plot depicts the elementary upper bound resulting
from \eqref{eqn:thm:Poi:capa:approx:eqns},
where the maximum is $\approx 0.479454$ and attained at $q \approx 0.727855$.
}
\label{fig:PoiSLope}
\end{figure}

Plots of the resulting capacity upper bounds for general $d$
are given in Figure~\ref{fig:PoiPlot}.
Moreover, the corresponding numerical values for the plotted curves are
listed, for various choices of the deletion parameter $d$, 
in Table~\ref{tab:PoiCapData}.

\begin{figure}[ht!] 
\begin{center}
\includegraphics[width=0.8\columnwidth]{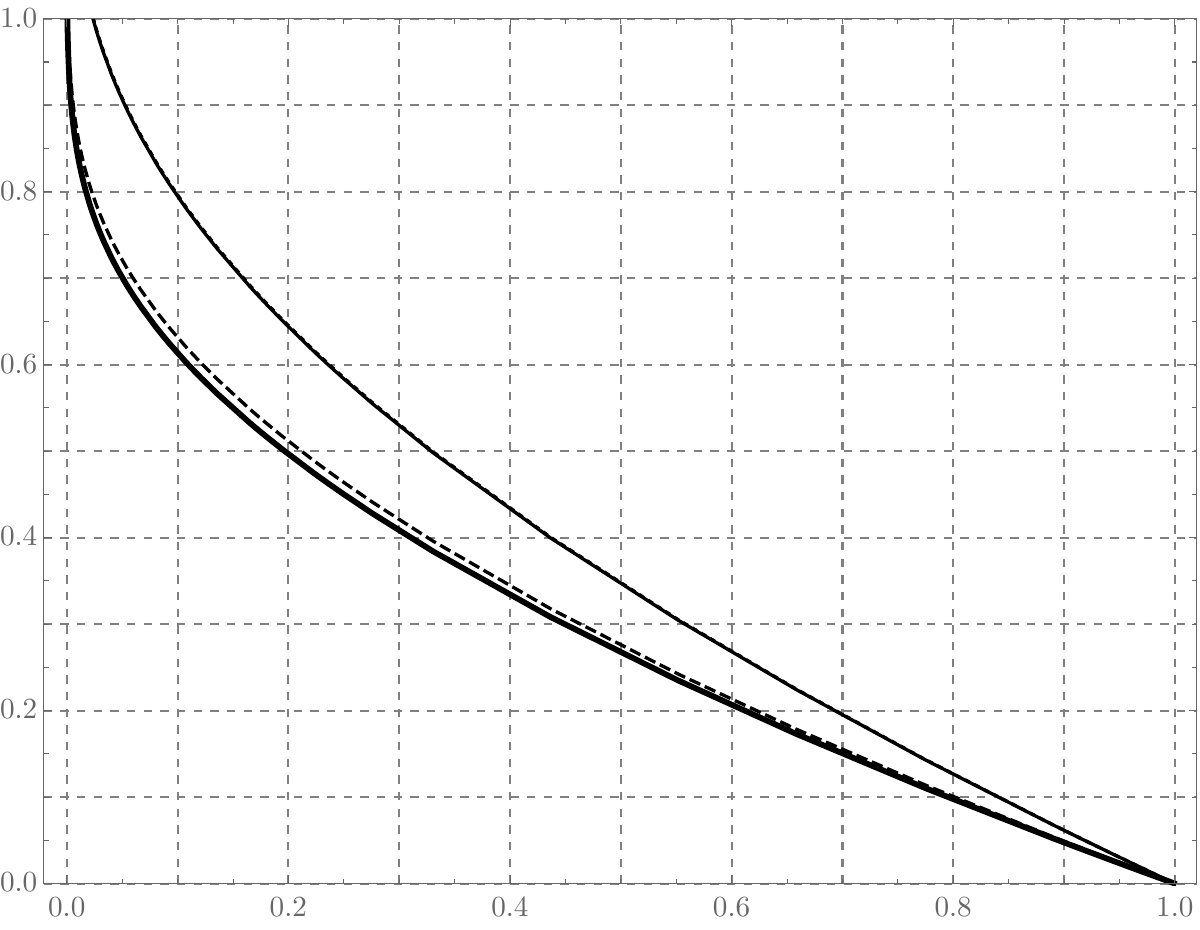}
\end{center}
\caption[]{Upper bounds (in bits per channel use) 
on the capacity of the Poisson-repeat channel given by
Theorem~\ref{thm:Poi:capa},
plotted as a function of the deletion probability $d=1-p=e^{-\lam}$.
The bounds are obtained using 
\begin{enumerate*}[label=\roman*)]
\item The digamma distribution \eqref{eqn:PoiYre} for $Y$  (solid, thick),
\item  The distribution \eqref{eqn:PoiYfirst} for $Y$ (solid),
\item The elementary upper bound estimates of  \eqref{eqn:thm:Poi:capa:approx:eqns}
on the parameters of the digamma distribution \eqref{eqn:PoiYre} (dashed),
\item The analytic upper bound estimates of Theorem~\ref{thm:PoiChUpperAnalytic}
on the parameters of \eqref{eqn:PoiYfirst}
(dotted, nearly indistinguishable from the solid curve obtained by a numerical 
computation of the actual parameters of the distribution). 
\end{enumerate*}
}
\label{fig:PoiPlot}
\vspace{2mm} \hrule
\end{figure}

\section{Upper bounds on the capacity of the deletion channel}
\label{sec:del}

\subsection{The inverse binomial distribution}
\label{sec:invbin}

Let $\ibin_{p,q}$ be a distribution on
non-negative integers, parameterized by
$p \in (0,1]$ and $q \in (0,1)$, and defined by the probability
mass function\footnote{%
We call this distribution ``inverse binomial'' since,
as we see in Section~\ref{sec:del:general},
the mutual information between any binomial distribution
$\Bin_{x,p}$ and the inverse binomial distribution essentially
simplifies to a linear term in $x$.
Therefore, intuitively, the binomial distribution ``neutralizes'' any 
binomial distribution in the divergence computations. Moreover, 
thinking of $y$ as the posterior of a binomial sampling $\Bin_{X,p}$,
we ``invert'' it back to the expected prior $y/p$.}
\begin{equation} \label{eqn:invbin}
\ibin_{p,q}(y) := y_0 \binom{y/p}{y} q^y \exp(-y h(p)/p),
\end{equation}
where $h(p)$ is the binary entropy function, and 
$y_0 \in (0,1)$ is the appropriate
normalizing constant:

\[
1/y_0 = 1+\sum_{y=0}^\infty \binom{y/p}{y} q^y \exp(-y h(p)/p).
\]
Note that, when $p=1$, the above reduces to a geometric distribution.

By the Stirling approximation, we may write the well known asymptotic expression
\begin{align*}
\binom{x}{px} &\sim 
\frac{(x/e)^x}{\sqrt{2\pi p(1-p)x}(px/e)^{px} (x(1-p)/e)^{(1-p)x}} \\
&= \frac{\exp(h(p)x)}{\sqrt{2\pi p(1-p)x}},
\end{align*}
and thus, we can 
identify the asymptotic
behavior of \eqref{eqn:invbin} at $y \rightarrow \infty$ as
\begin{align} \label{eqn:ibin:asymp}
\ibin_{p,q}(y) &\sim y_0 \frac{q^y}{\sqrt{2\pi (1-p) y}},
\end{align}
implying that the normalizing constant is well defined, 
leading to a legitimate distribution, exactly when $q\in (0,1)$.
Moreover, the expectation of the distribution can be made
arbitrarily small as $q\to 0$ and arbitrarily large as $q\to 1$
(since, by \eqref{eqn:ibin:asymp},
when $q = 1$ we have $\ibin_{p,1}(y)/y_0 = \Theta(1/\sqrt{y})$
and the summation defining $y_0$ becomes divergent). Therefore,
by varying the value of $q \in (0,1)$, it is possible to adjust
the expectation of the distribution to any desired positive value.

\subsubsection{Estimate by the negative binomial distribution}
\label{sec:invbin:negbin}

As was the case for the distribution \eqref{eqn:PoiY} 
related to the Poisson-repeat
channel, in this section we show that the inverse binomial
distribution can be approximated by a negative binomial
distribution of order $r=1/2$.
Towards this goal, we will use the following analytic claim
(derived in Appendix~\ref{app:invbin:ratio}):

\begin{claim} \label{claim:invbin:ratio}
Let $p\in (0,1)$. The ratio
\begin{equation} \label{eqn:invbin:ratio}
\rho(y) := \frac{\binom{y/p}{y} \exp(-y h(p) /p)}{\binom{y-1/2}{y}}
\end{equation}
is a decreasing function of $y>0$ for $p<1/2$, is equal to $1$ for $p=1/2$,
and is an increasing function of $y>0$ for $p>1/2$.
\end{claim}

The fact that, by Claim~\ref{claim:invbin:ratio},
$\rho(y)=1$ for $p=1/2$ implies that the inverse binomial distribution
for the special case $p=1/2$ is \emph{precisely} a negative binomial
distribution, and thus in this case, we have 
\begin{equation} \label{eqn:ibin:half}
y_0=\sqrt{1-q}, \quad \E[\ibin_{1/2,q}]=q/(2(1-q)).
\end{equation}
We show that, when $p \neq 1/2$, the inverse binomial distribution
is still reasonably approximated by a negative binomial 
distribution of order $r=1/2$, and that the quality of this
approximation improves as $p$ gets closer to $1/2$.
%
Define \begin{equation} \label{eqn:beta0}
\beta_0 := \rho(1) = (2/p) \exp(-h(p)/p).
\end{equation} 
By combining \eqref{eqn:negbin:asym} and
\eqref{eqn:ibin:asymp}, the ratio in \eqref{eqn:invbin:ratio}
satisfies
\begin{equation} \label{eqn:beta1}
\beta_1 := \lim_{y\to \infty} \rho(y) = \frac{1}{\sqrt{2(1-p)}}.
\end{equation}
We observe that, when $p=1/2$, we have $\beta_0 = \beta_1 = 1$,
and $\beta_0 \to 2/e \approx 0.735759$ as $p\to 0$.
For $p<1/2$, we have $\beta_0 > \beta_1$, whereas for $p>1/2$,
we have $\beta_0 < \beta_1$. This leads to the following
analogous result to \eqref{eqn:Poi:PYnegbin}:
\begin{lem} \label{lem:invbin:estimate}
For a parameter $p \in (0,1)$, let 
the $\beta_0$ and $\beta_1$ be the constants defined in 
\eqref{eqn:beta0} and \eqref{eqn:beta1}. 
Let $\bDown := \min\{\beta_0, \beta_1\}$ and
$\bUp := \max\{\beta_0, \beta_1\}$ (in particular,
for $p=1/2$, we have $\bDown=\bUp=1$). Then, for all $y \geq 1$,
\begin{equation} \label{eqn:invbin:estimate}
\bDown\, \nbin_{1/2,q}(y) \leq
\frac{\sqrt{1-q}}{y_0}\, \ibin_{p,q}(y) \leq \bUp\, \nbin_{1/2,q}(y).
\end{equation}
\end{lem}

The above lemma can now be used to derive upper and lower 
estimates on the mean and normalizing constant of an
inverse binomial distribution.
Let $Y$ be distributed according to $\ibin_{p,q}$,
and $y_0$ denote the corresponding normalizing constant in
\eqref{eqn:invbin}. Furthermore, let a random variable $Z$ be distributed 
according to $\nbin_{1/2,q}$.
As in \eqref{eqn:Poi:y0:nbin}, we may now proceed by writing
\begin{align*}
1/y_0 &= 1+\sum_{y=1}^\infty \Pr[Y=y]/y_0 \\
&\stackrel{\eqref{eqn:invbin:estimate}}{\leq}
1+ \frac{\bUp}{\sqrt{1-q}}\, \sum_{z=1}^\infty \Pr[Z=z] \\
&= 1+\frac{\bUp}{\sqrt{1-q}}\left(1-\sqrt{1-q}\right) \\
&= 1+\bUp \left(\frac{1}{\sqrt{1-q}}-1\right).
\end{align*}
Moreover, we may derive a similar expression to \eqref{eqn:Poi:mu:nbin},
namely, letting $\mu := \E[Y]$, that we have
\begin{equation}
\frac{\mu}{y_0} \leq \frac{\bUp q}{2(1-q)^{3/2}}.
\label{eqn:ibin:mu:nbin}
\end{equation}
Similarly, we may derive lower bounds on $\mu$ and $y_0$
by using the lower bounding constant $\bDown$.
This leads to the following result, which is analogous
to Corollary~\ref{coro:Poi:negBinApprox} (see Figure~\ref{fig:BinEstimates}
for plots on the quality of this approximation): 
\begin{coro} \label{coro:invbin:negBinApprox}
Consider the inverse binomial distribution $\ibin_{p,q}$,
with mean $\mu$, as defined in \eqref{eqn:invbin}.
Define $\bDown$ (resp., $\bUp$) to be the minimum (resp., the maximum)
of the two constants 
$(2/p) \exp(-h(p)/p)$ and  ${1}/{\sqrt{2(1-p)}}$.
Then,
\begin{gather} \label{eqn:ibin:y0:negbin}
\log\left(1+\bDown\left(\frac{1}{\sqrt{1-q}}-1\right)\right) \leq
-\log y_0 \leq \log\left(1+\bUp\left(\frac{1}{\sqrt{1-q}}-1\right)\right),
 \\
 \label{eqn:ibin:mu:negbinA}
\frac{\bDown q}{2 (1-q)^{3/2}} \leq \frac{\mu}{y_0} \leq \frac{\bUp q}{2 (1-q)^{3/2}}, \\
\frac{\bDown q}{2 (1-q) (\sqrt{1-q}+\bUp(1-\sqrt{1-q}))} \leq
\mu \leq \frac{\bUp q}{2 (1-q) (\sqrt{1-q}+\bDown(1-\sqrt{1-q}))}.
\label{eqn:ibin:mu:negbinB}
\end{gather}

\end{coro}

\subsubsection{Estimates by standard special functions}
\label{sec:invbin:spacial}

The estimates of Corollary~\ref{coro:invbin:negBinApprox} provide
high-quality upper and lower bounds on the mean and the normalizing
constant of an inverse binomial distribution in terms of 
elementary functions. In fact, the bounds
are exact for $p=1/2$, and a numerical computation shows that
\eqref{eqn:ibin:y0:negbin} and
\eqref{eqn:ibin:mu:negbinB}
are within a multiplicative factor of about $1.2$ for all
$p \leq 0.8$ and $q \in (0,1)$.
However, as $p$ approaches $1$, the quality of the estimates
degrade, as the ratio between $\bUp$ and $\bDown$
tends to infinity when $p\to 1$. In this section, we provide
a different set of upper and lower bounds in terms of 
standard special functions.

Our starting point is the following claim on the binomial coefficients
(see Appendix~\ref{app:claim:binomial} for a derivation):
\begin{claim} \label{claim:binomial}
There are universal constants $\aDown, \aUp > 0$ such that
for all $y \geq 1$ and $p \in (0,1)$, we have
\begin{equation} \label{eqn:claim:binomial}
\frac{\exp(y h(p)/p)}{\sqrt{2\pi ((1-p)y+\aDown)}} \leq
\binom{y/p}{y} \leq \frac{\exp(y h(p)/p)}{\sqrt{2\pi ((1-p)y+\aUp)}}.
\end{equation}
In particular, one may take $\aDown = 0.19$ and $\aUp = 0.12$.
\end{claim}
Now, given a random variable $Y$ that is distributed 
according to \eqref{eqn:invbin}, we can write
\begin{align*}
1/y_0 &= 1+\sum_{y=1}^\infty \binom{y/p}{y} q^y \exp(-y h(p)/p)\\
&\stackrel{\eqref{eqn:claim:binomial}}{\leq}
1+\sum_{y=1}^\infty \frac{q^y}{\sqrt{2\pi ((1-p)y+\aUp)}} \\
&= 1+\frac{1}{\sqrt{2\pi (1-p)}} \sum_{y=1}^\infty \frac{q^y}{\sqrt{y+\aUp/(1-p)}}\\
&= 1+\frac{q\, \Phi(q,1/2,1+\aUp/(1-p))}{\sqrt{2\pi (1-p)}},
\end{align*}
where $\Phi(\cdot)$ denotes the
Lerch transcendent \eqref{eqn:Lerch}.
Similarly, an upper bound on $y_0$ may be obtained by replacing
$\aUp$ with $\aDown$ in the above. We now upper bound the mean
$\mu=\E[Y]$ as follows:
\begin{align*}
\mu/y_0 &= \sum_{y=1}^\infty \binom{y/p}{y} y q^y \exp(-y h(p)/p)\\
& \stackrel{\eqref{eqn:claim:binomial}}{\leq}
\sum_{y=1}^\infty \frac{y q^y}{\sqrt{2\pi ((1-p)y + \aUp)}}  \\
&= \frac{1}{\sqrt{2\pi (1-p)}} \sum_{y=1}^\infty \frac{y q^y}{\sqrt{y + \aUp/(1-p)}}\\
&= \frac{1}{\sqrt{2\pi (1-p)}} \left(
\sum_{y=1}^\infty q^y \sqrt{y + \aUp/(1-p)}
-\frac{\aUp}{1-p}\sum_{y=1}^\infty \frac{ q^y}{\sqrt{y + \aUp/(1-p)}}\right) \\
&= \frac{\Phi(q,-1/2,1+\aUp/(1-p))}{\sqrt{2\pi (1-p)}} -
\frac{\aUp\, \Phi(q,1/2,1+\aUp/(1-p))}{\sqrt{2\pi (1-p)^3}},
\end{align*}
and similarly, a lower bound may be obtained by replacing $\aUp$
with $\aDown$.
The above estimates are summarized in the following (and depicted in Figure~\ref{fig:BinEstimates}):

\begin{thm} \label{thm:ibin:Lerch}
For parameters $p,q \in (0,1)$, let $Y$ be distributed 
according to the inverse binomial distribution \eqref{eqn:invbin},
for an appropriate normalizing constant $y_0$, and let $\mu := \E[Y]$.
Let $\aDown$, $\aUp$ be the constants from Claim~\ref{claim:binomial};
namely, $\aDown := 0.19$ and $\aUp := 0.12$.
Define the functions
\begin{align}
S_0(p,q,\alpha)&:=
1+\frac{q\, \Phi(q,1/2,1+\alpha/(1-p))}{\sqrt{2\pi (1-p)}},
 \label{eqn:S0:ibin}  \\
S_1(p,q,\alpha)&:=
\frac{\Phi(q,-1/2,1+\alpha/(1-p))}{\sqrt{2\pi (1-p)}} -
\frac{\alpha, \Phi(q,1/2,1+\alpha/(1-p))}{\sqrt{2\pi (1-p)^3}},
 \label{eqn:S1:ibin} 
\end{align}
where $\Phi(\cdot)$ denotes the
Lerch transcendent \eqref{eqn:Lerch}.
Then, the following estimates hold:
\begin{equation}
\begin{gathered}
S_0(p,q,\aDown) \leq 1/y_0 \leq S_0(p,q,\aUp),\\
S_1(p,q,\aDown) \leq \mu/y_0 \leq S_1(p,q,\aUp),\\
S_1(p,q,\aDown)/S_0(p,q,\aUp) \leq \mu \leq S_1(p,q,\aUp)/S_0(p,q,\aDown).
\end{gathered}
\end{equation}
\end{thm}
Compared with the negative binomial estimate (Corollary~\ref{coro:invbin:negBinApprox}), the above estimates are
inferior around $p=1/2$. However, unlike the former,
Theorem~\ref{thm:ibin:Lerch} provides remarkably accurate
estimates on the mean and the normalizing constant of an inverse
binomial distribution for all $p, q \in (0,1)$.

\begin{figure}[t!]
\begin{center}
\includegraphics[width=0.49 \columnwidth]{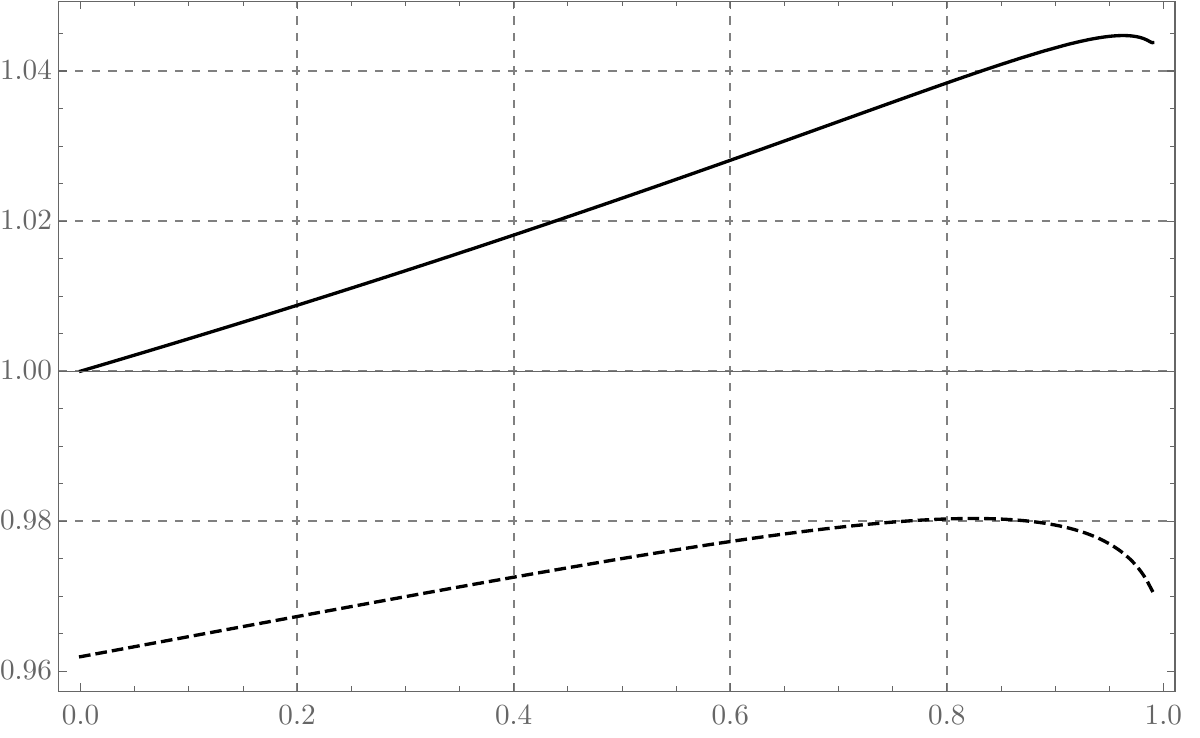} 
\includegraphics[width=0.49 \columnwidth]{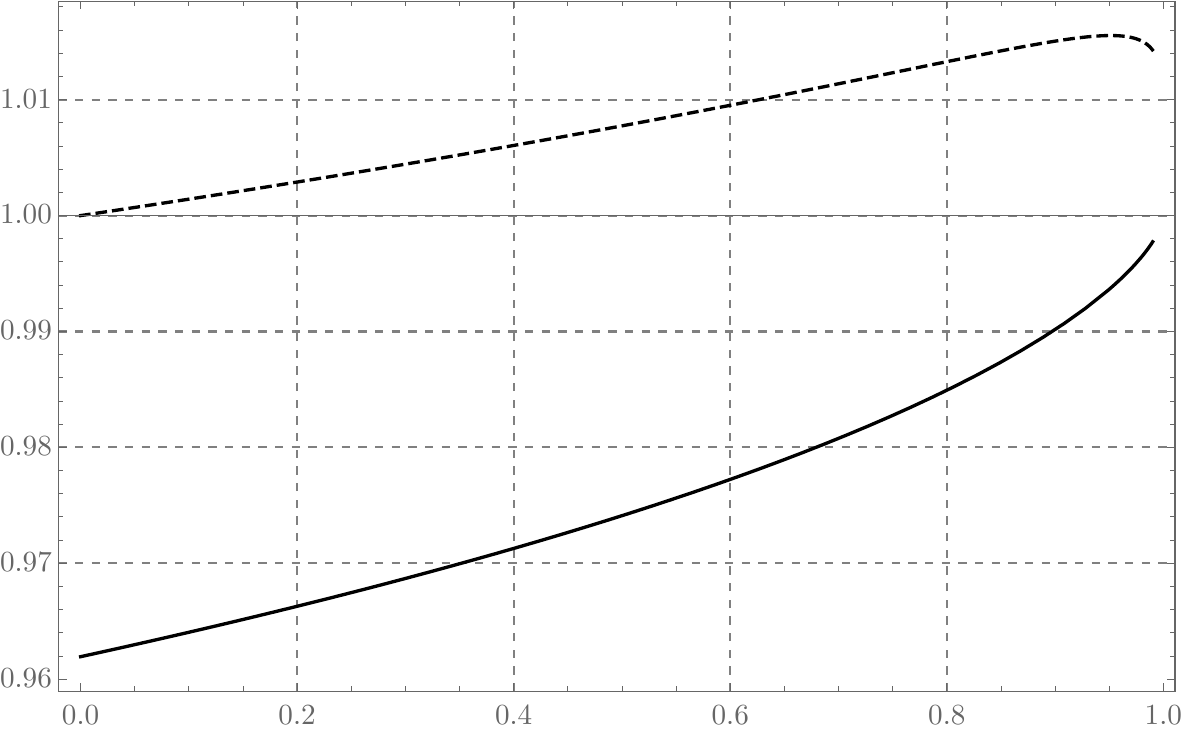} \\[2mm]
\includegraphics[width=0.49 \columnwidth]{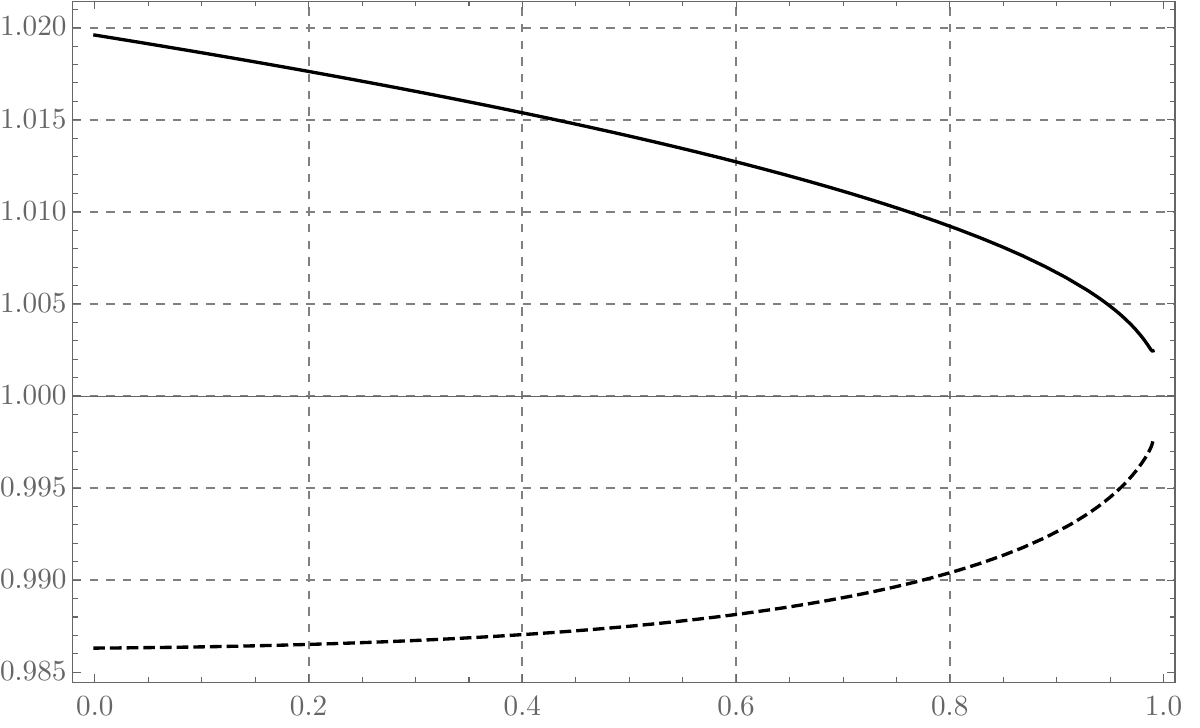}
\includegraphics[width=0.49 \columnwidth]{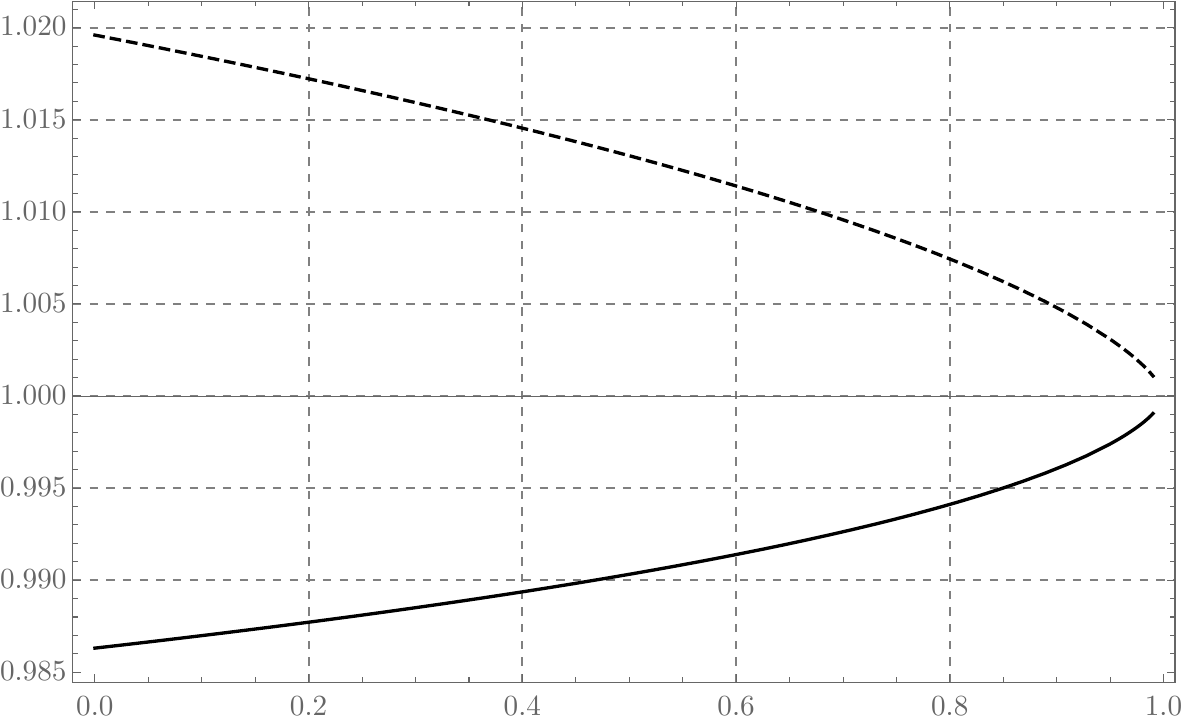}
\end{center}
\caption{Quality of the approximations on the parameters
of a negative binomial distribution with $p=0.1$ in terms of elementary functions 
(top; using Corollary~\ref{coro:invbin:negBinApprox}) and
standard special functions (bottom; using Theorem~\ref{thm:ibin:Lerch}). 
The plots on the left depict the ratios $\overline{\mu}/\mu$ (solid)
and  $\underline{\mu}/\mu$ (dashed), as functions of $q$,
where $\overline{\mu}$ and $\underline{\mu}$ are respectively
the upper and lower bound estimates on the mean $\mu$.
Similarly, the plots on the right depict the ratios $(\log \overline{y_0})/(\log y_0)$ (solid)
and  $(\log \underline{y_0})/(\log y_0)$ (dashed), as functions of $q$,
where $\overline{y_0}$ and $\underline{y_0}$ are respectively
the upper and lower bound estimates on the normalizing constant $y_0$.
}
\label{fig:BinEstimates}
\end{figure}

\subsection{Upper bounds on the capacity of a mean-limited Binomial channel}
\label{sec:BinMeanLim}

In this section, we consider the convolution channel
$\ch_\mu(\ber_p)$, where $\ber_p$ is the Bernoulli distribution
with mean $p$ (and $d=1-p$ is the deletion probability). 
The input to this channel is a non-negative
integer $X$ and the output is a sample from the binomial
distribution $\Bin_{X,p}$. In particular, we have $\E[Y]=p \E[X]$,
and thus the mean constraint implies that $\E[X]=\mu/p$.

\subsubsection{Capacity upper bound using the inverse binomial distribution}
\label{sec:del:general}

In order to upper bound the capacity of $\ch_\mu(\ber_p)$, 
we shall apply Theorem~\ref{thm:meanlimited} with an 
appropriate choice for the distribution of $Y$. Recall that,
for every $x \geq 0$, the distribution of $Y_x$ is binomial
($\Bin_{x,p}$)
with parameters $x$ (number of trials), $p$ (success probability),
and mean $\E[Y_x]=px$. 
That is, for all integers $y \geq 0$,
\begin{equation} \label{eqn:del:Yx}
\Pr[Y_x=y] = \binom{x}{y} p^y (1-p)^{x-y}.
\end{equation}
Intuitively, the constraints in \eqref{eqn:cvx1:dualF} 
should be generally satisfied with as small a gap as possible.
Indeed, according to the KKT conditions, equality must hold for every
point on the support of an optimal input distribution $X$.
Moreover, the optimal (and in fact, any feasible) 
$X$ must have infinite support, since otherwise,
for large enough $x$ the KL divergence $\KL{Y_x}{Y}$ for 
the corresponding output distribution $Y$ would be infinite, violating
\eqref{eqn:cvx1:dualF}.
One may hope that, in the ideal case, the optimum distribution
would satisfy all constraints with equality (and that would be
necessary if the optimum $X$ had full support). However,
in Remark~\ref{rem:bin:noFull}, we rule out this possibility.

In order to identify a feasible choice for $Y$, we first use 
convexity to show that an inverse binomial distribution,
as defined in \eqref{eqn:invbin}, is a feasible choice. 
Let $Y$ be distributed according to $\ibin_{p,q}$.
Note that the entropy of $Y_x$ can be written as
\begin{align} \label{eqn:del:HYx}
H(Y_x)&=x h(p)-\E\left[ \log \binom{x}{Y_x} \right],
\end{align}
where $h(p)$ is the binary entropy function.
Then, $\KL{Y_x}{Y}$ can be written as
\begin{align}
\KL{Y_x}{Y} &= \sum_{y=0}^\infty \Pr[Y_x = y] \log \frac{\Pr[Y_x = y]}{\Pr[Y=y]}
\nonumber \\
&\stackrel{\eqref{eqn:invbin}}{=} -H(Y_x)-\log y_0-\E[Y_x] \log q
+\E[Y_x] h(p)/p
-\E\left[ \log \binom{Y_x/p}{Y_x} \right] \nonumber \\
&\stackrel{\eqref{eqn:del:HYx}}{=}
\E \log \left[\binom{x}{Y_x}\Big/  \binom{Y_x/p}{Y_x}\right]
-xp \log q-\log y_0. \label{eqn:del:KLYxY}
\end{align}
Let $\overline{Y_x}:= x-Y_x$, and observe that the distribution
of $\overline{Y_x}$ is also binomial over $x$ trials, 
but with success probability $1-p$; i.e., $\Bin_{x,1-p}$.
Define 
\[
f(y) := \log \left[\binom{x}{y}\Big/  \binom{y/p}{y}\right].
\]
We can now write the expectation in \eqref{eqn:del:KLYxY} as
\[
\E[f(Y_x)]=\log x! - \E[\log Y_x] - \E[\log \overline{Y_x}] - \E\left[ \log \binom{x}{Y_x} \right].
\]
By the Bohr-Mollerup theorem (or simply positivity of the trigamma function), 
$-\log \Gamma(y+1)$ is a concave function of $y$.  
Moreover, we observe the following,
which implies that $f(y)$ is a concave function of $y$ 
(see Appendix~\ref{app:binom:monotone} for a proof):
\begin{claim}  \label{claim:binom:monotone}
The function 
\[f(y) := \log \binom{y}{p y} = \log \frac{\Gamma(y+1)}{\Gamma(py+1)\Gamma((1-p)y+1)},\]
defined for $p \in (0,1)$ and $x>0$, is completely monotone. That is,
for all integers $j=0,1, \ldots$, $(-1)^j f^{(j)}(y) > 0$ for all $y>0$. 
\end{claim}
\noindent Therefore, we can apply Jensen's inequality and deduce that
\[
\E[f(Y_x)] \leq f(\E[Y_x]) = f(px)=\log 1 = 0,
\]
and thus, plugging this result into \eqref{eqn:del:KLYxY}, that
for $x=0,1, \ldots,$
\begin{equation}
\KL{Y_x}{Y} \leq -xp \log q-\log y_0. \label{eqn:del:KLYxY:line}
\end{equation}
Now, Theorem~\ref{thm:meanlimited} can be applied with the above
choice for $Y$, which proves the following:

\begin{thm} \label{thm:BinChUpper}
Let $p \in (0,1)$ and $q\in(0,1)$ be given parameters and $Y$ be a random
variable distributed according to the inverse binomial distribution \eqref{eqn:invbin}
for an appropriate normalizing constant $y_0$ and mean $\mu=\E[Y]$. 
Then, the capacity of the mean-limited binomial channel 
$\ch_\mu(\ber_p)$ satisfies 
\begin{equation} \label{eqn:BinMeanLimCapUpperA}
\capa(\ch_\mu(\ber_p)) \leq -\mu \log q-\log y_0.
\end{equation}
\end{thm}
If desired, the right hand side of \eqref{eqn:BinMeanLimCapUpperA}
can in turn be upper bounded by elementary functions
using the estimates provided by Corollary~\ref{coro:invbin:negBinApprox},
or by standard special functions using Theorem~\ref{thm:ibin:Lerch}.

\subsubsection{Improving the capacity upper bound using the truncation technique}
\label{sec:del:general:trunc}

The choice of the inverse binomial distribution for
the random variable $Y$ in Theorem~\ref{thm:BinChUpper}
already achieves strong capacity upper bounds for
the mean-limited binomial channel (and consequently, the deletion
channel with deletion probability $1-d$, as shown
in Section~\ref{sec:del:derivation} and Figure~\ref{fig:DelPlot}) for all $p \in (0,1)$.
However, although $Y$ achieves the feasibility requirement
\eqref{eqn:del:KLYxY:line} with equality at $x=0$, 
for large $x$ the inequality remains
strict, and, as Figure~\ref{fig:BinRpX} depicts, exhibits a constant asymptotic gap 
of $1/2$ to the linear term on
the right hand side (a similar constant gap exists for the Poisson
channel with the choice of \eqref{eqn:PoiYfirst} for $Y$,
which is eliminated by using the digamma distribution \eqref{eqn:PoiYre}, leading to the
improved bounds). 
In this section, we obtain improved results by 
implementing the truncation technique of Section~\ref{sec:Poi:truncation}
for binomially distributed random variables.

We start from the following integral representation for the log-gamma
function\footnote{A simple proof of this identity is to observe that, letting $f(z)$ denote the right hand side 
expression in \eqref{eqn:LGint}, $f(0)=0$, and furthermore, $f(z+1)-f(z) = \log(1+z)$ (which can, in turn,
be verified by taking the derivative of the integral expression of the difference in $z$, and
verify that it is indeed equal to $1/(1+z)$). Since the expression defining $f(z)$ is convex in $z$, 
and satisfies the same
recursion as the log-gamma function, it must be equal to the log-gamma function by the
Bohr-Mollerup theorem.} \cite{ref:Che17}:
\begin{align}
\log \Gamma(1+z)&=\int_0^1 \frac{1-tz-(1-t)^z}{t \log(1-t)}\,dt \label{eqn:LGint} \\
&=\int_0^1 \sum_{j=2}^\infty \frac{ \binom{z}{j} (-t)^{j-1} }{\log(1-t)}\,dt. \label{eqn:LGsum}
\end{align}
For an $x>0$, let $Y_x$ be a sample from $\Bin_{x,p};$ i.e., a binomial random
variable over $x$ trials with success probability $p$. Recall that 
\[
\Pr[Y_x=y]=\binom{x}{y} p^y (1-p)^{x-y},
\]
and that the factorial
moments of $Y_x$ are given by
\begin{equation} \label{eqn:bin:factorial}
\E\left[\binom{Y_x}{j}\right] = \binom{x}{j} p^j.
\end{equation}
Define
\begin{align}
\Ex_p(x) &:= \E[\log(Y_x!)] \label{eqn:bin:Enp} \\
&= \sum_{y=0}^\infty \binom{x}{y} p^y (1-p)^{x-y} \log y! \nonumber \\
&\stackrel{\eqref{eqn:LGsum}}{=} \E\left[ \int_0^1 \sum_{j=2}^\infty \frac{ \binom{Y_x}{j} (-t)^{j-1} }{\log(1-t)}\,dt \right] \nonumber \\
&\stackrel{\eqref{eqn:bin:factorial}}{=}
 \int_0^1 \sum_{j=2}^\infty \frac{ \binom{x}{j} p^j (-t)^{j-1} }{\log(1-t)}\,dt \nonumber \\
 &=
 \int_0^1 \frac{ 1-p t x-(1-p t)^x}{t \log(1-t)}\,dt 
 \label{eqn:bin:Enp:expand}
\end{align}

\noindent The asymptotic growth of the function $\cE_p$ is derived below.
\begin{claim} \label{claim:Enp:asymptotic}
For large $x>0$ and $p \in (0,1]$, we have $\Ex_p(x)=x p (\log(x p)-1)+\frac{1}{2} \log(x p) \pm O(1)$.
\end{claim}

\begin{proof}
By Stirling's approximation, we have
\[
\Ex_1(x)=\log \Gamma(1+x)=x \log x-x+\log \sqrt{2\pi x} + o(1).
\]
This particularly shows the claim for $p=1$.
Now, using the integral expression \eqref{eqn:LGint}  for the log-gamma function, 
and  \eqref{eqn:bin:Enp:expand},
we may write
\begin{align*}
\Ex_p(x)-\log \Gamma(1+p x) &= 
\int_0^1 \left( \frac{1-tpx-(1-tp)^x}{t \log(1-t)}- \frac{1-tpx-(1-t)^{p x}}{t \log(1-t)}\right)dt  \\
&= \int_0^1 \frac{(1-t)^{px}-(1-tp)^x}{t \log(1-t)}\,dt 
\end{align*}
By decomposing the integration interval in the last expression
over small $t$ and the remaining interval
(at which the integrand exponentially vanishes as $x$ grows), and estimating
the integrand for small $t$ by a series expansion, it is seen that the
difference is $O(1)$. The claim follows by applying Stirling's 
approximation on $\log \Gamma(1+px)$.
\end{proof}

Towards designing a distribution $Y$ over positive integers that
satisfies \eqref{eqn:cvx1:dualF} as tightly as possible (and
particularly with equality at $x=0$ and a sharply vanishing gap as $x$ grows),
suppose that the
probability mass function for the distribution of $Y$ (that we wish to design)
is given by a generalization of \eqref{eqn:invbin} as follows:
\begin{equation} \label{eqn:bin:y:general}
p(y) := \Pr[Y=y] = y_0 \frac{q^y \exp(g(y) -y h(p)/p)}{y!},
\end{equation}
for a function $g\colon \NN\to \R$ and appropriate normalizing 
constant $y_0$. Note that, in order to have a well-defined 
distribution, $g(y)$ can asymptotically
be at most $y \log y+O(y)$.
We may write
\begin{align}
\KL{Y_x}{Y} &= \sum_{y=0}^\infty \Pr[Y_x=y] (\log(\Pr[Y_x=y])-\log p(y)) \nonumber\\
&= -x h(p)+\E\left[\log \binom{x}{Y_x}-\log p(Y_x)\right]\nonumber \\
&\stackrel{\eqref{eqn:bin:y:general}}{=} -x h(p)+\log x!-\E[\log (x-Y_x)]-\E[\log g(Y_x)]-\E[Y_x](\log q-h(p)/p)-\log y_0 \nonumber\\
&\stackrel{\eqref{eqn:bin:Enp}}{=}\log x!-\Ex_{1-p}(x)-\E[\log g(Y_x)]-xp\log q-\log y_0.
\label{eqn:bin:KL:general}
\end{align}
Therefore, satisfying \eqref{eqn:cvx1:dualF} is equivalent to having,
for some values $a$ and $b$ possibly depending on $p$, and for all integers
$x >0$,
\begin{equation} \label{eqn:bin:KL:ab}
\E[g(Y_x)] \geq \log x!-\Ex_{1-p}(x)+a x +b.
\end{equation}

Given a parameter $\eps \in (0,1]$, 
Let $f_\eps$ be a function such that, for every integer $x \geq 0$,
\begin{equation}
\E[f(Y_x)] = \Ex_\eps(x).
\end{equation}
Such a function exists and can explicitly (and uniquely) be written, using 
\eqref{eqn:bin:factorial} and \eqref{eqn:bin:Enp:expand},
as
\begin{align}
f_\eps(y) &=  \int_0^1 \sum_{j=2}^\infty \frac{ \binom{y}{j} (\eps/p)^j (-t)^{j-1} }{\log(1-t)}\,dt 
\label{eqn:fey:a} \\
&= \sum_{k=0}^y \binom{y}{k} (1/p)^k (1-1/p)^{y-k} \Ex_\eps(k)  \label{eqn:fey:b} \\
&= \Ex_{\eps/p}(y), \label{eqn:fey:c}
\end{align}
where the second equality can be seen by observing that the equation for
factorial moments \eqref{eqn:bin:factorial} of a binomial distribution still remains syntactically valid even
if the ``success probability'' $p$ defining the distribution is greater than $1$.
Note that the summation in \eqref{eqn:fey:a} is finite for any non-negative
integer $y$, and thus the series representing $f_\eps(y)$ is always convergent
to a finite value. 

When $\eps/p < 1$, the terms in \eqref{eqn:fey:a} exponentially vanish in magnitude
as $j$ grows and $f_\eps(y)$ maintains a manageable asymptotic growth (as
shown by Claim~\ref{claim:Enp:asymptotic}). However,
for $\eps > p$, the value of $f_\eps(y)$ exponentially grows in $\eps/p$.
In this case, we will use the truncation technique of 
Section~\ref{sec:Poi:truncation} to modify $f_\eps(y)$ to a function
with manageable asymptotic growth rate whose expectation under $Y_x$ still provides
a very accurate estimate on $\Ex_\eps(x)$. Towards this goal, let $\eps \in (0,1]$,
and consider the following truncation of the integral expression \eqref{eqn:bin:Enp:expand} for $\Ex_{1/\eps}(y)$:
\begin{align}
\Lam_\eps(y) &:= \int_0^1 \frac{ 1- ty-(1-t)^y }{t \log(1-\eps t)}\,dt \label{eqn:Lam} \\
&= \int_0^\eps \frac{ 1- ty/\eps-(1-t/\eps)^y }{t \log(1-t)}\,dt \label{eqn:Lam:b} \\
&= \int_0^1 \sum_{j=2}^\infty \frac{ \binom{y}{j} (-t)^{j-1} }{\log(1-\eps t)}\,dt 
\label{eqn:Lam:a}  \\
&= \int_0^\eps \sum_{j=2}^\infty \frac{ \binom{y}{j} (-t/\eps)^{j-1} }{\eps \log(1-t)}\,dt.
\nonumber 
\end{align}
Note that $\Lam_1(y) = \log \Gamma(1+y)$.
For any fixed $\eps$ and $t<\eps$, the term $(1-t/\eps)^y$ in the integrand
of \eqref{eqn:Lam:b} exponentially vanishes in $y$, and the aim is to show that
the error caused by the truncation exponentially converges to a linear expression in $y$ as $y$ grows.
An important property of the function $\Lam_\eps$ is its expected value with respect
to a binomial distribution; namely, $\E[\Lam_\eps(Y_x)]$.
Using \eqref{eqn:Lam:a} and  \eqref{eqn:bin:factorial}, we have
\begin{align}
\E[\Lam_\eps(Y_x)] &= \int_0^1 \sum_{j=2}^\infty \frac{ \binom{x}{j} p^j (-t)^{j-1} }{\log(1-\eps t)}\,dt \nonumber \\
&= \int_0^1 \sum_{j=2}^\infty \frac{ \binom{x}{j} (p/\eps)^j (-\eps t)^{j-1} }{\log(1- \eps t)}\,d(\eps t) \nonumber \\
&= \int_0^\eps \sum_{j=2}^\infty \frac{ \binom{x}{j} (p/\eps)^j (-t)^{j-1} }{\log(1-t)}\,dt\nonumber  \\
& \stackrel{\eqref{eqn:bin:Enp:expand}}{=} \Ex_{p/\eps}(x)-
\int_\eps^1 \sum_{j=2}^\infty \frac{ \binom{x}{j} (p/\eps)^j (-t)^{j-1} }{\log(1-t)}\,dt \nonumber \\
&= \Ex_{p/\eps}(x)+\int_\eps^1 \frac{ (1-tp/\eps)^x+ xtp/\eps-1 }{t \log(1-t)}\,dt                                                                                                        
\nonumber \\
&= \Ex_{p/\eps}(x)+\int_\eps^1 \frac{ (1-tp/\eps)^x+ xtp/\eps-1 }{t \log(1-t)}\,dt
\nonumber \\ 
&= \Ex_{p/\eps}(x)+xp\li(1-\eps)/\eps-\eta(1-\eps)+
\int_\eps^1 \frac{ (1-tp/\eps)^x }{t \log(1-t)}\,dt. \label{eqn:ELam} 
\end{align}
where $\li(\cdot)$ the logarithmic integral function
\begin{equation} \label{eqn:Li:def}
\li(z) = \int_0^z \frac{dt}{\log t},
\end{equation}
and we have defined
\begin{equation} \label{eqn:eta}
\eta(z) = \int_0^z \frac{dt}{(1-t)\log t} = \sum_{j=1}^\infty \li(z^j),
\end{equation}
If $p \leq \eps \leq 1$, we notice that the integrand
in the residual term in \eqref{eqn:ELam} is non-positive,
and thus, in this case, 
\[
\eqref{eqn:ELam} \leq \Ex_{p/\eps}(x)+xp\li(1-\eps)/\eps -\eta(1-\eps). 
\]
On the other hand, the integrand is at least $(1-p)^x/(t \log(1-t))$
(note that $\log(1-t) < 0$). Therefore, we may also write, for $p \leq \eps \leq 1$,
\begin{align}
 \nonumber
\eqref{eqn:ELam} &\geq \Ex_{p/\eps}(x)+xp\li(1-\eps)/\eps -\eta(1-\eps)
+\int_\eps^1 \frac{ (1-p)^x }{t \log(1-t)}\,dt \\
&= \Ex_{p/\eps}(x)+xp\li(1-\eps)/\eps +((1-p)^x-1)\eta(1-\eps).
\label{eqn:ELam:lower}
\end{align}

\noindent The asymptotic growth rate of the function $\Lam_\eps$ is described
by the following result:
\begin{claim} \label{claim:Lam}
For large $x$, and a fixed $\eps \in (0,1]$, 
the function $\Lam_\eps$ defined in \eqref{eqn:Lam}
satisfies
\[
\Lam_\eps(x) = \frac{x}{\eps}\left(\log (x/\eps)+\li(1-\eps)-1\right) + \frac{1}{2} \log(x/\eps) \pm O(1). 
\]
\end{claim}
\begin{proof}
The proof is quite similar to that of  Claim~\ref{claim:Enp:asymptotic}.
When $\eps=1$, we have $\Lam_1(x) = \log \Gamma(1+x)$ and the claim follows
from Stirling's approximation. In general,
let $\Delta(x) := \Lam_\eps(x)-\log \Gamma(1+x/\eps)-x \li(1-\eps)/\eps$. Using
Stirling's approximation, it suffices to show that $|\Delta(x)| = O(1)$.
Let us write
\begin{align*}
\log \Gamma(1+x/\eps)= \Lam_1(x/\eps)&=\int_0^1 \frac{ 1- tx/\eps-(1-t)^{x/\eps} }{t \log(1-t)}\,dt \\
&=
\int_0^\eps \frac{ 1- tx/\eps-(1-t)^{x/\eps} }{t \log(1-t)}\,dt+
\int_\eps^1 \frac{ 1- tx/\eps-(1-t)^{x/\eps} }{t \log(1-t)}\,dt \\
&=
\int_0^\eps \frac{ 1- tx/\eps-(1-t)^{x/\eps} }{t \log(1-t)}\,dt-
\frac{x \li(1-\eps)}{\eps}-\int_\eps^1 \frac{(1-t)^{x/\eps} }{t \log(1-t)}\,dt \\
&=
\int_0^\eps \frac{ 1- tx/\eps-(1-t)^{x/\eps} }{t \log(1-t)}\,dt-
\frac{x \li(1-\eps)}{\eps}+o(1),
\end{align*}
where for the last step we have used the fact that $(1-t)^{x/\eps}$
goes down to zero as $x$ grows. 
Thus, we have
\begin{align*}
\Delta(x) &\stackrel{\eqref{eqn:Lam}}{=} \int_0^1 \frac{ 1- tx-(1-t)^{x} }{t \log(1-\eps t)}\,dt
-\int_0^\eps \frac{ 1- tx/\eps-(1-t)^{x/\eps} }{t \log(1-t)}\,dt-o(1) \\
&= \int_0^\eps \frac{ (1-t)^{x/\eps}-(1-t/\eps)^{x} }{t \log(1-t)}\,dt-o(1),
\end{align*}
where in the second step we have used a change of variables. 
 The last integral can be shown to be upper bounded by a constant by
 decomposing the integral over small $t$ (and using a series estimate
 of the integrand) and the remaining interval (over which the integrand
 tends to zero as $x$ grows). \Mnote{expand on this?}
  This completes the proof. 
\end{proof}

Now, we are ready to define an appropriate expression for the distribution
of $Y$ that satisfies \eqref{eqn:bin:KL:ab} for choices of $a$ and $b$
possibly depending on $p$. Before doing so, let $H(\Bin_{x,p})$
denote the entropy of the binomial distribution with parameters $x$
and $p$ and consider the following
function defined over the integers $y \geq 0$:
\begin{align} 
\lam_p(y) &= \sum_{k=0}^y \binom{y}{k} (1/p)^k (1-1/p)^{y-k} H(\Bin_{k,p})
\nonumber \\
&= \sum_{k=0}^y \binom{y}{k} (1/p)^k (1-1/p)^{y-k} (y h(p) - \log k!
+ \Ex_p(k) + \Ex_{1-p}(k)) \nonumber \\
&= yh(p)/p + \log y! -\Ex_{1/p}(y) + \Ex_{1/p-1}(y), \label{eqn:lamP}
\end{align}
where we have syntactically extended the definition of the function $\Ex_p$
in \eqref{eqn:bin:Enp} to $p > 1$, used the fact that the expressions
for the moments (in particular, the
mean) of the binomial distribution, regarded syntactically,
remain true even if $p>1$, as well as the following observation:

\begin{prop} \label{prop:Enp}
The functions $\Ex_p$ defined in \eqref{eqn:bin:Enp} satisfy,
for all integers $y \geq 0$ and all $q>0$, 
\[
\sum_{k=0}^y \binom{y}{k} q^k (1-q)^{n-k} \Ex_p(k) = \Ex_{pq}(y).
\] 
\end{prop}

\begin{proof}
This is a direct consequence of the fact that, letting
$Y \sim \Bin_{y,q}$ and $K \sim \Bin_{Y,p}$, 
we have $K \sim \Bin_{y,pq}$, and that this is true in
a syntactical sense even if $p$ and $q$ are allowed to
be larger than one. In this case, the left hand side
is $\E[\log K!]$ which, noting that the distribution
of $K$ is binomial (namely, $\Bin_{y,pq}$), 
can be rewritten as $\Ex_{pq}(y)$.
\end{proof}

Now, recall the binomially distributed random variable $Y_x$ and
observe that, using \eqref{eqn:lamP} and Proposition~\ref{prop:Enp}
we may write
\[
\E[\lam_p(Y_x)] = xh(p) + \Ex_p(x) - \Ex_1(x) + \Ex_{1-p}(x)=H(\Bin_{x,p}),
\]
and therefore, if there is a normalizing constant $y_0$ that, for some $q>0$, makes the
following a legitimate probability mass function over the non-negative integers,
\begin{equation} \label{eqn:bin:pY:alleged}
\Pr[Y=y] = y_0 q^y \exp(-\lam_p(y)),
\end{equation}
the distribution of $Y$ would then satisfy the KKT conditions 
\eqref{eqn:cvx1:dualF}  with equality for all $x$. However, as we
discussed before, for all $p \in (0,1]$, the value of $\lam_p(y)$
(in particular, the value of $\Ex_{1/p}(y)$) exponentially grows in $y$.
Therefore, no value of $q$ and $y_0$ can normalize the above
distribution to a legitimate one. To address this, we employ the
truncation technique that we developed in this section to 
modify $\lam_p(y)$ to a function that exhibits a controllable growth 
rate, but nevertheless, approximates the behavior of $\lam_p(y)$
and satisfies the KKT conditions sharply.
Towards this goal, we distinguish two cases; namely, when $p \in (0,1/2)$
and when $p \in [1/2,1]$.

\begin{remark} \label{rem:bin:noFull}
As was the case for the mean-limited Poisson channel (see Remark~\ref{rem:Poi:noFull}),
for the mean-limited binomial channel with input distribution $X$ and output distribution
$Y$, we may observe that the capacity achieving distribution for $Y$ (and subsequently,
the capacity achieving distribution for $X$) must have
infinite support. Otherwise, for sufficiently large $x$, the quantity 
$\KL{Y_x}{Y}$ would be infinite, violating the KKT conditions \eqref{eqn:cvx1:dualF}.
 Moreover, we observe that the function $\lam_p(y)$ is
 the unique function that satisfies $\E[\lam_p(Y_x)]=H(\Bin_{x,p})$ for all
 non-negative integers $x$ (due to the triangular nature of the corresponding
 system of linear equations, the solution to this functional equality must be unique).
 In turn, this implies that, up to a change in the linear term, the expression
  \eqref{eqn:bin:pY:alleged} for the distribution of $Y$
  is the unique choice that would satisfy the KKT conditions \eqref{eqn:cvx1:dualF}
  with equality for all $x$. However, as this alleged distribution cannot be 
  normalized, and since the KKT conditions must be satisfied with equality over the entire
  support of any optimal $X$, it follows that the capacity achieving distribution for $X$
  cannot have full support over non-negative integers, even though its support must be infinite. As was
  discussed for the Poisson case in Remark~\ref{rem:Poi:noFull}, this makes
  intuitive sense from a coding perspective. Intuitively, for every nonzero $x$ on the support
  of $X$, the next higher integer supported by $X$ should be about $x+\Omega(\sqrt{x(1-p)/p})$ 
  (and the first nonzero $x$ should be $\Omega(1/p)$).
  This ensures that the channel outputs corresponding to different elements 
  on the support of $X$ are sufficiently spread out (as dictated by the corresponding standard
  deviations) to avoid substantial confusions on the decoder side.
\end{remark}

\subsubsection*{Case 1: Truncation when $p \in [1/2,1]$}

When $p \geq 1/2$, the only exponentially growing term in \eqref{eqn:lamP}
is $\Ex_{1/p}(y)$, which we truncate to $\Lam_p(y)$ as defined in \eqref{eqn:Lam}.
Consequently, we define the function
\begin{equation} \label{eqn:gpy:largeP}
g_p(y) := \Lam_p(y)-\Ex_{1/p-1}(y)-y \li(1-p)/p+\eta(1-p),
\end{equation}
where $\eta$ is the function in \eqref{eqn:eta} and $\Ex_{1/p-1}$ is
defined according to  \eqref{eqn:bin:Enp}. We may write,
using \eqref{eqn:ELam} and Proposition~\ref{prop:Enp},
\begin{align*}
\E[g_p(Y_x)] &= \Ex_{1}(x)+x\li(1-p) -\eta(1-p)
-\Ex_{1-p}(x) - x\li(1-p) + \eta(1-p) + \int_{p}^1 \frac{(1-t)^x dt}{t \log(1-t)} \\
&= \log x!-\Ex_{1-p}(x) + \int_{p}^1 \frac{(1-t)^x dt}{t \log(1-t)},
\end{align*}
where we notice that the ``error term'' 
$\int_{p}^1 (1-t)^x dt/(t \log(1-t))$ is an incomplete
variation (or tail) of the integral
\eqref{eqn:LGsum} defining the log-gamma function, 
and exponentially vanishes as $x$ grows.

Recall the Kronecker delta function
\[
\delta(y) := \begin{cases}
1 & y=0, \\
0 & y \neq 0,
\end{cases}
\]
and note that $\E[\delta(Y_x)] = (1-p)^x$. 
Now, define
\begin{equation} \label{eqn:bin:g:plarge}
g(y) := g_p(y)-\eta(1-p) \delta(y) = \begin{cases}
0 & y=0 \\
g_p(y) & y > 0.
\end{cases}
\end{equation}
By combining the above expectation and \eqref{eqn:ELam:lower},
we see that
\begin{equation} \label{eqn:bin:EgYx:pLarge}
\E[g(Y_x)] = \log x!
-\Ex_{1-p}(x) + R_p(x) \geq \log x!-\Ex_{1-p}(x),
\end{equation}
where we have defined
\begin{equation} \label{eqn:RpX:pLarge}
R_p(x) := \int_{p}^1 \frac{(1-t)^x-(1-p)^x}{t \log(1-t)}\, dt \geq 0.
\end{equation}
Note that the error $R_p(x)$ is zero for $p=1$, and is always 
non-negative since the integrand is non-negative for all $t \in [p,1]$.
Using the above choice for $g(y)$ in the probability mass
function of $Y$ in \eqref{eqn:bin:y:general} results in
\[
\KL{Y_x}{Y} \stackrel{\eqref{eqn:bin:KL:general}}{=}
-xp \log q - \log y_0 - R_p(x) \leq -xp \log q - \log y_0,
\]
thus satisfying the KKT conditions \eqref{eqn:cvx1:dualF}
for all $x \geq 0$ and with the choices $\nu_1 = -\log q$
and $\nu_0 = -\log y_0$.

The asymptotic behavior of the above $g(y)$ can be deduced
from Claims \ref{claim:Enp:asymptotic}~and~\ref{claim:Lam}. Namely, we have
\begin{align*} 
g(y) &= \frac{y}{p}(\log(y/p)+\li(1-p)-1)-\frac{y(1-p)}{p} \left(\log\left(\frac{y(1-p)}{p}\right)-1\right)-\frac{y}{p}\li(1-p) \\&+
\frac{1}{2} \log(y/p) - \frac{1}{2} \log\left(\frac{y(1-p)}{p}\right) \pm O(1) \\
&= y h(p)/p+y \log(y/e)\pm O_p(1),
\end{align*}
which, combined with Stirling's approximation $\log y! = y \log(y/e)+\log \sqrt{2\pi y}+o(1)$, 
implies that the probability mass
function of $Y$ in \eqref{eqn:bin:y:general} asymptotically
behaves as $q^y/\sqrt{y}$ and can thus be normalized to a legitimate
distribution precisely when $q \in (0,1)$.
Consequently, as was the case for the inverse binomial
distribution, the expectation of the distribution can be made
arbitrarily small as $q\to 0$ and arbitrarily large as $q\to 1$.
 
\subsubsection*{Case 2: Truncation when $p \in (0,1/2]$}

For values of $p$ below $1/2$, both $\Ex_{1/p}(y)$ and $\Ex_{1/p-1}$
in \eqref{eqn:lamP} exponentially grow and must both be truncated.
Accordingly, we modify the function $g_p(y)$ in \eqref{eqn:gpy:largeP}
to, recalling the functions $\Lam_\eps$ from \eqref{eqn:Lam} and 
$\eta$ from \eqref{eqn:eta}, the following:
\begin{align} 
g_p(y) &:= \Lam_p(y)-\Lam_{p/(1-p)}(y)+
\frac{y}{p}\left( (1-p)\li\left(\frac{1-2p}{1-p}\right)-\li(1-p) \right)+
\eta(1-p)-\eta\left(\frac{1-2p}{1-p}\right) \nonumber \\
&= \Lam_p(y)-\Lam_{p/(1-p)}(y)+
\frac{y}{p}\left( (1-p)\li\left(\frac{1-2p}{1-p}\right)-\li(1-p) \right)+
\int_{p}^{\frac{p}{1-p}} \frac{dt}{t \log(1-t)}, \label{eqn:gpy:smallP}
\end{align}
and, similar to the previous case, adjust it with a Kronecker delta as follows
\begin{equation} \label{eqn:bin:g:psmall}
g(y) = g_p(y)  - \int_{p}^{\frac{p}{1-p}} \frac{\delta(y) dt}{t \log(1-t)}=
\begin{cases}
0 & y=0, \\
g_p(y) & y > 0.
\end{cases}
\end{equation}

\begin{remark}
Note that $\li(0)=\eta(0)=0$, and $\Lam_1(y)=\Ex_{1}(y)=\log \Gamma(1+y)$. Therefore,
we see that for the boundary case $p=1/2$, the expressions given by  \eqref{eqn:gpy:largeP}
and  \eqref{eqn:gpy:smallP} coincide.
\end{remark}

A remarkable property of the distribution defined with respect to the above choice
of $g$ in \eqref{eqn:bin:g:psmall} is that, in the limit $p \to 0$, 
the distribution converges to the digamma distribution \eqref{eqn:PoiYre} that we designed for the 
mean-limited Poisson channel. Therefore, the distribution designed using
the truncation technique in
this section is indeed the right generalization of what we constructed for the Poisson case
to the more general setting of the binomial channel. This is formalized below.

\begin{prop} \label{prop:smallp:approx}
Consider the distribution \eqref{eqn:bin:y:general}, where $g(y)$ is
defined according to \eqref{eqn:bin:g:psmall}. Then,
\[
\lim_{p \to 0} \Pr[Y=y]=y_0 \frac{\exp(y \psi(y)) (q/e)^y}{y!},
\]
where $\psi(\cdot)$ is the digamma function. That is, 
the distribution converges, pointwise, to the digamma distribution \eqref{eqn:PoiYre}.
\end{prop}

\begin{proof}
Recall, from \eqref{eqn:Lam:a}, the series expansion
\[
\Lam_\eps(y)=\int_0^1 \sum_{j=2}^\infty \frac{ \binom{y}{j} (-t)^{j-1} }{\log(1-\eps t)}\,dt.
\]
Let $\Lam_0(y) := \lim_{\eps \to 0} \eps \Lam_\eps(y)$.
In the limiting case $\eps \to 0$, the denominator of the above series can be estimated by $\eps t$, so that
we have,
\begin{align}
\Lam_0(y) &= \int_0^1 \sum_{j=2}^\infty \binom{y}{j} (-t)^{j-2}\,dt \nonumber \\
&= \sum_{j=2}^\infty \binom{y}{j} \frac{(-1)^{j}}{j-1} \nonumber \\
&= -y \sum_{j=1}^\infty \binom{y-1}{j} \frac{(-1)^{j}}{j(j+1)} \nonumber \\
&= -y \left(\sum_{j=1}^\infty \binom{y-1}{j} \frac{(-1)^{j}}{j}-
\sum_{j=1}^\infty \binom{y-1}{j} \frac{(-1)^{j}}{j+1}\right) \nonumber \\
&\stackrel{\eqref{eqn:digamma:newton}}{=} 
-y \left(-\gamma-\psi(y)-
\sum_{j=1}^\infty \binom{y-1}{j} \frac{(-1)^{j}}{j+1}\right) \nonumber \\
&= y \psi(y)+\gamma y +\sum_{j=1}^\infty \binom{y}{j+1} (-1)^{j} \nonumber \\
&= y \psi(y)+(\gamma-1) y+1, \label{eqn:lam:limit:a}
\end{align}
where for the last equality we have used $\sum_{j=0}^\infty \binom{y}{j} (-1)^j=(1-1)^y=0$.
From the expansion of the logarithmic integral (equivalently, exponential integral
combined with logarithm, see \cite[p.\ 229]{ref:Ab72}), which is,
\[
\li(1-\eps)=\gamma+\log \eps -\eps/2 - \eps^2/24 - O(\eps^3),
\]
we may write
\begin{align}
(1-p)\li\left(\frac{1-2p}{1-p}\right)-\li(1-p)=-\gamma p +(1-p) \log(p/(1-p))-\log p
-O(p)=-\gamma p+h(p)-O(p^3). \label{eqn:lam:limit:b}
\end{align}
Furthermore, by approximating $1/(t \log(1-t))$ by $-1/t^2$ for small $t$, we can deduce
that
\begin{align}
\lim_{p \to 0}\int_{p}^{\frac{p}{1-p}} \frac{dt}{t \log(1-t)}=
-\lim_{p \to 0}\int_{p}^{\frac{p}{1-p}} \frac{dt}{t^2}=-1. \label{eqn:lam:limit:c}
\end{align}
By plugging the above results \eqref{eqn:lam:limit:a}, \eqref{eqn:lam:limit:b},
and \eqref{eqn:lam:limit:c} into \eqref{eqn:gpy:smallP}, we may now see that,
for any $y>0$,
\begin{equation} \label{eqn:lam:limit:d}
g(y) = g_p(y)= y (\psi(y)+ h(p)/p - 1 \pm O(p)).
\end{equation}
Recalling from \eqref{eqn:bin:y:general}, that
\[
\Pr[Y=y] = y_0 \frac{q^y \exp(g(y) -y h(p)/p)}{y!},
\]
we can use \eqref{eqn:lam:limit:d} to write
\[
\lim_{p \to 0} \Pr[Y=y]=y_0 \frac{\exp(y \psi(y)) (q/e)^y}{y!},
\]
as claimed.
\end{proof}

An important property of the function $g$ that we need to use is 
its expectation with respect to a binomial distribution.
This can be expressed, using \eqref{eqn:ELam}, as
\begin{align}
\E[g(Y_x)] &= \Ex_{1}(x)+x \li(1-p)-\eta(1-p)+\int_{p}^1 \frac{(1-t)^x - (1-p)^x}{t \log(1-t)}\, dt \nonumber \\
&- \Ex_{1-p}(x)-x(1-p)\li\left(\frac{1-2p}{1-p}\right)+\eta\left(\frac{1-2p}{1-p}\right)
-\int_{\frac{p}{1-p}}^1 \frac{(1-(1-p)t)^x - (1-p)^x}{t \log(1-t)}\, dt  \nonumber \\
&+x\left( (1-p)\li\left(\frac{1-2p}{1-p}\right)-\li(1-p) \right)
+\eta(1-p)-\eta\left(\frac{1-2p}{1-p}\right)  \nonumber \\
&= \log x!-\Ex_{1-p}(x)+R_p(x), \label{eqn:bin:EgYx:pSmall}
\end{align}
where we have defined
\begin{align}
R_p(x) &:= \int_{p}^1 \frac{(1-t)^x - (1-p)^x}{t \log(1-t)}\, dt
-\int_{\frac{p}{1-p}}^1 \frac{(1-(1-p)t)^x - (1-p)^x}{t \log(1-t)}\, dt  \nonumber  \\
&= \int_{p}^1 \frac{(1-t)^x}{t \log(1-t)}\, dt
-\int_{\frac{p}{1-p}}^1 \frac{(1-(1-p)t)^x}{t \log(1-t)}\, dt
 -(1-p)^x \int_{p}^{\frac{p}{1-p}} \frac{dt}{t \log(1-t)}. \label{eqn:RpX:pSmall}
\end{align}  

\begin{figure}[t!]
\begin{center}
\includegraphics[width=0.49 \columnwidth]{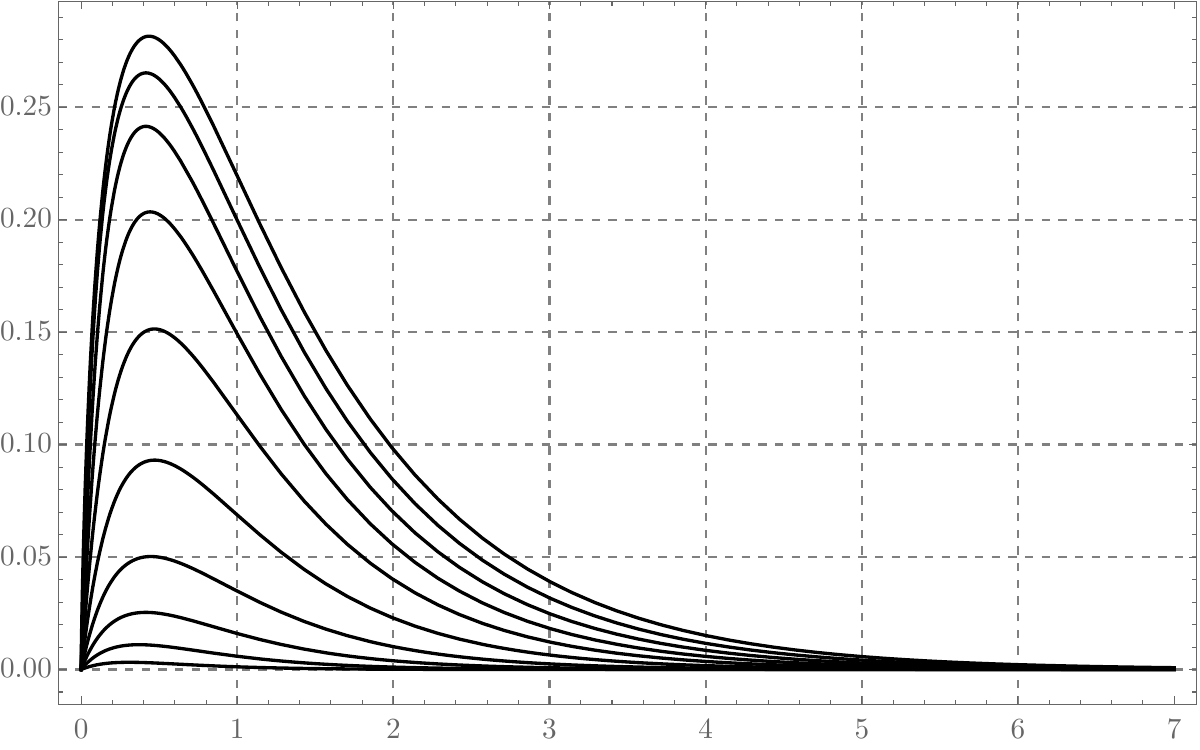}
\includegraphics[width=0.49 \columnwidth]{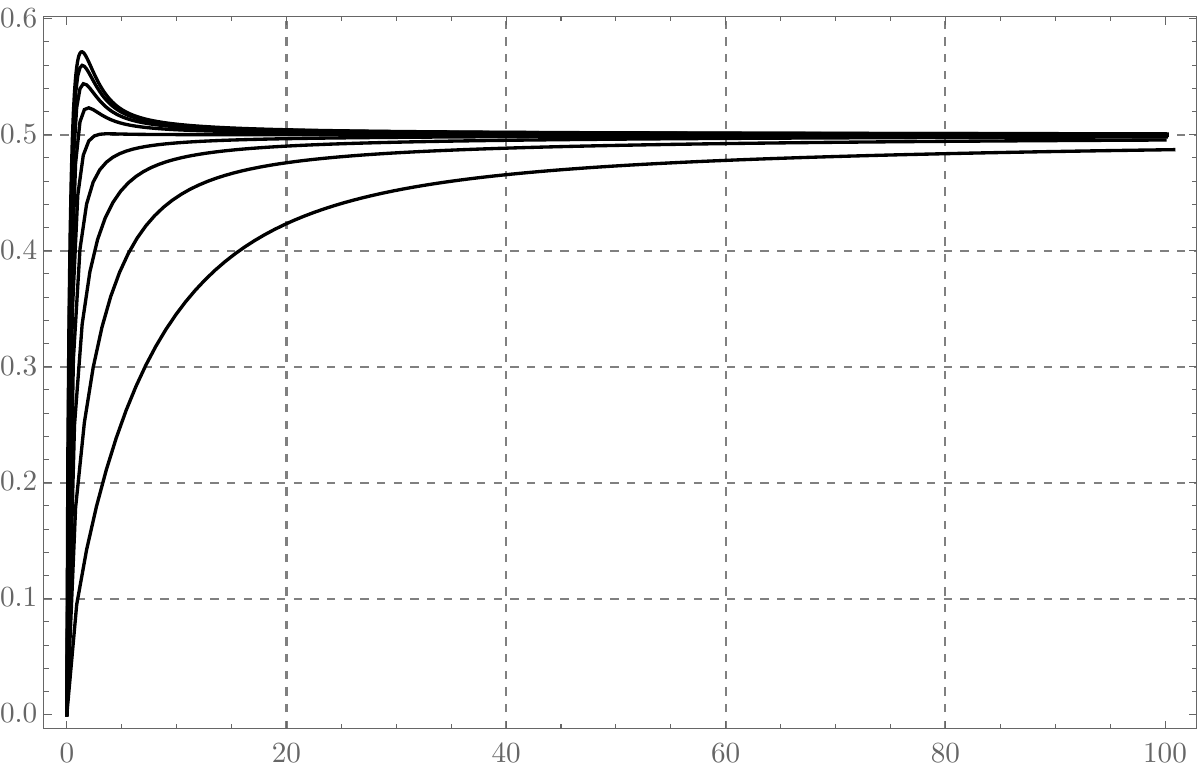}
\end{center}
\caption{Left: Plots of the truncation error $R_p(x/p)$, defined in  \eqref{eqn:RpX:pLarge} and  
\eqref{eqn:RpX:pSmall}, as a function of $x$, for various values of $p$.  
From the highest to the lowest plot: $p\to 0$ (i.e., $x E_1(x)$), and
$p=0.1, 0.2, \ldots, 0.9$. We recall that $R_1(x)=0$.
Right: A similar plot of the KKT gaps attained by the inverse binomial distribution \eqref{eqn:invbin}.
}
\label{fig:BinRpX}
\end{figure}

The error quantity $R_p(x)$, depicted in Figure~\ref{fig:BinRpX}, can be shown to be always non-negative:
\begin{claim} \label{claim:bin:Rpx}
For all $x \geq 0$ and $p \in (0,1/2]$, the quantity $R_p(x)$ defined in 
\eqref{eqn:RpX:pLarge} is non-negative.
\end{claim}

\begin{proof}
We can rearrange \eqref{eqn:RpX:pSmall} as
\begin{align*}
R_p(x) &= \int_{p}^{\frac{p}{1-p}} \frac{(1-t)^x - (1-p)^x}{t \log(1-t)}\, dt
-\int_{\frac{p}{1-p}}^1 \frac{-(1-t)^x + (1-p)^x + (1-(1-p)t)^x - (1-p)^x}{t \log(1-t)}\, dt  \\
&= \int_{p}^{\frac{p}{1-p}} \frac{(1-t)^x - (1-p)^x}{t \log(1-t)}\, dt
+\int_{\frac{p}{1-p}}^1 \frac{(1-t)^x - (1-(1-p)t)^x}{t \log(1-t)}\, dt,
\end{align*}
and observe that the integrands inside both integrals in the second equation are non-negative
over the integration interval.
\end{proof}

In fact, as $p$ gets small, $R_p(x)$ converges to $px E_1(px)$,
which, as shown in \eqref{eqn:Poi:trunc}, is
the gap to equality achieved by the digamma distribution \eqref{eqn:PoiYre} for the mean-limited
Poisson channel.
This can be seen from the result of Proposition~\ref{prop:smallp:approx}
(showing that the truncated distribution for the binomial
channel converges to the digamma distribution \eqref{eqn:PoiYre} in the limit $p \to 0$),
combined with the fact that the binomial channel converges to a 
Poisson channel as $p \to 0$.
 
%
%

Using Claim~\ref{claim:bin:Rpx} and \eqref{eqn:bin:EgYx:pSmall}, we now have that,
for all integers $x \geq 0$,
\[
\E[g(Y_x)] \geq \log x! - \Ex_{1-p}(x),
\]
which, similar to the case $p \geq 1/2$, implies that the probability mass
function of $Y$ in \eqref{eqn:bin:y:general} satisfies
\[
\KL{Y_x}{Y} \stackrel{\eqref{eqn:bin:KL:general}}{=}
-xp \log q - \log y_0 - R_p(x) \leq -xp \log q - \log y_0,
\]
thereby satisfying the KKT conditions \eqref{eqn:cvx1:dualF}
for all $x \geq 0$ and with the choices $\nu_1 = -\log q$
and $\nu_0 = -\log y_0$.

Finally, we derive the asymptotic growth of $g(y)$ using Claim~\ref{claim:Lam}
as follows:
\begin{align*}
g(y) &= \frac{y}{p}(\log(y/p)+\li(1-p)-1)-
\frac{y(1-p)}{p}\left(\log\left(\frac{(1-p)y}{p}\right)+\li\left(\frac{1-2p}{1-p}\right)-1\right)  \\
&+\frac{y}{p}\left( (1-p)\li\left(\frac{1-2p}{1-p}\right)-\li(1-p) \right)
+\frac{1}{2} \log(y/p)- \frac{1}{2} \log\left(\frac{(1-p)y}{p}\right) \pm O(1)\\
&= y h(p)/p+y \log(y/e) +O_p(1),
\end{align*}
which, as was the case for $p \geq 1/2$, confirms
that the probability mass
function of $Y$ in \eqref{eqn:bin:y:general}  can be normalized to a legitimate
distribution precisely when $q \in (0,1)$, and that the expectation can
be adjusted to any desired positive value by choosing $q$ appropriately.
The resulting probability mass function of $Y$ is plotted,
for various choices of $p$ and $q$, in Figure~\ref{fig:bin:PPDF:trunc}.
Furthermore, the mean of the resulting distribution is depicted,
for various choices of $p$, in Figure~\ref{fig:bin:mu:trunc}.

\begin{figure}[t!]
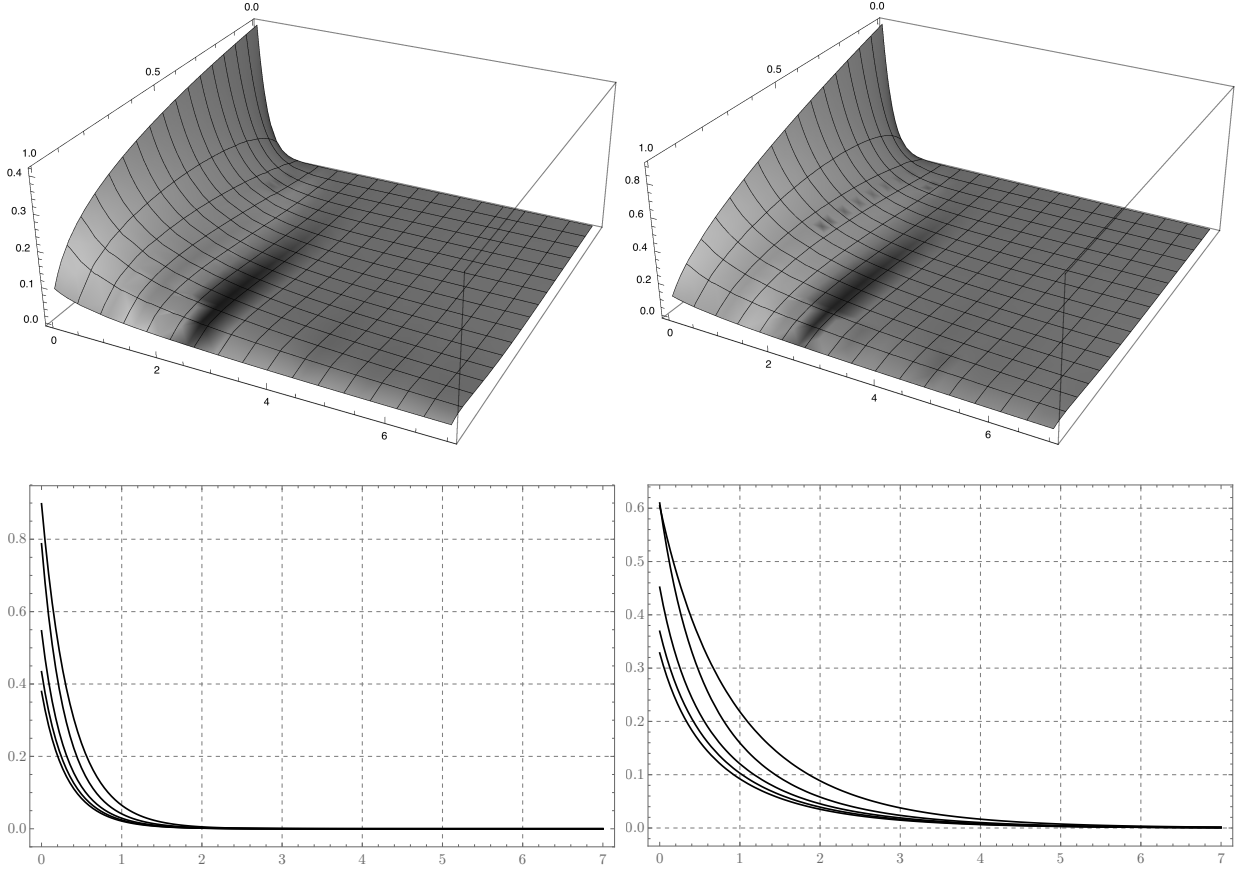

\begin{center}
\includegraphics[width=0.49 \columnwidth]{\detokenize{figPPDF3D_Trunc_p02}}
\includegraphics[width=0.49 \columnwidth]{\detokenize{figPPDF3D_Trunc_p08}}\\[5mm]
\includegraphics[width=0.49 \columnwidth]{\detokenize{figPPDF2D_Trunc_q01}}
\includegraphics[width=0.49 \columnwidth]{\detokenize{figPPDF2D_Trunc_q05}}
\end{center}
\caption{Plots of the probability mass function corresponding
to the choice of $Y$ in Theorem~\ref{thm:BinChUpper:trunc}.
For the top plots, the length represents $y$, the width
represents $q$, and the height is the probability (left: $p=0.2$,
right: $p=0.8$).
For the bottom plots, the probability is plotted as a function
of $y$, for $p=0.1,0.3,0.5,0.7,0.9$ (from the lowest to the highest
curve), where we have chosen $q=0.1$ (left), and $q=0.5$ (right).
}
\label{fig:bin:PPDF:trunc}
\end{figure}

\begin{figure}[t!]
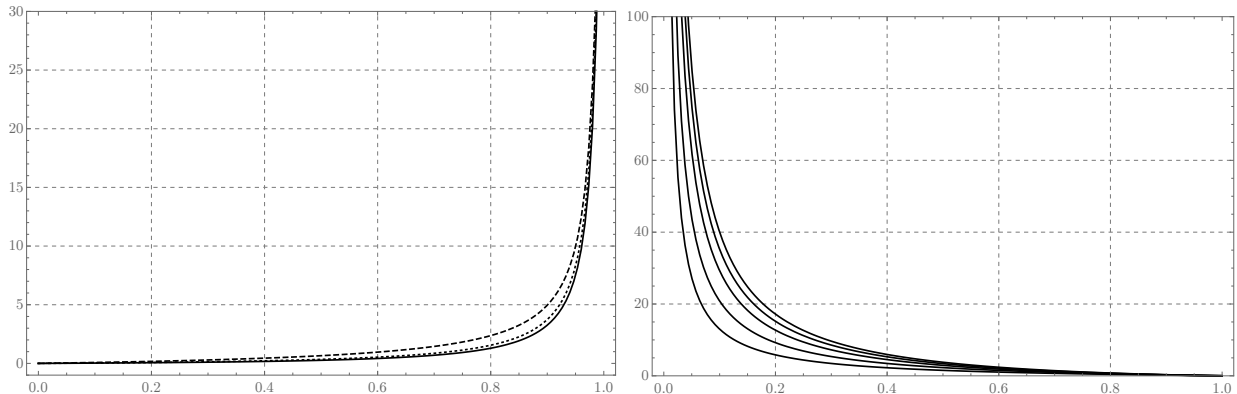

\begin{center}
\includegraphics[width=0.49 \columnwidth]{\detokenize{figMu_Trunc}}
\includegraphics[width=0.49 \columnwidth]{\detokenize{figInvMu_Trunc}}

\end{center}
\caption{Plots of the mean of the probability distribution corresponding
to the choice of $Y$ in Theorem~\ref{thm:BinChUpper:trunc}.
Left: the mean as a function of $q$ for $p=0.1,0.5,0.9$ (from the
lowest to the highest curve). Right: the inverse of the mean as
a function of $q$ for $p=0.1,0.3,0.5,0.7,0.9$ (from the
highest to the lowest curve).
}
\label{fig:bin:mu:trunc}
\end{figure}

\subsubsection*{Wrapping up}

Equipped with an improved choice for the distribution of the channel output $Y$,
we may now proceed as in Section~\ref{sec:del:general} with the alternative choice
applied in Theorem~\ref{thm:BinChUpper}, which is restated with the modified
distribution below.

\begin{thm} \label{thm:BinChUpper:trunc}
Let $p \in (0,1)$ and $q\in(0,1)$ be given parameters and $Y$ be a random
variable distributed according to \eqref{eqn:bin:y:general}; that is,
\[
\Pr[Y=y] = y_0 \frac{q^y \exp(g(y) -y h(p)/p)}{y!},
\]
for an appropriate normalizing constant $y_0$ and mean $\mu=\E[Y]$,
where
\[ g(y) :=
\begin{cases}
0 & y=0 \\
g_p(y) & y > 0,
\end{cases}
\]
and $g_p(y)$ is defined for $p < 1/2$ by \eqref{eqn:gpy:smallP} and
for $p \geq 1/2$ by \eqref{eqn:gpy:largeP}; that is,
\[
g_p(y) := \begin{cases}
\Lam_p(y)-\Ex_{1/p-1}(y)-y \li(1-p)/p+\eta(1-p) & p \geq 1/2, \\
\Lam_p(y)-\Lam_{p/(1-p)}(y)+
\frac{y}{p}\left( (1-p)\li\left(\frac{1-2p}{1-p}\right)-\li(1-p) \right)+
\int_{p}^{\frac{p}{1-p}} \frac{dt}{t \log(1-t)} & p < 1/2.
\end{cases}
\]
Then, capacity of the mean-limited binomial channel 
$\ch_\mu(\ber_p)$ satisfies 
\begin{equation} \label{eqn:BinMeanLimCapUpperB}
\capa(\ch_\mu(\ber_p)) \leq -\mu \log q-\log y_0.
\end{equation}
\end{thm}


\subsection{Derivation of the capacity upper bound for 
the deletion channel}
\label{sec:del:derivation}

In order to complete the derivation of the capacity upper bound for
the deletion channel, we combine the results of Section~\ref{sec:del:general},
and their improvements in Section~\ref{sec:del:general:trunc},
with the general framework developed in Section~\ref{sec:general:repeat};
in particular, Corollary~\ref{coro:Cupper}.

Denoting by $\ch$ the deletion channel with deletion probability $d$,
and letting $p := 1-d$, the result of Corollary~\ref{coro:Cupper} 
combined with either Theorem~\ref{thm:BinChUpper} or
Theorem~\ref{thm:BinChUpper:trunc} gives
us the capacity upper bound
\begin{align}
\capa(\ch) &\leq \sup_{\mu\geq 0} \frac{\capa(\ch_\mu(\ber_p))}{1/p+\mu/p}  \nonumber \\
&\stackrel{\eqref{eqn:BinMeanLimCapUpperA}}{\leq}
p\, \sup_{q \in (0,1)} \frac{-\mu \log q-\log y_0}{1+\mu} =: \CBer(p), \label{eqn:del:upper:formula}
\end{align}
where $\mu$ and $y_0$ respectively denote the mean and normalizing constant
of the distribution of $Y$ (defined by either Theorem~\ref{thm:BinChUpper} or
Theorem~\ref{thm:BinChUpper:trunc}) 
and each depend
on both $p$ and $q$. 
Let $\CBer(p)$ denote the expression on the right hand side of \eqref{eqn:del:upper:formula}
and $\CBer(p,q)$ denote\footnote{Note that both $\CBer(p)$ and $\CBer(p,q)$
also depend on the underlying distribution for $Y$. However, we suppress this dependence
in the notation, which should be clear from the context.} the expression inside the supremum in \eqref{eqn:del:upper:formula}.
We have proved the following:

\begin{thm} \label{thm:del:upper:general}
Let $\ch$ be a deletion channel with deletion probability $d$,
and let $p := 1-d$. Given a parameter $q \in (0,1)$, consider
a random variable $Y$ distributed according to 
either Theorem~\ref{thm:BinChUpper} (inverse binomial) or
Theorem~\ref{thm:BinChUpper:trunc} (the truncated distribution),
with an
appropriate normalizing constant $y_0(q)$ and mean $\mu(q) = \E[Y]$. Then,
\begin{equation} \label{eqn:del:upper:general}
\capa(\ch) \leq \CBer(1-d) = (1-d)\, \sup_{q \in (0,1)} \frac{-\mu(q) \log q-\log y_0(q)}{1+\mu(q)}.
\end{equation}
\end{thm}

\begin{figure}[t!]
\begin{center}
\includegraphics[width=0.49\columnwidth]{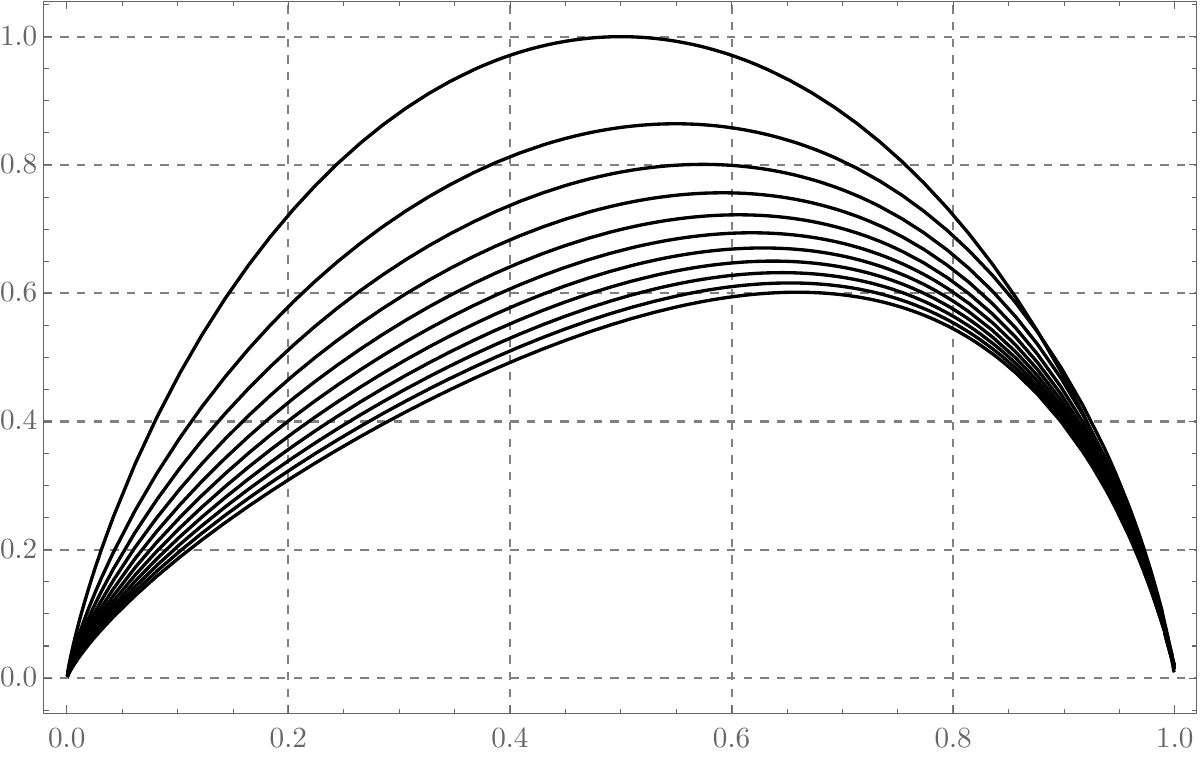}
\includegraphics[width=0.49\columnwidth]{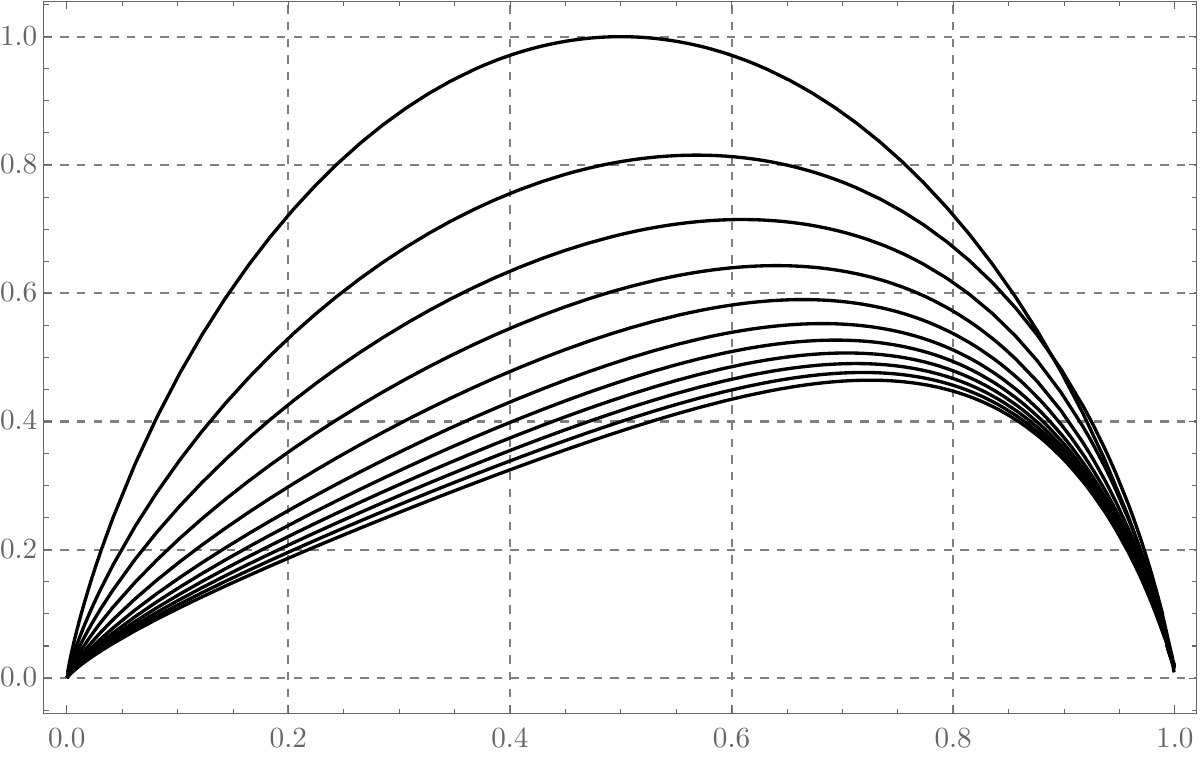}
\end{center}
\caption{Numerical plots of the capacity upper bound
slope $\CBer(p,q)$ (the expression
under the supremum in \eqref{eqn:del:upper:formula}), measured in bits,
for various choices
of $p$ and as a function of $q$, where the distribution 
of $Y$ is given by the inverse binomial distribution in 
Theorem~\ref{thm:BinChUpper} (left), or the improved (truncated) distribution 
in Theorem~\ref{thm:BinChUpper:trunc} (right).
The chosen values for $p$ are (from the lowest to the highest curve):
$p=10^{-4},0.1,0.2,\ldots, 0.9$, and $p = 1$ (in which case the 
curve is equal to $h(q)$).
}
\label{fig:DelSlopes}
\vspace{2mm} \hrule
\end{figure}

\begin{figure}[t!]
\begin{center}
\includegraphics[width=0.49\columnwidth]{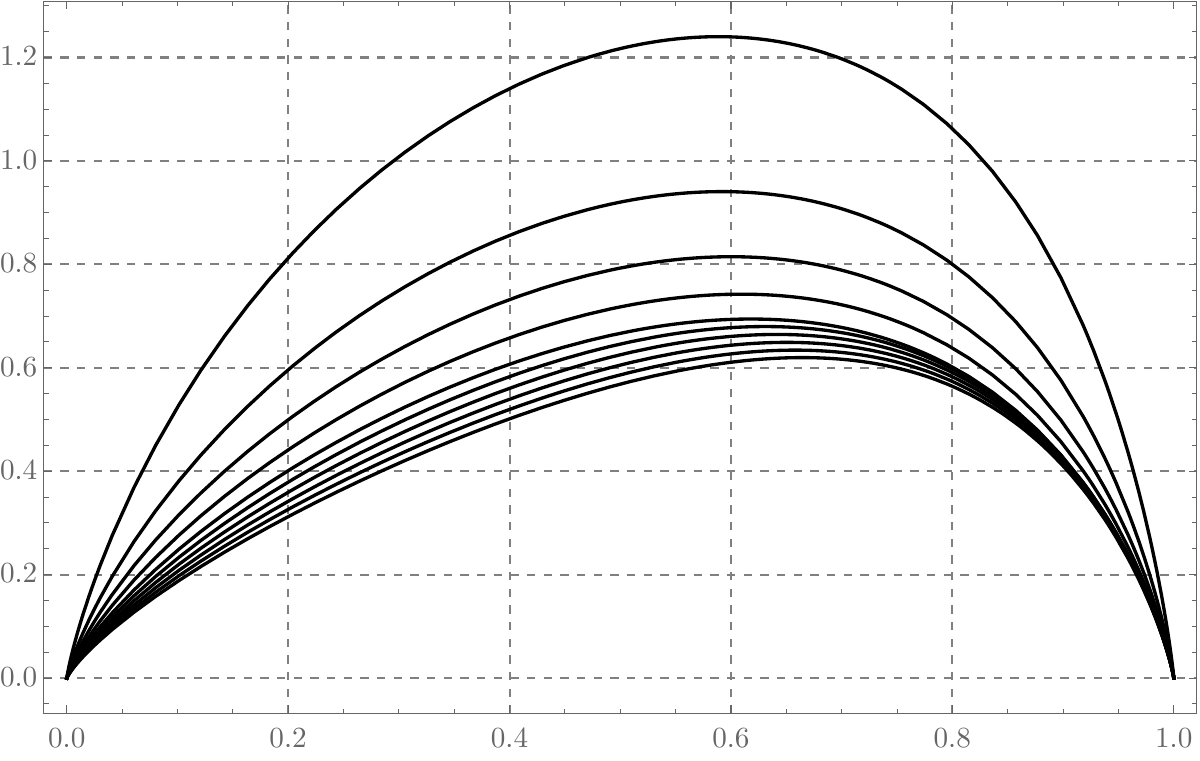}
\includegraphics[width=0.49\columnwidth]{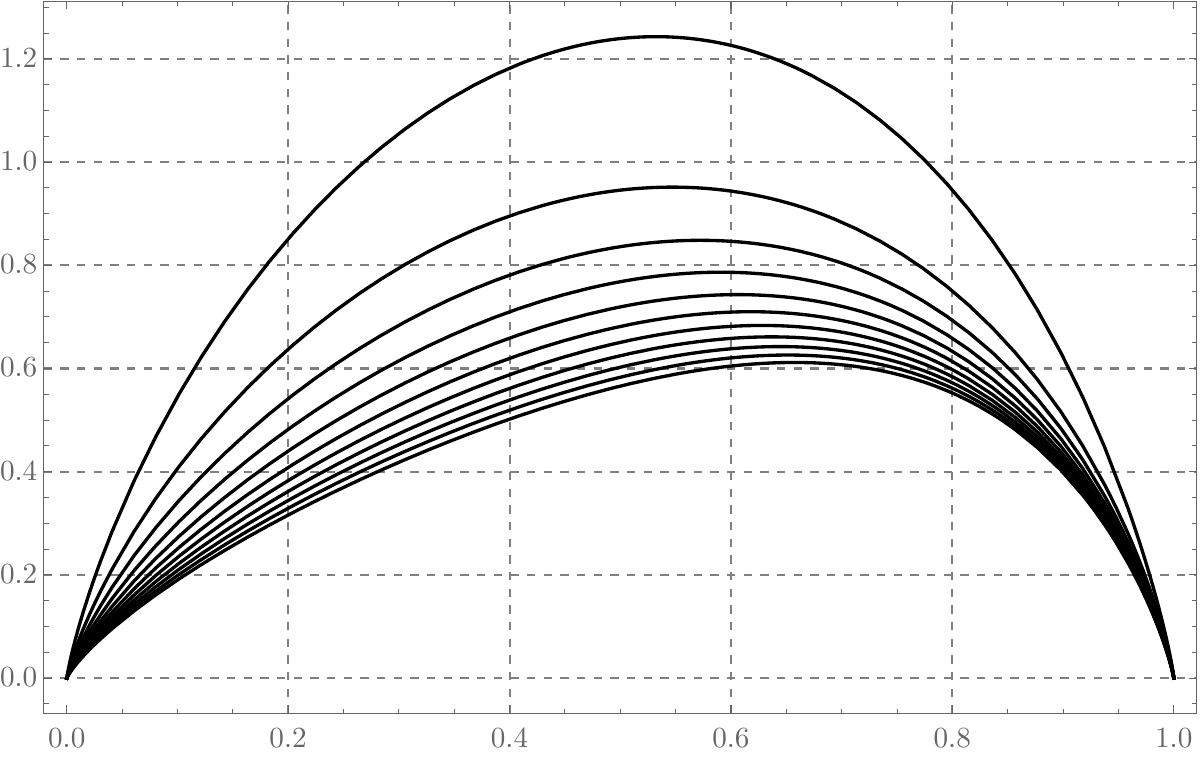}
\end{center}
\caption{Analytic upper bounds on the capacity slope $\CBer(p,q)$ 
(the expression
under the supremum in \eqref{eqn:del:upper:formula}), measured in bits,
for various choices
of $p$ and as a function of $q$, where the distribution 
of $Y$ is given by the inverse binomial distribution in 
Theorem~\ref{thm:BinChUpper}. The diagram on the left plots 
the upper bound on $\CBer(p,q)$ in terms of elementary functions,
using Corollary~\ref{coro:invbin:negBinApprox}.
The diagram on the right uses the upper bounds in terms of standard
special functions given by Theorem~\ref{thm:ibin:Lerch}.
The choices of $p$ are, from the lowest curve to the highest:
$p=10^{-4},0.1,0.2,\ldots, 0.9$, and $p = 0.999$
(the latter being excluded in the first diagram as the resulting
capacity upper bounds are already trivial at $p=0.9$).
}
\label{fig:DelSlopesUB}
\vspace{2mm} \hrule
\end{figure}

Note that, when $p=1$, both Theorem~\ref{thm:BinChUpper} and
Theorem~\ref{thm:BinChUpper:trunc} assign a geometric distribution to $Y$.
As we show in Section~\ref{sec:del:smallD}, in both cases
we have $\CBer(1,q)=h(q)$.

\subsubsection{The particular case $d=1/2$}
\label{sec:del:dhalf}
Recall that, for the special case $p=d=1/2$, the inverse binomial
distribution defined in \eqref{eqn:invbin} becomes, precisely, a negative binomial distribution. 
Thus, in this case, the right hand side of \eqref{eqn:del:upper:general} 
can be analytically optimized in closed form. Using \eqref{eqn:ibin:half},
the expression under supremum is equal to
\begin{equation} \label{eqn:CberHalf}
\CBer(1/2,q) = \frac{-\frac{q}{2(1-q)} \log q-\log \sqrt{1-q}}{1+\frac{q}{2(1-q)}}=
\frac{h(q)}{2-q},
\end{equation}
whose maximum is attained at the golden ratio conjugate $q=(\sqrt{5}-1)/2$
(shown by equating the derivative of the expression to zero, which results
in a quadratic equation).
It can be verified by straightforward
manipulations that, in this case, the resulting capacity upper bound 
given by  \eqref{eqn:del:upper:general} is equal to
\begin{equation} \label{eqn:del:pHalf}
\CBer(1/2) =\frac{1}{4}\log \frac{3+\sqrt{5}}{2} = \frac{1}{2} \log \varphi 
\approx 0.347120\ \text{\rm (in bits per channel use)},
\end{equation}
where $\varphi =(1+\sqrt{5})/2$ is the golden ratio.

An extension of the above result for $p=1/2$ to smaller values of $p$
is presented in Appendix~\ref{app:claim:cber:analytical}. However, 
the convexification result
of \cite{ref:RD15} may be used along with our capacity upper bound for
$d=p=1/2$ 
to derive simple and closed-form capacity upper bounds for general 
deletion probabilities. Namely,
we prove the following:

\begin{coro} \label{coro:del:convexify}
Let $\ch$ be the deletion channel with deletion probability $d$. Then,
\[
\capa(\ch) \leq \begin{cases}
(1-d) \log \varphi
\approx 0.694242 (1-d) & d \geq 1/2,\\
1-d \log(4/\varphi) 
\approx 1-1.305758 d & d < 1/2, \\
\end{cases}
\]
where $\varphi =(1+\sqrt{5})/2$ is the golden ratio,
the entropy is in bits per channel use, and the bound for $d<1/2$
holds under the plausible
conjecture \cite{ref:Dal11}
that the capacity function is convex over $d \in [0,1/2]$.
%
\end{coro}

\begin{proof}
Suppose that capacity upper bounds of $c_1$ and $c_2$, respectively
for deletion probabilities $d_1$ and $d_2$ are known.
Let $\ell \in (0,1)$ and $c$ be the capacity of the channel
at deletion probability $d := \ell d_1 + (1-\ell) d_2$.
Under the assumption that the capacity function for the
deletion channel is a convex function of $d$, it would trivially follow that
$c \leq \ell c_1 + (1-\ell) c_2$, hence proving the claim by letting
$d_1 = 1/2$, $c_1 = (\log \varphi)/2$, and either $d_2=0$ with $c_2=\log 2$
or $d_2=1$ with $c_2=0$.
Without assuming the convexity conjecture, it has been shown in \cite[Theorem~1]{ref:RD15}
that, unconditionally, one has (with entropy measured in bits)
\[
c \leq \ell c_1 + (1-\ell) c_2 + (1-d) \log(1-d)-\ell(1-d_1) \log(\ell(1-d_1))
- (1-\ell)(1-d_2) \log((1-\ell)(1-d_2)).
\]
By letting $d_2=1$ and $c_2=0$, and $d=\ell d_1+(1-\ell)$ (thus,
$1-d=\ell(1-d_1)$) one gets 
$c \leq \ell c_1 + \ell(1-d_1)\log(\ell(1-d_1))- \ell(1-d_1)\log(\ell(1-d_1))=\ell c_1$,
proving the unconditional claim for $d \geq 1/2$.
We remark that one could use the result of \cite{ref:RD15} with $d<1/2$
and get nontrivial upper bounds, unconditionally, for a range of $d$ below
$1/2$ as well (e.g., $d \geq 0.48$). 
\end{proof}

\subsubsection{The general case.}
For general $p$, the function $\CBer(p,q)$, under both Theorems 
\ref{thm:BinChUpper}~and~\ref{thm:BinChUpper:trunc},
is numerically plotted\footnote{We remark that,
since the probability mass function of $Y$ has
an exponential decay, the function $\CBer(p,q)$ can be 
numerically computed efficiently and in polynomial
time in the desired accuracy. Moreover, the plots
suggest that for every $p$, this function is 
concave in $q$ and thus its maximum can also be
numerically computed in polynomial time in 
the desired accuracy (e.g., by a simple binary search,
or the Newton's method). Even though concavity
is evident from Figure~\ref{fig:DelSlopes} (and similarly,
Figure~\ref{fig:PoiSLope} for the Poisson case),
it is not proved formally, and we leave it as an interesting
remaining task. The concavity of the upper bound estimates in terms
of analytic functions (Figure~\ref{fig:DelSlopesUB})
is, however, straightforward to analytically verify.
} in Figure~\ref{fig:DelSlopes}.
Furthermore, if Theorem~\ref{thm:BinChUpper} is used
to determine the distribution of $Y$ (i.e., the inverse binomial
distribution), high-quality upper bound
estimates on the value of $\CBer(p,q)$ in terms
of elementary or standard special functions (when $p$
is not too close to $1$) are available 
via Corollary~\ref{coro:invbin:negBinApprox} and
Theorem~\ref{thm:ibin:Lerch}. These upper bounds
are plotted in Figure~\ref{fig:DelSlopesUB}.

Maximizing the value of $\CBer(p,q)$ with respect to $q$
for a given $p=1-d$ results in capacity upper bounds
for the deletion channel with deletion probability $d$.
The quality of the upper bound depends on the result used
to determine $\CBer(p,q)$, including, and in decreasing order of quality:
\begin{enumerate}
\item Theorem~\ref{thm:BinChUpper:trunc} for the distribution of $Y$ (i.e.,
the truncated distribution),
\item Theorem~\ref{thm:BinChUpper} 
for the distribution of $Y$ (i.e., the inverse binomial distribution),
\item The analytic upper bound estimates of Theorem~\ref{thm:ibin:Lerch} on
the inverse binomial distribution,
\item The elementary upper bound estimates of 
Corollary~\ref{coro:invbin:negBinApprox} on the parameters of the
inverse binomial distribution.
\end{enumerate}

Plots of the resulting capacity upper bounds for general $d$
are given in Figure~\ref{fig:DelPlot}, and are compared with
the best available capacity upper bounds as reported in \cite{ref:RD15} (which
are, in turn, based on the results of \cite{ref:FD10} combined with
a convexification technique).
The corresponding numerical values for the plotted curves are
also listed, for various choices of the deletion probability $d$, 
in Table~\ref{tab:DelCapData}.

\begin{figure}[t!] 
\begin{center}
\includegraphics[width=0.8\columnwidth]{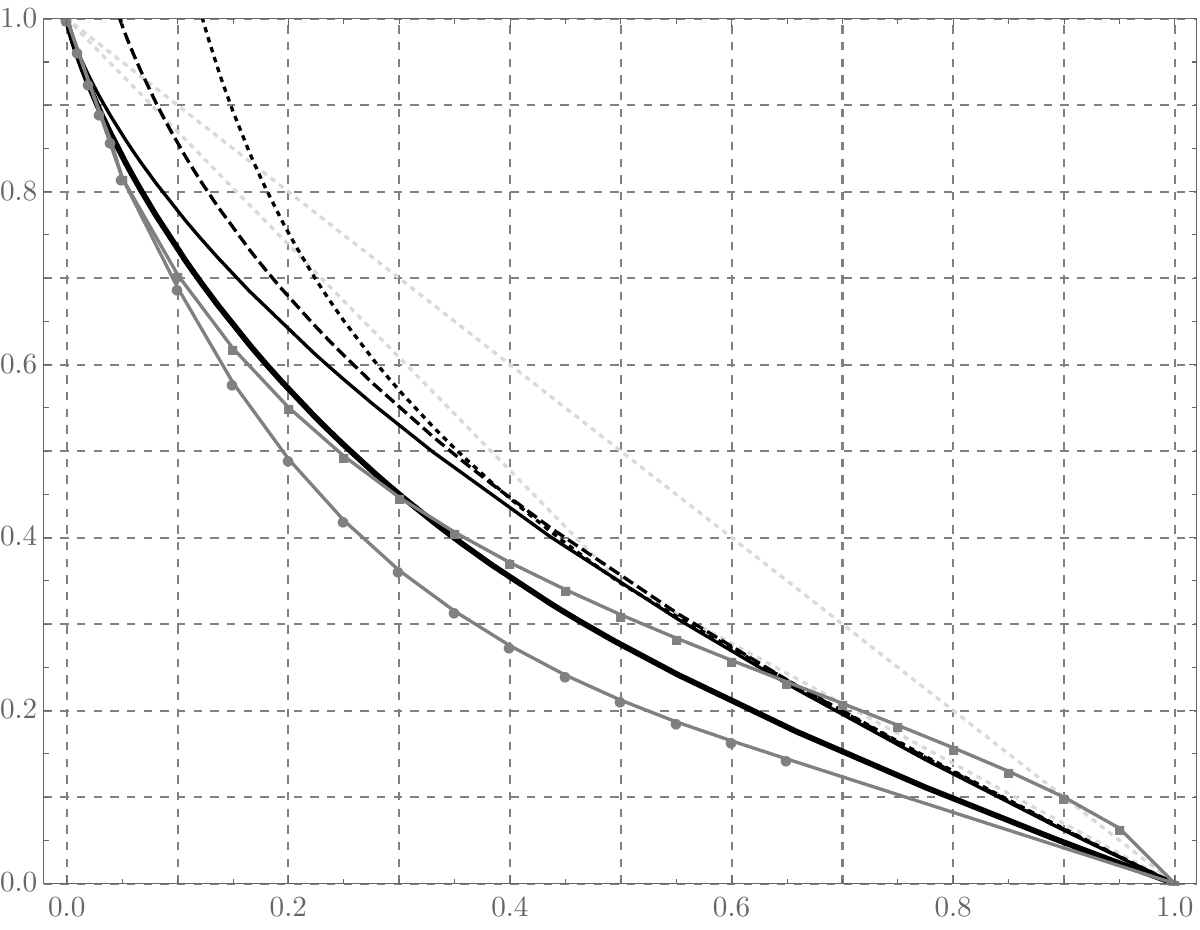}
\end{center}
\caption[]{Upper bounds (in bits per channel use) 
on the capacity of the deletion channel, 
plotted as a function of the deletion probability $d$.
The bounds are obtained using 
\begin{enumerate*}[label=\roman*)]
\item Theorem~\ref{thm:BinChUpper:trunc} (solid, thick),
\item Theorem~\ref{thm:BinChUpper} (solid, black),
\item analytic upper bound estimates of Theorem~\ref{thm:ibin:Lerch} (dashed, black),
\item elementary upper bound estimate of 
Corollary~\ref{coro:invbin:negBinApprox} (dotted, black).
\end{enumerate*}
The best known capacity upper bounds reported in \cite{ref:RD15} 
are shown in gray with the circle markers representing the
explicitly reported data points. The gray plot with square markers
are the upper bounds reported in \cite{ref:MDP07}.
The trivial erasure capacity upper bound $1-d$ as well as the
fully analytic upper bounds of Corollary~\ref{coro:del:convexify}
are displayed in dotted
light gray.
The numeric values corresponding to the plots are listed 
in Table~\ref{tab:DelCapData}.
}
\label{fig:DelPlot}
\vspace{2mm} \hrule
\end{figure}

\subsubsection{The limiting case $d \to 1$}
\label{sec:del:largeD}

When the deletion probability $d$ tends to $1$, or equivalently $p = 1-d \to 0$, 
the capacity upper bounds take the form $\CBer(p) \leq C_0 (1-d)$,
for an absolute constant $C_0$. The value of $C_0$ depends on which one
of the above four results is used, and by numerically maximizing the
univariate concave function $\CBer(p,q)$ over $q \in (0,1)$, can be approximated,
respectively, as follows (measured in bits):
\begin{enumerate}
\item Under Theorem~\ref{thm:BinChUpper:trunc},
$C_0 \approx 0.4644$ with maximizer $q \approx 0.7247$,

\item Under Theorem~\ref{thm:BinChUpper},
$C_0 \approx 0.6015$ with maximizer $q \approx 0.6590$,

\item Using the analytic upper bound estimate of Theorem~\ref{thm:ibin:Lerch},
$C_0 \approx 0.6115$ with maximizer $q \approx 0.6573$,

\item Using the elementary upper bound estimate of 
Corollary~\ref{coro:invbin:negBinApprox},
$C_0 \approx 0.6196$ with maximizer $q \approx 0.6644$.
\end{enumerate}
We also recall that, as we saw in the result of Section~\ref{sec:del:dhalf}
for the particular case $d=1/2$,
the fully explicit upper bound of $(\log \varphi)/2$
for this case can be linearly extended to all $d \geq 1/2$,
including the limiting case $d \to 1$, using the convexification
technique of \cite{ref:RD15}. 
Indeed, Corollary~\ref{coro:del:convexify} implies the
choice of $C_0 = \log \varphi \approx 0.694242 (1-d)$
(in bits), where $\varphi$ is the golden ratio, for this case.

\begin{remark} \label{rem:slope:same}   
We remark that the above upper bound estimate of
$0.4644(1-d)$ for $d \to 1$
coincides with what we achieve for the Poisson-repeat channel 
with the same deletion probability in 
Corollary~\ref{coro:Poi:slope}.  
This is due to the fact that the truncated distribution
\eqref{eqn:PoiYre} (the digamma distribution) for the Poisson-repeat channel
is the limiting distribution of what 
we obtain in Section~\ref{sec:del:general:trunc}
using the truncation method (as shown by
Proposition~\ref{prop:smallp:approx}),
and that the mean-limited binomial channel converges
a mean-limited Poisson channel in the limit $d \to 1$.
\end{remark}

\subsubsection{The limiting case $d \to 0$}
\label{sec:del:smallD}

The limiting behavior of the capacity of the deletion channel for
$d \to 0$ is very well understood \cite{ref:KMS10,ref:KM13}.
In particular, in this case, it is known that the capacity 
behaves as $1-h(d)+O(d) = 1+ d \log d + O(d)$ (in bits per channel use).
The goal of this section is to prove that the capacity upper
bounds obtained by Theorem~\ref{thm:del:upper:general}
exhibit the correct asymptotic behavior of $1-\Theta(h(d))$ 
for small $d$ (albeit with a slightly sup-optimal constant
behind $h(d)$). We demonstrate the result for 
Theorem~\ref{thm:del:upper:general} applied with the weaker
choice of the inverse binomial distribution. The same approach
could be used to obtain an analogous results for the truncated
distribution of Theorem~\ref{thm:BinChUpper:trunc}.

\begin{thm} \label{thm:del:smallD}
Consider the deletion channel $\ch$ with deletion probability $d \to 0$.
Then, the capacity upper bound of Theorem~\ref{thm:del:upper:general}
with respect to the inverse binomial distribution \eqref{eqn:invbin}
takes the form
\[
\capa(\ch) \leq 1-(1-O(d)) h(d)/2 = 1-h(d)/2 + o(d)
\text{ (bits per channel use)}.
\]
\end{thm}

\begin{proof}
Let $Y$ be an inverse binomial random variable with parameter $q$,
normalizing constant $y_0$ and mean $\mu := \E[Y]$.
We first recall \eqref{eqn:del:upper:formula}, that
\[
\capa(\ch) \leq 
(1-d)\, \sup_{q \in (0,1)} \frac{-\mu \log q-\log y_0}{1+\mu} \nonumber. 
\]
We first rule out the possibility of the optimum $q$ being close to $1$
(which is evident from Figure~\ref{fig:DelSlopesUB}). This is immediate
from the estimates on the parameters of an inverse binomial distribution
given by Corollary~\ref{coro:invbin:negBinApprox}. Namely,  
Corollary~\ref{coro:invbin:negBinApprox} proves that, when $q \to 1$,
we have $-\log y_0 = O(-\log(1-q))$, and $\mu = \Omega(1/(1-q))$.
Thus, in this regime we have
\[
\frac{-\mu \log q-\log y_0}{1+\mu} \leq -\log q-\frac{\log y_0}{\mu}
= -\log q+O(-(1-q)\log(1-q))=O(h(q)) \to 0.
\]
Without loss of generality we may therefore assume that
there is an absolute constant $q_0<1$ such that $q \leq q_0$ (since
we now know that the
supremum in \eqref{eqn:del:upper:formula} over the values of $q$ close to $1$
approaches zero).
Logarithm of the binomial coefficient $\binom{y/p}{y}$ admits
the series expansion (in $d$)
\[
\log \binom{\frac{y}{1-d}}{y} = (\gamma+\psi(y+1)) yd+
\frac{y}{12} \left(12 \gamma +12 \psi(y+1)-\pi ^2 y+6 \psi'(y+1)y\right) d^2+O(d^3 y^2+
d^4 y^4) + \cdots,
\]
where $\gamma$ is the Euler-Mascheroni constant, and $\psi$
is the digamma function. Consider a parameter $y_1 = O(\log(1/d))$
to be determined later.
For any $y \leq y_1$, we may thus write
\[
1 \leq \binom{\frac{y}{1-d}}{y} \leq 1+O(d y \log y) = 1+O(d y^2).
\]
From the definition of the inverse binomial distribution \eqref{eqn:invbin},
and the above estimate, we may write
\begin{equation} \label{eqn:ibin:geom:sw}
y_0 q^y \exp(-y h(d)/(1-d)) \leq \Pr[Y=y] \leq y_0 q^y \exp(-y h(d)/(1-d)) (1+O(d y^2))
\end{equation}
for all $y \leq y_1$.
We also recall that, letting $\delta := \exp(-h(d)/(1-d))$, 
for all $y \geq 0$,
\begin{equation} \label{eqn:binom:upper}
\binom{\frac{y}{1-d}}{y} \leq \exp(y h(1-d)/(1-d)) = 1/\delta^y,
\end{equation}
and thus, for all $y \geq 0$,
\begin{equation} \label{eqn:ibin:trivial:upper}
\Pr[Y=y]/y_0 = \binom{\frac{y}{1-d}}{y} (\delta q)^y
\stackrel{\eqref{eqn:binom:upper}}{\leq} q^y.
\end{equation}
Since $1-q = \Omega(1)$, we may choose $y_1$ large enough so as to
ensure that
\begin{equation} \label{eqn:ibin:tail}
\sum_{y=y_1}^\infty y \Pr[Y=y]/y_0 \leq dq < d.
\end{equation}
In the sequel, we use asymptotic notation with respect to the vanishing
parameter $d$; i.e., $d = o(1)$.
We may write, 
\begin{align*}
1/y_0 &= \sum_{y=0}^\infty \Pr[Y=y]/y_0 \\
&= \sum_{y=0}^\infty \binom{\frac{y}{1-d}}{y} (\delta q)^y  \\
&\stackrel{\eqref{eqn:ibin:tail}}{\leq} 
d+\sum_{y=0}^{y_1} \binom{\frac{y}{1-d}}{y} (\delta q)^y  \\
&\leq d+\sum_{y=0}^\infty (\delta q)^y (1+O(d y^2))\\
&=d+ \frac{1}{1-\delta q} + O \left(\frac{d \delta q (1+\delta q)}{(1-\delta q)^3}\right) \\
&\leq d+\frac{1}{1-\delta q} 
\left(1+ O \left(\frac{d \delta q }{(1-\delta q)^2}\right)\right) \\
&\leq \frac{1+O(d)}{1-\delta q},
\end{align*}
where in the last inequality we have used the assumptions $q \leq q_0$, 
$1-q_0 = \Omega(1)$, $d=o(1)$, and 
$\delta = 1-\Theta(d \log d)=1-o(1)$.
Upper bounding $y_0$ is slightly simpler:
\begin{align*}
1/y_0 = \sum_{y=0}^\infty \Pr[Y=y]/y_0 
&\geq \sum_{y=0}^\infty (q \delta)^y = \frac{1}{1-\delta q}.
\end{align*}
Using a similar approach, we may upper bound $\mu$ as follows:
\begin{align*}
\mu/y_0 &= \sum_{y=1}^\infty y \Pr[Y=y]/y_0 \\
&\stackrel{\eqref{eqn:ibin:tail}}{\leq} dq+\sum_{y=1}^{y_1} y \Pr[Y=y]/y_0 \\
&\leq 
dq+\sum_{y=1}^\infty y(\delta q)^y (1+O(d y^2)) \\
&= dq+\frac{\delta q}{(1-\delta q)^2} + O\left(
\frac{d (\delta^2 q^2 +4 \delta q +1)}{(1-\delta q)^4}\right)\\
&\leq dq+\frac{\delta q}{(1-\delta q)^2} 
\left(1+ O \left(\frac{d}{(1-\delta q)^2}\right)\right) \\
&\leq \frac{\delta q}{(1-\delta q)^2} (1+O(d)),
\end{align*}
so that, using the upper bound on $y_0$,
\[
\mu \leq \frac{\delta q (1+O(d))}{1-\delta q}.
\]
Finally, we lower bound $\mu$ as
\begin{align*}
\mu/y_0 = \sum_{y=1}^\infty y \Pr[Y=y]/y_0 
&\geq \sum_{y=1}^\infty y (q \delta)^y = \frac{\delta q}{(1-\delta q)^2},
\end{align*}
and using the lower bound on $y_0$, we write
\[
\mu \geq \frac{\delta q (1-O(d))}{1-\delta q}.
\]
In conclusion, we have so far shown the estimates
\begin{align} \label{eqn:ibin:geom:c}
y_0 = (1-\delta q) (1 - O(d)), \quad \mu = \frac{\delta q (1 \pm O(d))}{1-\delta q}
\end{align}
We are now ready to apply the above estimates in \eqref{eqn:del:upper:formula}, and write
\begin{align}
\capa(\ch) &\leq 
(1-d)\, \sup_{q \in (0,1)} \frac{-\mu \log q-\log y_0}{1+\mu} \nonumber  \\
&\leq
\sup_{q \in (0,q_0)} \frac{-\mu \log q-\log y_0}{1+\mu} \nonumber \\
&\stackrel{\eqref{eqn:ibin:geom:c}}{\leq} 
\sup_{q \in (0,q_0)} (1-\delta q) \frac{-\frac{\delta q}{1-\delta q} \log q-
\log (1-\delta q) \pm O(d)}{1 \pm O(d)}\, (1+O(d)) \nonumber \\
&\leq \sup_{q \in (0,q_0)} (h(\delta q)+\delta q \log \delta \pm O(d))\, (1+O(d)) \nonumber \\
&\leq \sup_{q \in (0,q_0)} (h(\delta q)-\delta q h(d)/(1-d) \pm O(d))\, (1+O(d)) \nonumber \\
&\leq \sup_{q \in (0,q_0)} (h(\delta q)-\delta q h(d) \pm O(d))\, (1+O(d)). \nonumber
\end{align}
In the above result, the expression under the supremum approaches zero
for $q \to 0$, and approaches one for $\delta q = 1/2$. Therefore, we expect the
supremum to occur around $q \approx 1/2$ and be close to one. In particular,
we know that the supremum is away from zero (by a constant), and may thus write the above as
\begin{align}
\capa(\ch) \leq (1+O(d)) \sup_{q \in (0,q_0)} (h(\delta q)-\delta q h(d)).
\label{eqn:ibin:geom:d}
\end{align}

Consider the function $c(\rho) := h(\rho)+\eps \rho$. By simply equating
the derivative of the function to zero, it follows that the maximum
of this function is attained at 
$\rho^\star=e^{\eps}/(1+e^\eps)$ and the that maximum value is
\[c^\star = \frac{e^{\eps}(\eps+\log(1+e^{-\eps}))+\log(1+e^\eps)}{1+e^\eps}
= \log 2 + \eps/2 + \eps^2/8 -O(\eps^4).
\]
By letting $\rho:=\delta q$ and $\eps := -h(d)$, we conclude that
\[
\eqref{eqn:ibin:geom:d} \leq (\log 2 - h(d)/2)(1+O(d)) = \log 2 - (1-O(d)) h(d)/2,
\] 
as desired.
\end{proof}

\iftoggle{full}{\SecDiscussion}{}

\section*{Acknowledgments}

The author thanks Suvrit Sra and Iosif Pinelis for referring
him to \cite{ref:KP16} for the proof of Claim~\ref{claim:binom:monotone},
as well as Marco Dalai, Suhas Diggavi, Tolga Duman, Michael Mitzenmacher, and
Andrea Montanari for their comments on earlier drafts of this article.

\bibliographystyle{abbrv}
\bibliography{\jobname}

\appendix

\section{Proof of Theorem~\ref{thm:PoiChUpperAnalytic}}
\label{app:thm:PoiChUpperAnalytic}

Our starting point is the Stirling approximation of the gamma function
\[
n!=\Gamma(n+1)\sim\sqrt{2\pi n} \left(\frac{n}{e}\right)^n,
\]
which asymptotically matches $n!$ and, 
for all $n > 0$, provides a lower bound on $n!$.
A generalized form of the approximation is
\begin{equation} \label{eqn:Stirling}
n! \sim\sqrt{2\pi (n+\sigma)} \left(\frac{n}{e}\right)^n,
\end{equation}
for some real parameter $\sigma \geq 0$, so the
Stirling approximation is the special case $\sigma=0$.
By ``fine tuning'' the constant $\sigma$, it is possible
to obtain sharp lower bounds, and upper bounds, on 
$n!$ for all $n \geq 1$. Lets us call a value of $\sigma$
a \emph{lower (resp., upper) bounding constant for \eqref{eqn:Stirling}}
if the resulting estimate provides a lower (resp., upper) bound on $n!$
for all $n \geq 1$.
It was shown by Gosper~\cite{ref:Gos78}
that $\sigma=1/6\approx 0.166$ provides a remarkably accurate lower bound
estimate on $n!$. The accuracy of this estimate has been studied
in \cite{ref:Nem11}, where it is shown that
\[
\Gamma(n+1)=\left(\frac{n}{e}\right)^n \sqrt{2\pi(x+1/6)}
\left( 1+\frac{1}{144(n+1/4)^2}-\frac{1}{12960(n+1/4)^3}-\cdots \right),
\]
leading to a multiplicative error of about $0.4\%$ even for $n=1$
(as opposed to about $8\%$ achieved by the Stirling approximation).
Let us consider fixed choices of lower and upper bounding constants,
respectively, $\lsigma$ and $\usigma$,
for \eqref{eqn:Stirling}. By inspection, we find that the following
are valid choices\footnote{The choice of $\lsigma$ has been rigorously validated
by \cite{ref:Nem11}. However, we have only numerically validated $\usigma$ 
and it would be interesting to obtain a rigorous proof of its validity.}
\begin{equation} \label{eqn:sigma}
\lsigma := 1/6 = 15/90, \quad \usigma := 0.177 \approx 16/90.
\end{equation}
\begin{figure}[t!]
\begin{center}
\includegraphics[height=2in]{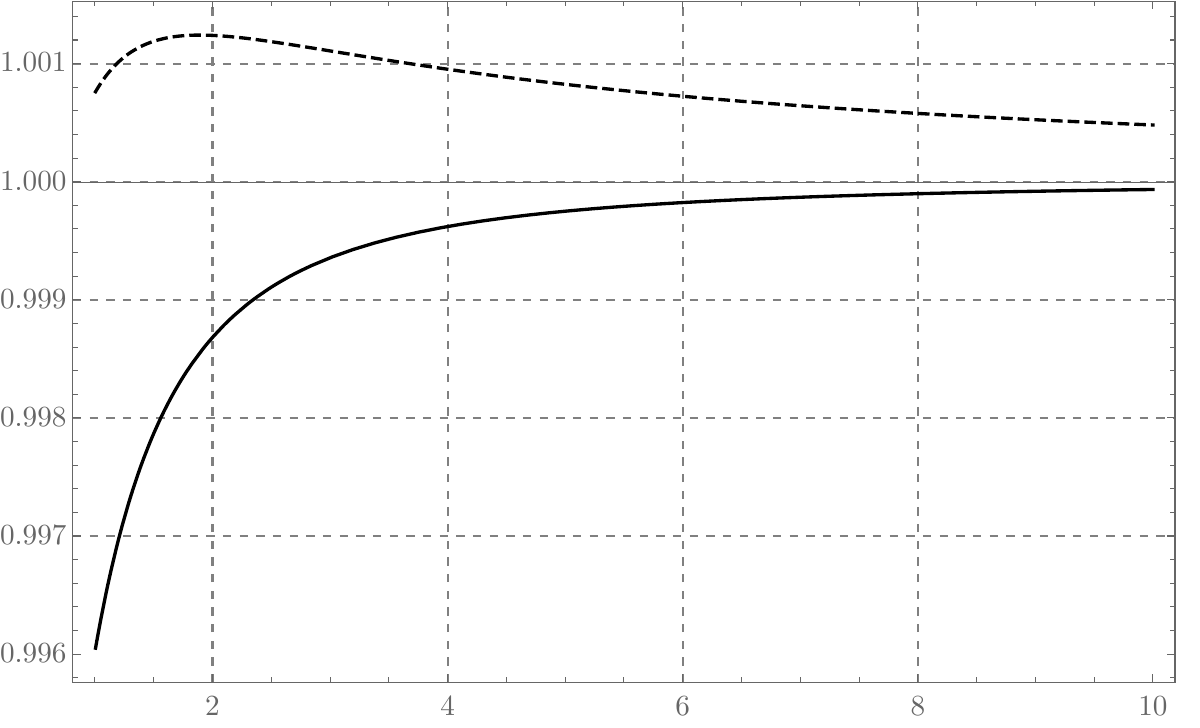}
\end{center}
\caption{Plot of $\left(\sqrt{2\pi (n+\sigma)} \left(\frac{n}{e}\right)^n\right)/\Gamma(n+1)$, as a function of $n \geq 1$, 
with $\sigma=\lsigma$ (solid) and and $\sigma=\usigma$ (dashed)
as determined by \eqref{eqn:sigma}.}
\label{fig:Stirling}
\end{figure}

\noindent Now, let us define functions
\begin{align*}
S_0(q) &:= \sum_{y=1}^\infty \frac{y^y}{y!}\, (q/e)^y,\\
S_1(q) &:= \sum_{y=1}^\infty \frac{y^{y+1}}{y!}\, (q/e)^y.
\end{align*}
Using \eqref{eqn:Stirling}, we can compute upper and lower bounds
on the values of these functions as follows. Define
\begin{align}
S_0(q,\sigma) &:= \sum_{y=1}^\infty 
\frac{y^y}{\sqrt{2\pi (y+\sigma)} \left(\frac{y}{e}\right)^y}\, (q/e)^y \label{eqn:S0sigmaDef}\\
&= \frac{1}{\sqrt{2\pi}} \sum_{y=1}^\infty \frac{q^y}{\sqrt{y+\sigma}} \nonumber\\
&= \frac{q}{\sqrt{2\pi}} \Phi(q,1/2,1+\sigma),
\label{eqn:S0sigma}
\end{align}
where $\Phi(\cdot)$ denotes the
Lerch transcendent \eqref{eqn:Lerch}. 
For the special case $\sigma=0$, this slightly simplifies to 
\[
S_0(q,0)=\frac{1}{\sqrt{2\pi}} \li_{1/2}(q),
\]
where $\li_s(\cdot)$ denotes the
polylogarithm function 
\begin{equation} \label{eqn:polylog:def}
\li_s(z):=\sum_{k=1}^\infty z^k k^{-s}.
\end{equation}
Similarly, we define
\begin{align}
S_1(q,\sigma) &:= \sum_{y=1}^\infty 
\frac{y^{y+1}}{\sqrt{2\pi (y+\sigma)} \left(\frac{y}{e}\right)^y}\, (q/e)^y  \label{eqn:S1sigmaDef}\\
&= \frac{1}{\sqrt{2\pi}} \sum_{y=1}^\infty \frac{(y+\sigma-\sigma)
q^y}{\sqrt{y+\sigma}} \nonumber\\
&= \frac{q}{\sqrt{2\pi}} \Phi(q,-1/2,1+\sigma) -\sigma S_0(q,\sigma), \label{eqn:S1sigma}
\end{align}
which, for $\sigma=0$, simplifies to 
\[
S_1(q,0)=\frac{1}{\sqrt{2\pi}} \li_{-1/2}(q).
\]
Clearly, we have
\begin{equation} \label{eqn:Sbounds}
S_0(q) \geq S_0(q,\usigma), \quad
S_1(q) \geq S_1(q,\usigma), \quad
S_0(q) \leq S_0(q,\lsigma), \quad
S_1(q) \leq S_1(q,\lsigma).
\end{equation}
Given a parameter $q \in (0,1)$, using $\eqref{eqn:PoiYfirst}$
we may write
\[
y_0=1/(1+S_0(q)),\quad \mu:=\E[Y]=y_0 S_1(q)=S_1(q)/(1+S_0(q)).
\]
Using \eqref{eqn:S0sigma} and \eqref{eqn:S1sigma},
we may now derive the estimates \eqref{eqn:PoiEstimates} in the statement. 
This completes the proof.



\section{Analytic claims}

In this section, we verify a number of stand-alone analytic claims
that have been used in the proofs.

\subsection{Claim \ref{claim:Poi:ratio}}
\label{app:Poi:ratio}

Consider the derivative of the function $g$ defined in \eqref{eqn:Poi:ratio}:
\begin{align*}
g'(y)&= \frac{1}{\sqrt{\pi}}\left( \Gamma'(y+1/2) \exp(y-y \psi(y))+
\Gamma(y+1/2) \exp(y-y \psi(y)) (1-(\psi(y)+y \psi'(y)))\right)\\
&= \frac{-1}{\sqrt{\pi}} \Gamma(y+1/2) \exp(y-y \psi(y))\,
(-\psi(y+1/2)-1+\psi(y)+y \psi'(y)),
\end{align*}
where $\psi(\cdot)$ is the digamma function.
It can be seen that the function 
\begin{equation} \label{eqn:phi}
\varphi(y) := -\psi(y+1/2)+\psi(y)+y \psi'(y)-1,
\end{equation}
depicted in Figure~\ref{fig:phi},
is positive for all $y>0$ (in fact, this function is completely monotone,
which can be proved by expressing $\varphi(y)$ using Euler's integral
representation for the digamma function).
It follows that $g'(y)<0$ for all $y>0$, and thus the ratio $g$
is strictly decreasing. 
\Mnote{more details.}
\begin{figure}[t!]
\begin{center}
\includegraphics[height=2in]{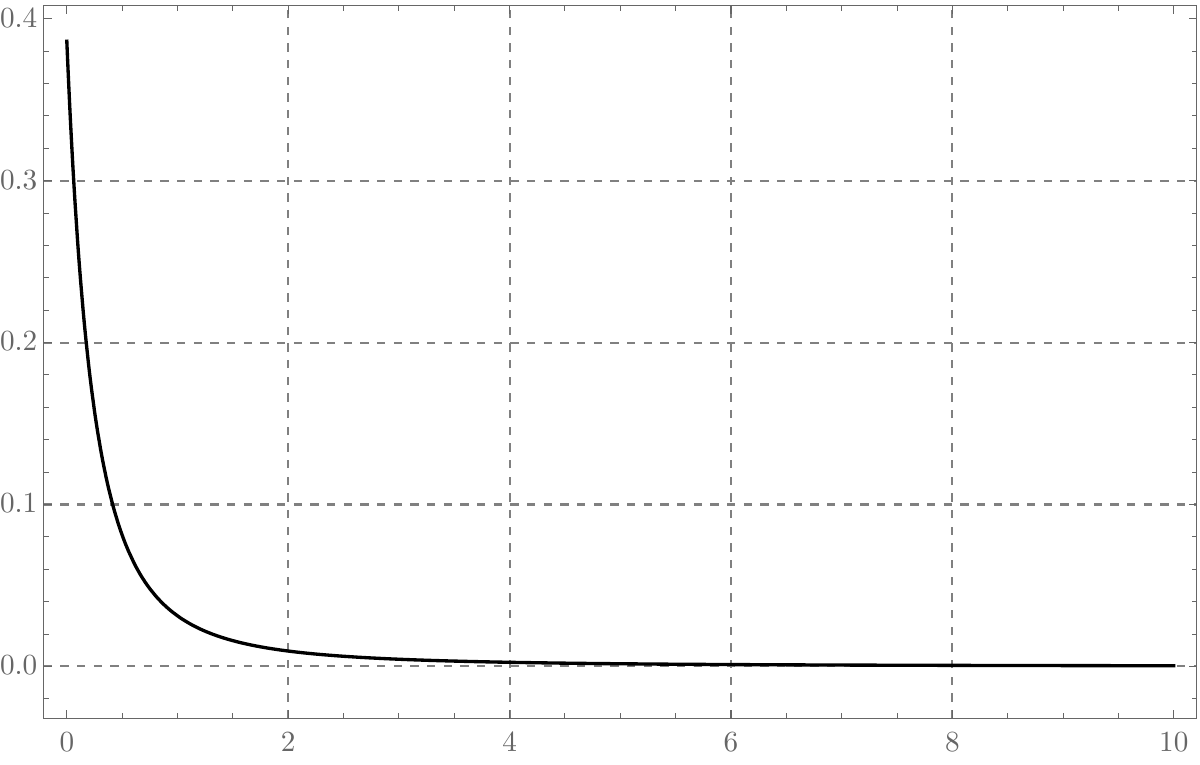}
\end{center}
\caption{Plot of the function $\varphi$ in \eqref{eqn:phi}.}
\label{fig:phi}
\end{figure}

\subsection{Claim~\ref{claim:invbin:ratio}}
\label{app:invbin:ratio}

For the special case $p=1/2$, the ratio is equal to $1$ for all $y>0$
due to the duplication formula for the gamma function:
\[
\Gamma (y)\Gamma \left(y+{\tfrac {1}{2}}\right)={\sqrt {4\pi }}\;\Gamma (2y)/4^y,
\]
which implies that,
\[
\binom{2y}{y}=\frac{2\Gamma(2y)}{\Gamma(y) \Gamma(y+1)}=\frac{4^y \Gamma(y+\frac{1}{2})}{\sqrt{\pi} \Gamma(y+1)}=
 \frac{4^y}{\sqrt{\pi}} \binom{y-\frac{1}{2}}{y} \Gamma(1/2)=
4^y \binom{y-\frac{1}{2}}{y},
\]
and the result follows by noting that $h(1/2) = \log 2$.
Since $\rho$ is always positive for $y>0$, it is increasing
(decreasing) \Iff\ $\log \rho$ is increasing (decreasing).
We can write the derivative of $\log \rho$ in $y$
by taking the derivative of each constituent term;
that is,
\[
p \frac{d}{dy}\log \rho(y) =
-h(p)-p \psi(y+1/2)-(1-p)\psi(y(1/p-1)+1) + \psi(y/p+1) =: \rho_1(y).
\]
One can verify that $\lim_{y \to \infty} \rho_1(y)=0$.
The derivative of $\rho_1(y)$, in turn, can be written as
\[
p \frac{d}{dy} \rho_1(y) =-p^2 \psi'(y+1/2)-(1-p)^2 \psi'(y(1/p-1)+1)+\psi'(y/p+1).
\]
\Mnote{more details}
This function is positive and decreasing in $y$ (in fact, completely monotone) when
$p<1/2$, identically zero when $p=1/2$,
and negative and increasing (in fact, negated completely monotone) when $p>1/2$
(this can be proved using Theorem~\ref{thm:KP16}).
It follows that $\rho_1(y)$ is increasing (and negative) when
$p<1/2$, zero when $p=1/2$, and decreasing (and positive) when $p>1/2$. 
The claim follows.

\subsection{Claim~\ref{claim:binomial}} 
\label{app:claim:binomial}


We do not present a rigorous proof that the particular choices of
$\overline{\alpha}$ and $\underline{\alpha}$
in the statement are valid, but rather, provide 
a convincing argument. It would be an interesting question to provide
a complete proof.
Consider the log-ratio
\begin{equation} \label{eqn:claim:binomial:logratio}
R_{p,\alpha}(y) := \log \binom{y/p}{y}-y h(p)/p+\frac{1}{2} \log(2\pi((1-p)y+\alpha)).
\end{equation}
By Stirling's approximation, it follows that for any $p \in (0,1)$ and any fixed $\alpha$,
we have \[\lim_{y \to \infty} R_{p,\alpha}(y)=0.\]
It can be proved, using 
Theorem~\ref{thm:KP16},
that the
function $-R_{p,0}$ is completely monotone for all $p \in (0,1)$. Therefore, even for $y\geq 0$,
the constant $\alpha=0$ provably provides an upper bound for $\binom{y/p}{y}$. That is,
for all $p \in (0,1)$,
\[
\binom{y/p}{p} \leq \frac{\exp(y h(p)/p)}{\sqrt{2\pi (1-p)y}}.
\]
We now have the more refined task of finding choices of $\alpha$ that makes the above log-ratio always positive (negative)
for all $p \in (0,1)$ and all $y \geq 1$ (rather than $y \geq 0$).
It is worthwhile to note that the $j$th derivative of the function $-R_{p,\alpha}(y)$ (for $j \geq 1$) can
be computed, letting $q := p/(1-p)$, as 
\[
\frac{A_{j-1} (-p)^j (1-p)^j + (a+y(1-p))^j ( p^j \psi^{(j-1)}(1+y) -\psi^{(j-1)}(1+\frac{y}{p}) + (1-p)^j \psi^{(j-1)}(1+\frac{y}{q}) )}{p^j (a+y(1-p))^j},
\]
where $\psi^{(j)}$ is the polygamma function, $A_j$ is the number of even permutations on $n$ items
(OEIS~A001710), except for the special cases $A_1:=A_0:=1/2$, and that for the first derivative,
an additional constant $h(p)/p$ is to be added to the above expression to obtain the correct derivative.
Figure~\ref{fig:IBinLerchConst} depicts the log-ratio function and its derivative for various choices of $\alpha$ and $p$.
At $y=1$, we have \[R_{p,\alpha}(1) = (1-p) \log(1-p)/p+\log(2\pi(1-p+\alpha))/2. \]
By setting $R_{p,\alpha}(1)=0$, and solving for $\alpha$, we obtain the solution
\[
\alpha_0 := -1+p+\frac{(1-p)^{2-2/p}}{2\pi}.
\]
The value of the solution $\alpha_0$ as a function of $p$ is plotted in Figure~\ref{fig:IBinLerchConst}
(the limit at $p \to 0$ is $e^2/(2\pi)-1 \approx 0.176$, and the minimum occurs a $p \approx 0.405$
at which point we have $\alpha_0 \approx 0.136$).
Note that $\alpha_0$ is the transition point for whether the plot for the log-ratio lies above or below
the horizontal axis around $y=1$, and we see that this transition always occurs within $\alpha \in (0.13, 0.18)$.
By choosing the value of $\alpha$ sufficiently outside this critical region (in particular, 
for the choices of $\underline{\alpha}$ and $\overline{\alpha}$ in the statement), 
we may ensure that the log-ratio
remains above, or below, zero for $y \geq 1$ until the asymptotic convergence towards zero dominates the behavior
of the function. The log-ratios for the chosen values of $\underline{\alpha}$ and $\overline{\alpha}$
and various choices of $p$ are plotted as functions of $y$ in Figure~\ref{fig:IBinLerchConst}.
We observe that, apart from the cases $p \to 0$ and $p \to 1$, the log-ratio
$R_{p,\alpha}(y+1)$ or its negation appear to be not only positive, but also
completely monotone for the chosen values of $\alpha$.

\begin{figure}[t!]
\begin{center}
\includegraphics[width=0.49 \columnwidth]{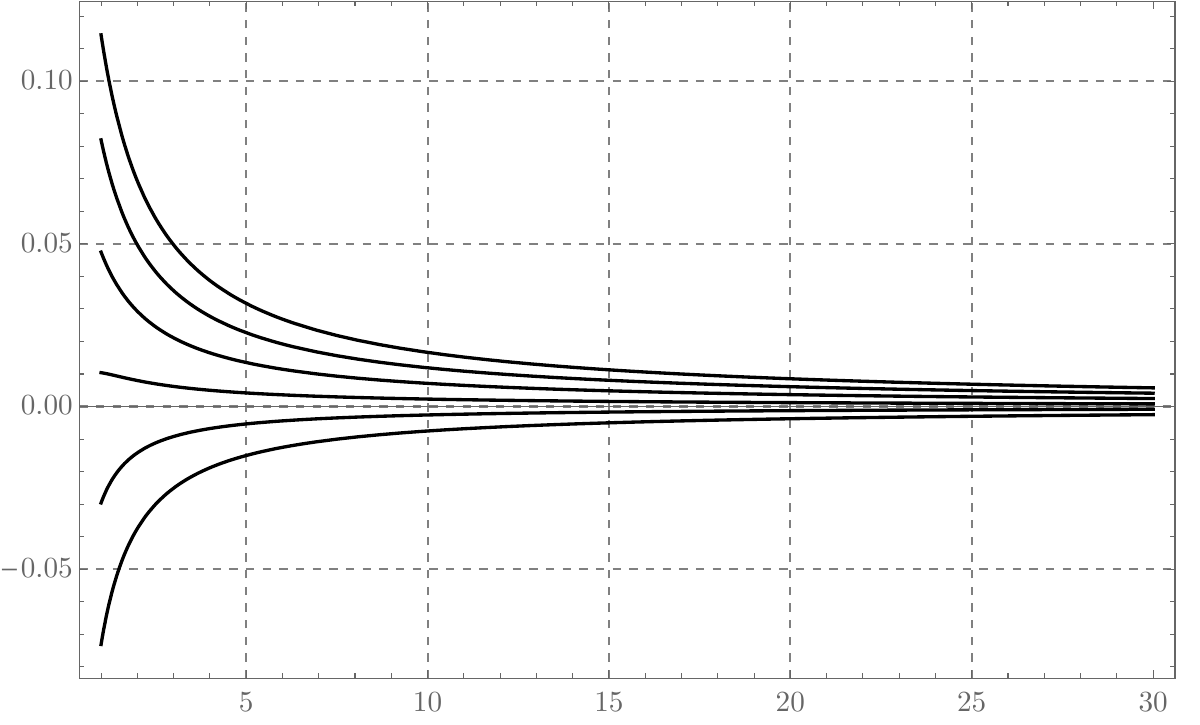} 
\includegraphics[width=0.49 \columnwidth]{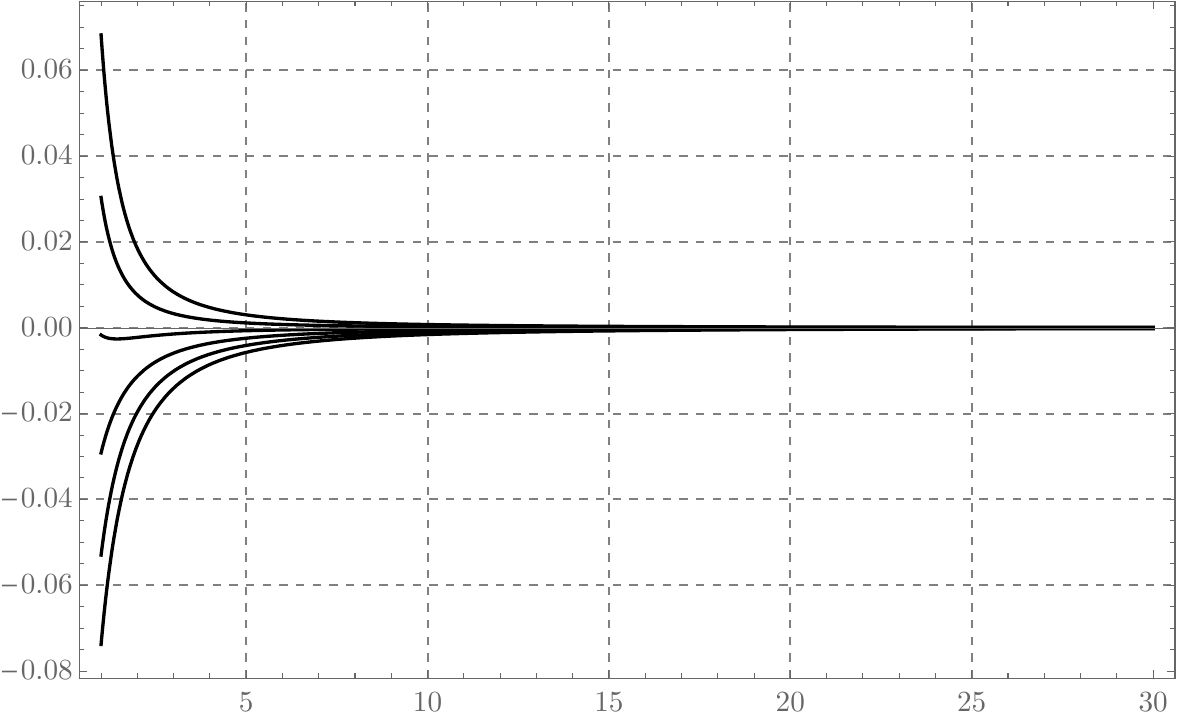} \\[2mm]
\includegraphics[width=0.49 \columnwidth]{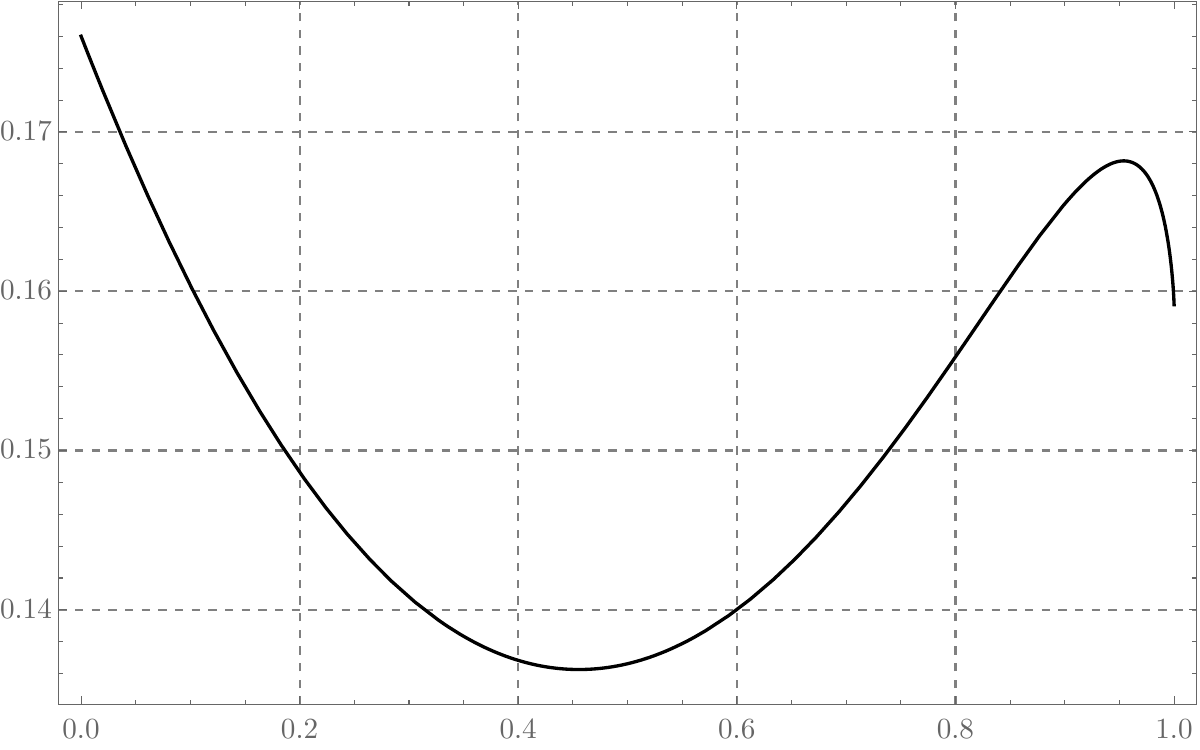}
\includegraphics[width=0.49 \columnwidth]{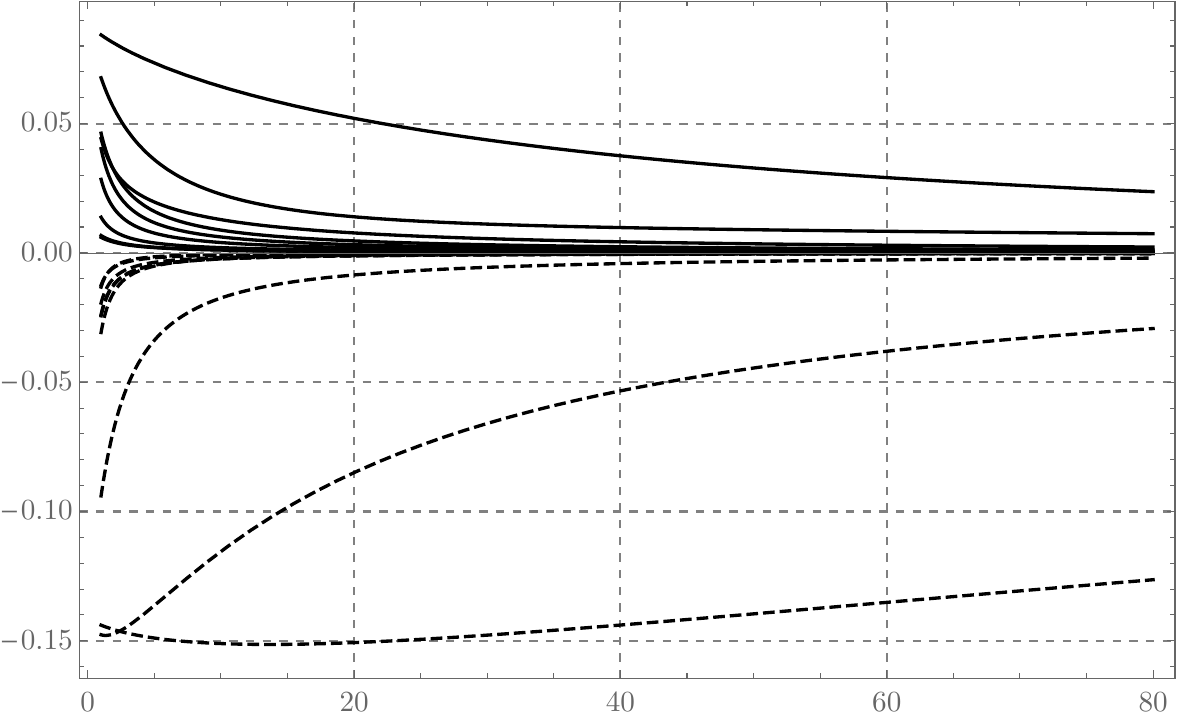}
\end{center}
\caption{Plots related to the log-ratio function $R_{p,\alpha}$ in \eqref{eqn:claim:binomial:logratio}.
Top: plots of $R_{0.5,\alpha}(y)$, as functions of $y$, 
for $\alpha=0.1,0.15,0.20,0.25,0.30$ (left, from the lowest to the highest curve)
and their derivatives in $y$ (right, from the highest to the lowest curve).
Bottom left: plot of $\alpha_0 = -1+p+(1-p)^{2-2/p}/(2\pi)$ as a function
of $p$. Bottom right: plots of $R_{p,\underline{\alpha}}(y)$ (solid) and 
$R_{p,\overline{\alpha}}(y)$ (dashed) in $y$ for $p=10^{-3},10^{-2},0.1,0.3,0.5,0.7,0.9,0.99,0.999$.
}
\label{fig:IBinLerchConst}
\end{figure}

\subsection{Claim~\ref{claim:binom:monotone}}
\label{app:binom:monotone}

In this section, we prove Claim~\ref{claim:binom:monotone}
that is restated below:
\begin{claimS} 
The function 
\[f(y) := \log \binom{y}{p y} = \log \frac{\Gamma(y+1)}{\Gamma(py+1)\Gamma((1-p)y+1)},\]
defined for $p \in (0,1)$ and $x>0$, is completely monotone. That is,
for all integers $j=0,1, \ldots$, $(-1)^j f^{(j)}(y) > 0$ for all $y>0$. 
\end{claimS}

\begin{proof}
By definition, $f(y)>0$ for all $y>0$. Moreover, 
$f'(y)=\psi(y+1)-p\psi(py+1)-(1-p)\psi((1-p)y+1)$,
which is positive for all $y>0$ by the concavity of the digamma
function. Thus, it suffices to show that $f''(y)$ is
completely monotone.
This claim is a special case of the following result proved in \cite{ref:KP16}:
\begin{thm}\cite{ref:KP16} \label{thm:KP16}
Let
\[
F(y):= \frac{\Pi_{i=1}^m \Gamma(A_i y+a_i)}{\Pi_{j=1}^n \Gamma(B_j+b_j)}.
\]
Then, 
$(\log F(y))''$ is completely monotone \Iff\ the function
\[
P(u) := \sum_{i=1}^m \frac{\exp(-a_i u/A_i)}{1-\exp(-u/A_i)}-
\sum_{i=1}^n \frac{\exp(-b_j u/B_j)}{1-\exp(-u/B_j)}
\]
is non-negative for all $u>0$.  \qed
\end{thm}
\noindent For our application, we have
\[
P(u) = \frac{1}{e^u-1}-\frac{1}{e^{u/p}-1}-\frac{1}{e^{u/q}-1}.
\]
In order to show that $P(u)\geq 0$, it suffices to verify that
\[
P_0(u) := \frac{p}{e^u-1}-\frac{1}{e^{u/p}-1} \geq 0
\]
for all $u>0$, or equivalently, that $p e^{u/p}-e^u +1-p \geq 0$ for
all $u>0$. The left hand side is zero at $u=0$ and has a positive
derivative in $u$ (which is $e^{u/p}-e^u$) for all $u \geq 0$. 
The result follows.
\end{proof}

\subsection{An analytically simple deletion capacity upper bound for $p \leq 1/2$}
\label{app:claim:cber:analytical}

In this appendix, we show that the result of \eqref{eqn:CberHalf}
for $p=1/2$ can be extended to smaller values of $p$ as well, leading
to simple and analytic capacity upper bound expressions. 
In particular, recalling the notation of Section~\ref{sec:del:derivation}, we observe the following:
\begin{claim} \label{claim:cber:analytical}
Let $\CBer(p,q)$ be defined with respect to the inverse binomial distribution for $Y$.
Then, for all $p \leq 1/2$, and all $q \in (0,1)$,  the following upper bound holds:
\begin{equation} \label{eqn:CberSmall}
\CBer(p,q) \leq \frac{\beta_0 h(q)}{2-(3-2 \beta_1) q},
\end{equation}
where $\beta_0$ and $\beta_1$ are defined according to 
 \eqref{eqn:beta0} and \eqref{eqn:beta1}, which are both equal to $1$ for $p=1/2$. 
\end{claim}
\begin{proof}
In order to derive the claim, consider  
\[
f(q) :=  \CBer(p,q) = \frac{-\mu(q) \log q-\log y_0(q)}{1+\mu(q)},
\]
where $\mu(q)$ and $y_0(q)$ are, respectively, the mean and the normalizing constant
of the inverse binomial distribution \eqref{eqn:invbin} for the given parameters
$p$ and $q$. Using Corollary~\ref{coro:invbin:negBinApprox}, we may
upper bound $f(q)$ by the elementary function $g(q)$ defined below as
\[
g(q) := \frac{-\frac{\bUp q \log q }{2 (1-q) (\sqrt{1-q}+\bDown(1-\sqrt{1-q}))}-\log\left(1+\bUp\left(\frac{1}{\sqrt{1-q}}-1\right)\right)}{1+\frac{\bDown q}{2 (1-q) (\sqrt{1-q}+\bUp(1-\sqrt{1-q}))}},
\]
where $\bDown$ (resp., $\bUp$) is the minimum (resp., the maximum)
of the two constants 
$\beta_0 = (2/p) \exp(-h(p)/p)$ and  $\beta_1 = {1}/{\sqrt{2(1-p)}}$ 
(so that, for $p \leq 1/2$, $\bUp = \beta_0$ and $\bDown = \beta_1$).
Define the ratio 
\begin{equation} \label{eqn:CBerRatio}
R(q) := h(q)/g(q).
\end{equation}
This ratio is plotted in Figure~\ref{fig:CBerRatio}.
A direct calculation shows that
\[
\lim_{q \to 0} R(q) = 2/\bUp, \quad \lim_{q \to 1} R(q) = \bDown/\bUp, \quad
\lim_{q \to 0 }\frac{\partial R}{\partial q} = (2 \bDown-3)/\bUp.
\]

Suppose $p \leq 1/2$, in which case it can be observed that $R(q)$ lies above its
tangent line at $q \to 0$
(and for $p=1/2$, $R(q)$ actually coincides with the tangent line). Therefore,
\[
R(q) = \frac{h(q)}{g(q)} \geq \frac{2+(2 \bDown-3)q}{\bUp} = \frac{2-(3-2 \beta_1)q}{\beta_0},
\]
and the claim follows. 
\end{proof}

\begin{figure}[t!] 
\begin{center}
\includegraphics[width=0.49 \columnwidth]{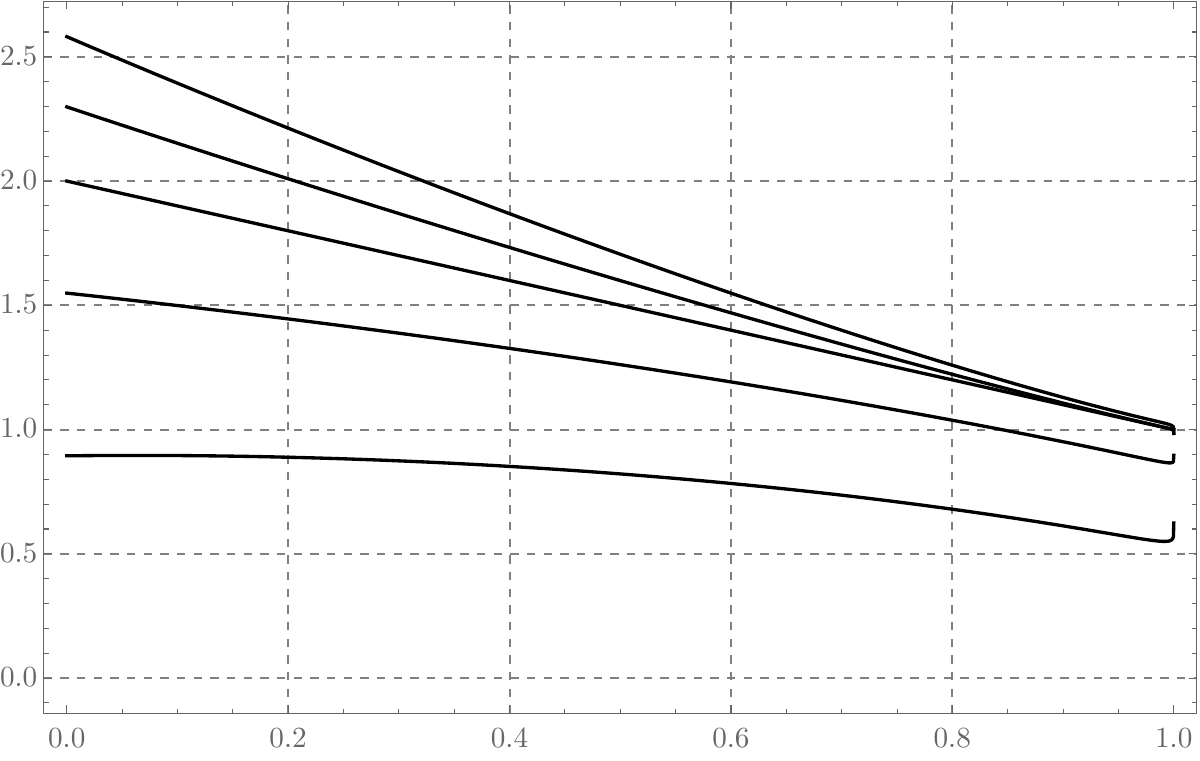}
\includegraphics[width=0.49 \columnwidth]{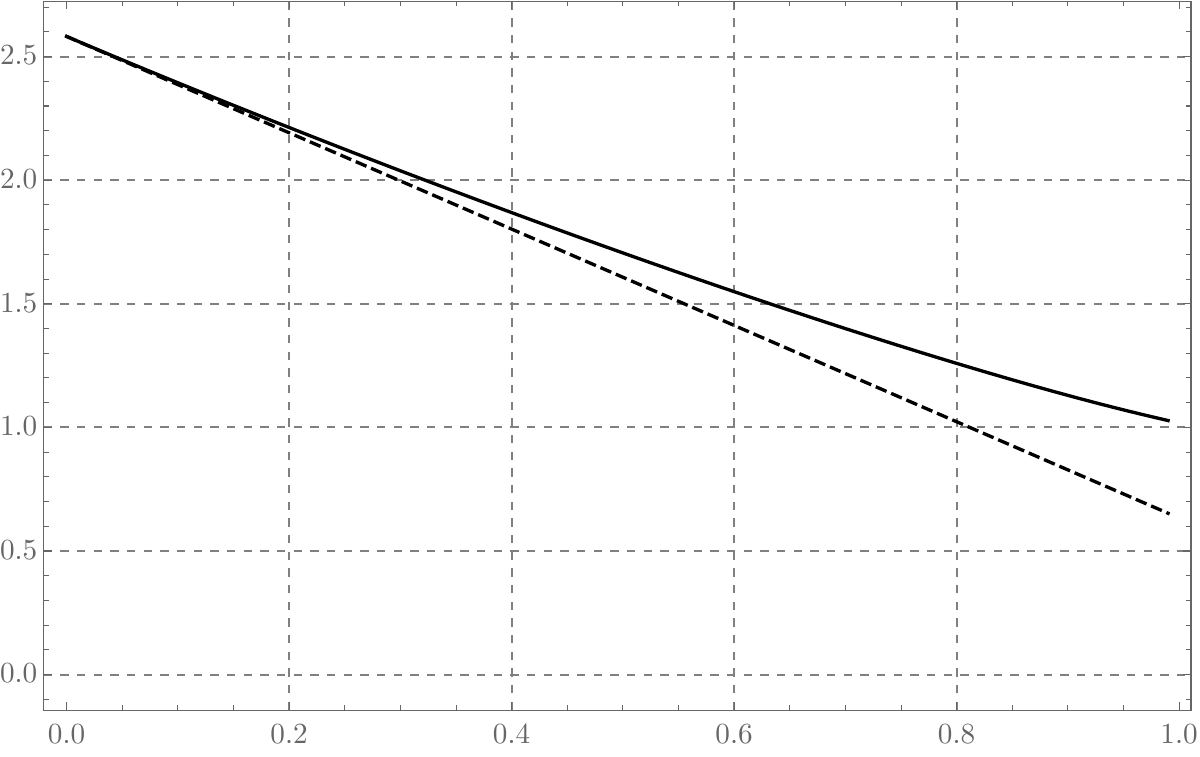}
\end{center}
\caption{Left: Plots of the ratio $R(q)$, defined in \eqref{eqn:CBerRatio}, for $p=0.1,0.3,0.5,0.7,0.9$ (from the highest
to the lowest graph), as a function of $q$. Right: Plot of $R(q)$ for $p=0.1$ and its 
corresponding tangent line (dashed) at $q=0$.}
\label{fig:CBerRatio}
\end{figure}


Using the result of Claim~\ref{claim:cber:analytical}, we can now prove the following:
\begin{coro} \label{coro:cber:analytical}
Let $p \leq 1/2$, $\beta_0$ and $\beta_1$ be defined according to 
 \eqref{eqn:beta0} and \eqref{eqn:beta1}, and $\ch$ be the deletion
 channel with deletion probability $d=1-p$. Then,
 \[
 \capa(\ch) \leq \frac{ (1-d) \beta_0 h(q^\star)}{2-(3-2 \beta_1)q^\star},
 \]
 where $q^\star \in (0,1)$ is the solution to
$
q^\star=(1-q^\star)^{\beta_1-1/2}.
$
\end{coro}

\begin{proof}
Recall that
\[
 \capa(\ch) \leq (1-d) \sup_{q \in (0,1)} \CBer(p,q) \leq
 \sup_{q \in (0,1)} \frac{\beta_0 h(q)}{2-(3-2 \beta_1) q},
\]
where the second inequality is due to Claim~\ref{claim:cber:analytical}.
Let $q^\star$ be the choice of $q$ that maximizes 
the right  
hand side of \eqref{eqn:CberSmall}.
The derivative of the right hand side of \eqref{eqn:CberSmall} in $q$
is equal to
\[
\beta_0 \frac{-2 \log q -(1-2 \beta_1) \log(1-q)}{(2-(3-2 \beta_1)q)^2}.
\]
By equating this derivative to zero, we see that  $q^\star$ is
the solution to
$
q^\star=(1-q^\star)^{\beta_1-1/2},
$
and the result follows.
\end{proof}

\begin{table}[!p]
\caption{Mean and normalizing constants for the distributions defined for
the Poisson channel by \eqref{eqn:PoiYfirst} $(i=1)$ and the 
digamma distribution 
\eqref{eqn:PoiYre} $(i=2)$, for various choices of the parameter $q \in (0,1)$.
For each $i \in \{1,2\}$, we have used the notation 
$\ell_i:=-\log y_0$ and $\mu_i:=\E[Y]$ for the corresponding distribution.}
 \vspace{5mm}
 \begin{center}
 \resizebox{0.9\columnwidth}{!}{%
\begin{tabular}{||c||c|c||c|c||}
\hline
$q$ & $\ell_1$ & $\mu_1$ & $\ell_2$ & $\mu_2$  \\  
\hline \hline 
0.01 & 0.003699 & 0.003719 & 0.002079 & 0.002093 \\ 
0.02 & 0.007439 & 0.007523 & 0.004186 & 0.004242 \\ 
0.03 & 0.011222 & 0.011413 & 0.006321 & 0.006450 \\ 
0.04 & 0.015049 & 0.015393 & 0.008486 & 0.008718 \\ 
0.05 & 0.018919 & 0.019464 & 0.010681 & 0.011049 \\ 
\hline 
0.06 & 0.022835 & 0.023631 & 0.012906 & 0.013445 \\ 
0.07 & 0.026797 & 0.027897 & 0.015163 & 0.015908 \\ 
0.08 & 0.030806 & 0.032264 & 0.017452 & 0.018440 \\ 
0.09 & 0.034863 & 0.036737 & 0.019774 & 0.021045 \\ 
0.10 & 0.038970 & 0.041319 & 0.022129 & 0.023725 \\ 
\hline 
0.11 & 0.043128 & 0.046014 & 0.024519 & 0.026482 \\ 
0.12 & 0.047337 & 0.050825 & 0.026944 & 0.029321 \\ 
0.13 & 0.051599 & 0.055758 & 0.029406 & 0.032243 \\ 
0.14 & 0.055916 & 0.060816 & 0.031905 & 0.035252 \\ 
0.15 & 0.060288 & 0.066004 & 0.034442 & 0.038352 \\ 
\hline 
0.16 & 0.064717 & 0.071327 & 0.037019 & 0.041546 \\ 
0.17 & 0.069204 & 0.076789 & 0.039636 & 0.044838 \\ 
0.18 & 0.073751 & 0.082397 & 0.042295 & 0.048232 \\ 
0.19 & 0.078360 & 0.088156 & 0.044996 & 0.051732 \\ 
0.20 & 0.083031 & 0.094071 & 0.047741 & 0.055343 \\ 
\hline 
0.21 & 0.087768 & 0.100150 & 0.050530 & 0.059069 \\ 
0.22 & 0.092570 & 0.106397 & 0.053367 & 0.062915 \\ 
0.23 & 0.097441 & 0.112821 & 0.056251 & 0.066887 \\ 
0.24 & 0.102381 & 0.119428 & 0.059183 & 0.070989 \\ 
0.25 & 0.107394 & 0.126226 & 0.062167 & 0.075228 \\ 
\hline 
0.26 & 0.112480 & 0.133223 & 0.065202 & 0.079610 \\ 
0.27 & 0.117642 & 0.140427 & 0.068291 & 0.084142 \\ 
0.28 & 0.122883 & 0.147848 & 0.071436 & 0.088829 \\ 
0.29 & 0.128204 & 0.155495 & 0.074637 & 0.093680 \\ 
0.30 & 0.133607 & 0.163377 & 0.077897 & 0.098702 \\ 
\hline 
0.31 & 0.139096 & 0.171506 & 0.081218 & 0.103903 \\ 
0.32 & 0.144673 & 0.179893 & 0.084601 & 0.109293 \\ 
0.33 & 0.150341 & 0.188549 & 0.088049 & 0.114879 \\ 
0.34 & 0.156101 & 0.197487 & 0.091564 & 0.120673 \\ 
0.35 & 0.161959 & 0.206722 & 0.095148 & 0.126685 \\ 
\hline 
0.36 & 0.167915 & 0.216266 & 0.098804 & 0.132925 \\ 
0.37 & 0.173975 & 0.226135 & 0.102534 & 0.139406 \\ 
0.38 & 0.180140 & 0.236346 & 0.106341 & 0.146141 \\ 
0.39 & 0.186415 & 0.246916 & 0.110227 & 0.153142 \\ 
0.40 & 0.192804 & 0.257863 & 0.114195 & 0.160426 \\ 
\hline 
0.41 & 0.199310 & 0.269207 & 0.118249 & 0.168006 \\ 
0.42 & 0.205937 & 0.280969 & 0.122392 & 0.175901 \\ 
0.43 & 0.212691 & 0.293172 & 0.126626 & 0.184128 \\ 
0.44 & 0.219575 & 0.305840 & 0.130957 & 0.192706 \\ 
0.45 & 0.226594 & 0.319000 & 0.135387 & 0.201657 \\ 
\hline 
0.46 & 0.233754 & 0.332679 & 0.139921 & 0.211002 \\ 
0.47 & 0.241060 & 0.346908 & 0.144563 & 0.220766 \\ 
0.48 & 0.248518 & 0.361719 & 0.149317 & 0.230976 \\ 
0.49 & 0.256134 & 0.377147 & 0.154188 & 0.241659 \\ 
0.50 & 0.263914 & 0.393231 & 0.159182 & 0.252846 \\ 
\hline 

\hline \hline
$q$ & $\ell_1$ & $\mu_1$ & $\ell_2$ & $\mu_2$  \\  
\hline
\end{tabular}
\begin{tabular}{||c||c|c||c|c||}
\hline
$q$ & $\ell_1$ & $\mu_1$ & $\ell_2$ & $\mu_2$  \\  
\hline \hline 
0.51 & 0.271866 & 0.410012 & 0.164304 & 0.264571 \\ 
0.52 & 0.279996 & 0.427534 & 0.169560 & 0.276869 \\ 
0.53 & 0.288312 & 0.445846 & 0.174955 & 0.289781 \\ 
0.54 & 0.296823 & 0.465001 & 0.180497 & 0.303350 \\ 
0.55 & 0.305537 & 0.485057 & 0.186193 & 0.317622 \\ 
\hline 
0.56 & 0.314465 & 0.506076 & 0.192049 & 0.332650 \\ 
0.57 & 0.323615 & 0.528128 & 0.198076 & 0.348489 \\ 
0.58 & 0.332999 & 0.551287 & 0.204280 & 0.365202 \\ 
0.59 & 0.342629 & 0.575638 & 0.210672 & 0.382858 \\ 
0.60 & 0.352517 & 0.601271 & 0.217262 & 0.401532 \\ 
\hline 
0.61 & 0.362676 & 0.628288 & 0.224060 & 0.421309 \\ 
0.62 & 0.373121 & 0.656801 & 0.231080 & 0.442280 \\ 
0.63 & 0.383869 & 0.686933 & 0.238332 & 0.464551 \\ 
0.64 & 0.394935 & 0.718822 & 0.245832 & 0.488235 \\ 
0.65 & 0.406338 & 0.752623 & 0.253595 & 0.513463 \\ 
\hline 
0.66 & 0.418099 & 0.788508 & 0.261637 & 0.540379 \\ 
0.67 & 0.430240 & 0.826670 & 0.269976 & 0.569147 \\ 
0.68 & 0.442784 & 0.867328 & 0.278633 & 0.599951 \\ 
0.69 & 0.455759 & 0.910727 & 0.287630 & 0.633001 \\ 
0.70 & 0.469192 & 0.957149 & 0.296989 & 0.668534 \\ 
\hline 
0.71 & 0.483117 & 1.006911 & 0.306740 & 0.706822 \\ 
0.72 & 0.497569 & 1.060380 & 0.316910 & 0.748178 \\ 
0.73 & 0.512586 & 1.117973 & 0.327534 & 0.792962 \\ 
0.74 & 0.528214 & 1.180178 & 0.338648 & 0.841592 \\ 
0.75 & 0.544500 & 1.247557 & 0.350295 & 0.894554 \\ 
\hline 
0.76 & 0.561501 & 1.320769 & 0.362520 & 0.952420 \\ 
0.77 & 0.579279 & 1.400589 & 0.375377 & 1.015862 \\ 
0.78 & 0.597905 & 1.487933 & 0.388928 & 1.085679 \\ 
0.79 & 0.617460 & 1.583895 & 0.403240 & 1.162828 \\ 
0.80 & 0.638037 & 1.689789 & 0.418395 & 1.248462 \\ 
\hline 
0.81 & 0.659744 & 1.807209 & 0.434485 & 1.343985 \\ 
0.82 & 0.682705 & 1.938106 & 0.451619 & 1.451121 \\ 
0.83 & 0.707068 & 2.084896 & 0.469924 & 1.572012 \\ 
0.84 & 0.733006 & 2.250603 & 0.489554 & 1.709350 \\ 
0.85 & 0.760730 & 2.439067 & 0.510690 & 1.866565 \\ 
\hline 
0.86 & 0.790489 & 2.655234 & 0.533555 & 2.048092 \\ 
0.87 & 0.822593 & 2.905584 & 0.558421 & 2.259765 \\ 
0.88 & 0.857427 & 3.198777 & 0.585631 & 2.509412 \\ 
0.89 & 0.895475 & 3.546649 & 0.615616 & 2.807777 \\ 
0.90 & 0.937364 & 3.965808 & 0.648940 & 3.169999 \\ 
\hline 
0.91 & 0.983921 & 4.480299 & 0.686347 & 3.618099 \\ 
0.92 & 1.036269 & 5.126285 & 0.728855 & 4.185347 \\ 
0.93 & 1.095987 & 5.960740 & 0.777904 & 4.924423 \\ 
0.94 & 1.165396 & 7.078892 & 0.835620 & 5.923839 \\ 
0.95 & 1.248104 & 8.652671 & 0.905333 & 7.344293 \\ 
\hline 
0.96 & 1.350183 & 11.027073 & 0.992684 & 9.510209 \\ 
0.97 & 1.483060 & 15.010052 & 1.108370 & 13.186490 \\ 
0.98 & 1.672514 & 23.035876 & 1.276746 & 20.695693 \\ 
0.99 & 2.001316 & 47.343266 & 1.576877 & 43.831689 \\ 
& & & & \\ 
\hline \hline
$q$ & $\ell_1$ & $\mu_1$ & $\ell_2$ & $\mu_2$  \\  
\hline
\end{tabular}
}
 \end{center}
%
%
\label{tab:PoiMeanY}
\end{table}

\begin{table}[!p]
\caption{Capacity upper bounds for the Poisson-repeat channel with
deletion probability $d=1-p=e^{-\lam}$,
as plotted in Figure~\ref{fig:PoiPlot}. For each
$d$, the quantity $(1-d) c_i$ is the capacity upper bound when 
method $i =1,2,3,4$ is used, and $q_i$ is the choice of $q$
that maximizes the function under the supremum
in \eqref{eqn:Poi:capa:general} for the chosen method. The methods
are: The digamma distribution \eqref{eqn:PoiYre} for $Y$ ($i=1$,
giving the strongest bounds);
The elementary upper bounds of \eqref{eqn:thm:Poi:capa:approx:eqns}
on the parameters of the digamma distribution \eqref{eqn:PoiYre} $(i=2)$;
The distribution defined by \eqref{eqn:PoiYfirst} for $Y$ $(i=3)$;
The analytic upper bounds of Theorem~\ref{thm:PoiChUpperAnalytic}
for the distribution of $Y$ defined by \eqref{eqn:PoiYfirst} $(i=4)$.
Note that the methods $i=3$ and $4$ give trivial results
for sufficiently small $d$. %
 }
 \begin{center} 
\resizebox{\columnwidth}{!}{%
\begin{tabular}{||c||c|c||c|c||c|c||c|c||}
\hline
$d$ & $c_1$ & $q_1$ & $c_2$ & $q_2$ & $c_3$ & $q_3$ & $c_4$ & $q_4$ \\  
\hline \hline 
0.01 & 0.849 & 0.881 & 0.872 & 0.883 & 1.095 & 0.849 & 1.096 & 0.849 \\ 
0.02 & 0.804 & 0.869 & 0.826 & 0.871 & 1.038 & 0.834 & 1.040 & 0.834 \\ 
0.03 & 0.775 & 0.861 & 0.797 & 0.864 & 1.002 & 0.825 & 1.004 & 0.824 \\ 
0.04 & 0.754 & 0.855 & 0.776 & 0.857 & 0.975 & 0.817 & 0.977 & 0.817 \\ 
0.05 & 0.737 & 0.850 & 0.758 & 0.852 & 0.954 & 0.810 & 0.955 & 0.810 \\ 
\hline 
0.06 & 0.723 & 0.845 & 0.744 & 0.848 & 0.936 & 0.805 & 0.937 & 0.804 \\ 
0.07 & 0.711 & 0.841 & 0.731 & 0.844 & 0.920 & 0.800 & 0.922 & 0.799 \\ 
0.08 & 0.700 & 0.837 & 0.720 & 0.840 & 0.906 & 0.795 & 0.908 & 0.795 \\ 
0.09 & 0.690 & 0.834 & 0.710 & 0.837 & 0.894 & 0.791 & 0.895 & 0.790 \\ 
0.10 & 0.681 & 0.831 & 0.701 & 0.833 & 0.882 & 0.787 & 0.884 & 0.787 \\ 
\hline 
0.11 & 0.673 & 0.828 & 0.693 & 0.831 & 0.872 & 0.783 & 0.874 & 0.783 \\ 
0.12 & 0.666 & 0.825 & 0.685 & 0.828 & 0.863 & 0.780 & 0.864 & 0.779 \\ 
0.13 & 0.659 & 0.822 & 0.678 & 0.825 & 0.854 & 0.776 & 0.855 & 0.776 \\ 
0.14 & 0.652 & 0.820 & 0.672 & 0.823 & 0.845 & 0.773 & 0.847 & 0.773 \\ 
0.15 & 0.646 & 0.818 & 0.665 & 0.820 & 0.837 & 0.770 & 0.839 & 0.770 \\ 
\hline 
0.16 & 0.640 & 0.815 & 0.660 & 0.818 & 0.830 & 0.768 & 0.832 & 0.767 \\ 
0.17 & 0.635 & 0.813 & 0.654 & 0.816 & 0.823 & 0.765 & 0.825 & 0.765 \\ 
0.18 & 0.630 & 0.811 & 0.649 & 0.814 & 0.817 & 0.762 & 0.818 & 0.762 \\ 
0.19 & 0.625 & 0.809 & 0.644 & 0.812 & 0.810 & 0.760 & 0.812 & 0.760 \\ 
0.20 & 0.620 & 0.807 & 0.639 & 0.810 & 0.804 & 0.757 & 0.806 & 0.757 \\ 
\hline 
0.21 & 0.616 & 0.805 & 0.634 & 0.808 & 0.798 & 0.755 & 0.800 & 0.755 \\ 
0.22 & 0.612 & 0.803 & 0.630 & 0.806 & 0.793 & 0.753 & 0.795 & 0.753 \\ 
0.23 & 0.607 & 0.801 & 0.626 & 0.804 & 0.788 & 0.751 & 0.789 & 0.750 \\ 
0.24 & 0.603 & 0.800 & 0.622 & 0.802 & 0.783 & 0.748 & 0.784 & 0.748 \\ 
0.25 & 0.600 & 0.798 & 0.618 & 0.801 & 0.778 & 0.746 & 0.779 & 0.746 \\ 
\hline 
0.26 & 0.596 & 0.796 & 0.614 & 0.799 & 0.773 & 0.744 & 0.774 & 0.744 \\ 
0.27 & 0.592 & 0.795 & 0.611 & 0.798 & 0.768 & 0.742 & 0.770 & 0.742 \\ 
0.28 & 0.589 & 0.793 & 0.607 & 0.796 & 0.764 & 0.741 & 0.765 & 0.740 \\ 
0.29 & 0.586 & 0.792 & 0.604 & 0.794 & 0.760 & 0.739 & 0.761 & 0.738 \\ 
0.30 & 0.583 & 0.790 & 0.600 & 0.793 & 0.756 & 0.737 & 0.757 & 0.737 \\ 
\hline 
0.31 & 0.579 & 0.789 & 0.597 & 0.791 & 0.751 & 0.735 & 0.753 & 0.735 \\ 
0.32 & 0.576 & 0.787 & 0.594 & 0.790 & 0.748 & 0.733 & 0.749 & 0.733 \\ 
0.33 & 0.573 & 0.786 & 0.591 & 0.789 & 0.744 & 0.732 & 0.745 & 0.731 \\ 
0.34 & 0.571 & 0.784 & 0.588 & 0.787 & 0.740 & 0.730 & 0.742 & 0.730 \\ 
0.35 & 0.568 & 0.783 & 0.585 & 0.786 & 0.736 & 0.728 & 0.738 & 0.728 \\ 
\hline 
0.36 & 0.565 & 0.782 & 0.583 & 0.785 & 0.733 & 0.727 & 0.734 & 0.727 \\ 
0.37 & 0.562 & 0.780 & 0.580 & 0.783 & 0.730 & 0.725 & 0.731 & 0.725 \\ 
0.38 & 0.560 & 0.779 & 0.577 & 0.782 & 0.726 & 0.724 & 0.728 & 0.723 \\ 
0.39 & 0.557 & 0.778 & 0.575 & 0.781 & 0.723 & 0.722 & 0.724 & 0.722 \\ 
0.40 & 0.555 & 0.777 & 0.572 & 0.780 & 0.720 & 0.721 & 0.721 & 0.720 \\ 
\hline 
0.41 & 0.553 & 0.775 & 0.570 & 0.778 & 0.717 & 0.719 & 0.718 & 0.719 \\ 
0.42 & 0.550 & 0.774 & 0.567 & 0.777 & 0.714 & 0.718 & 0.715 & 0.718 \\ 
0.43 & 0.548 & 0.773 & 0.565 & 0.776 & 0.711 & 0.716 & 0.712 & 0.716 \\ 
0.44 & 0.546 & 0.772 & 0.563 & 0.775 & 0.708 & 0.715 & 0.709 & 0.715 \\ 
0.45 & 0.543 & 0.771 & 0.560 & 0.774 & 0.705 & 0.714 & 0.706 & 0.713 \\ 
\hline 
0.46 & 0.541 & 0.770 & 0.558 & 0.773 & 0.702 & 0.712 & 0.704 & 0.712 \\ 
0.47 & 0.539 & 0.769 & 0.556 & 0.771 & 0.699 & 0.711 & 0.701 & 0.711 \\ 
0.48 & 0.537 & 0.767 & 0.554 & 0.770 & 0.697 & 0.710 & 0.698 & 0.709 \\ 
0.49 & 0.535 & 0.766 & 0.552 & 0.769 & 0.694 & 0.708 & 0.696 & 0.708 \\ 
0.50 & 0.533 & 0.765 & 0.550 & 0.768 & 0.691 & 0.707 & 0.693 & 0.707 \\ 
\hline 

\hline \hline
$d$ & $c_1$ & $q_1$ & $c_2$ & $q_2$ & $c_3$ & $q_3$ & $c_4$ & $q_4$ \\  
\hline
\end{tabular}
\begin{tabular}{||c||c|c||c|c||c|c||c|c||}
\hline
$d$ & $c_1$ & $q_1$ & $c_2$ & $q_2$ & $c_3$ & $q_3$ & $c_4$ & $q_4$ \\  
\hline \hline 
0.51 & 0.531 & 0.764 & 0.548 & 0.767 & 0.689 & 0.706 & 0.690 & 0.706 \\ 
0.52 & 0.529 & 0.763 & 0.546 & 0.766 & 0.686 & 0.705 & 0.688 & 0.704 \\ 
0.53 & 0.527 & 0.762 & 0.544 & 0.765 & 0.684 & 0.703 & 0.685 & 0.703 \\ 
0.54 & 0.526 & 0.761 & 0.542 & 0.764 & 0.682 & 0.702 & 0.683 & 0.702 \\ 
0.55 & 0.524 & 0.760 & 0.540 & 0.763 & 0.679 & 0.701 & 0.681 & 0.701 \\ 
\hline 
0.56 & 0.522 & 0.759 & 0.538 & 0.762 & 0.677 & 0.700 & 0.678 & 0.700 \\ 
0.57 & 0.520 & 0.758 & 0.537 & 0.761 & 0.675 & 0.699 & 0.676 & 0.698 \\ 
0.58 & 0.519 & 0.757 & 0.535 & 0.760 & 0.672 & 0.697 & 0.674 & 0.697 \\ 
0.59 & 0.517 & 0.756 & 0.533 & 0.759 & 0.670 & 0.696 & 0.672 & 0.696 \\ 
0.60 & 0.515 & 0.755 & 0.531 & 0.758 & 0.668 & 0.695 & 0.669 & 0.695 \\ 
\hline 
0.61 & 0.513 & 0.754 & 0.530 & 0.757 & 0.666 & 0.694 & 0.667 & 0.694 \\ 
0.62 & 0.512 & 0.754 & 0.528 & 0.757 & 0.664 & 0.693 & 0.665 & 0.693 \\ 
0.63 & 0.510 & 0.753 & 0.526 & 0.756 & 0.662 & 0.692 & 0.663 & 0.692 \\ 
0.64 & 0.509 & 0.752 & 0.525 & 0.755 & 0.660 & 0.691 & 0.661 & 0.691 \\ 
0.65 & 0.507 & 0.751 & 0.523 & 0.754 & 0.658 & 0.690 & 0.659 & 0.690 \\ 
\hline 
0.66 & 0.506 & 0.750 & 0.522 & 0.753 & 0.656 & 0.689 & 0.657 & 0.689 \\ 
0.67 & 0.504 & 0.749 & 0.520 & 0.752 & 0.654 & 0.688 & 0.655 & 0.688 \\ 
0.68 & 0.503 & 0.748 & 0.519 & 0.751 & 0.652 & 0.687 & 0.653 & 0.686 \\ 
0.69 & 0.501 & 0.747 & 0.517 & 0.750 & 0.650 & 0.686 & 0.651 & 0.685 \\ 
0.70 & 0.500 & 0.747 & 0.516 & 0.750 & 0.648 & 0.685 & 0.649 & 0.684 \\ 
\hline 
0.71 & 0.498 & 0.746 & 0.514 & 0.749 & 0.646 & 0.684 & 0.648 & 0.683 \\ 
0.72 & 0.497 & 0.745 & 0.513 & 0.748 & 0.644 & 0.683 & 0.646 & 0.682 \\ 
0.73 & 0.496 & 0.744 & 0.511 & 0.747 & 0.642 & 0.682 & 0.644 & 0.682 \\ 
0.74 & 0.494 & 0.743 & 0.510 & 0.746 & 0.641 & 0.681 & 0.642 & 0.681 \\ 
0.75 & 0.493 & 0.742 & 0.509 & 0.745 & 0.639 & 0.680 & 0.640 & 0.680 \\ 
\hline 
0.76 & 0.492 & 0.742 & 0.507 & 0.745 & 0.637 & 0.679 & 0.639 & 0.679 \\ 
0.77 & 0.490 & 0.741 & 0.506 & 0.744 & 0.635 & 0.678 & 0.637 & 0.678 \\ 
0.78 & 0.489 & 0.740 & 0.505 & 0.743 & 0.634 & 0.677 & 0.635 & 0.677 \\ 
0.79 & 0.488 & 0.739 & 0.503 & 0.742 & 0.632 & 0.676 & 0.634 & 0.676 \\ 
0.80 & 0.486 & 0.738 & 0.502 & 0.742 & 0.630 & 0.675 & 0.632 & 0.675 \\ 
\hline 
0.81 & 0.485 & 0.738 & 0.501 & 0.741 & 0.629 & 0.674 & 0.630 & 0.674 \\ 
0.82 & 0.484 & 0.737 & 0.499 & 0.740 & 0.627 & 0.673 & 0.629 & 0.673 \\ 
0.83 & 0.483 & 0.736 & 0.498 & 0.739 & 0.626 & 0.673 & 0.627 & 0.672 \\ 
0.84 & 0.481 & 0.735 & 0.497 & 0.739 & 0.624 & 0.672 & 0.626 & 0.671 \\ 
0.85 & 0.480 & 0.735 & 0.496 & 0.738 & 0.622 & 0.671 & 0.624 & 0.671 \\ 
\hline 
0.86 & 0.479 & 0.734 & 0.495 & 0.737 & 0.621 & 0.670 & 0.622 & 0.670 \\ 
0.87 & 0.478 & 0.733 & 0.493 & 0.736 & 0.619 & 0.669 & 0.621 & 0.669 \\ 
0.88 & 0.477 & 0.733 & 0.492 & 0.736 & 0.618 & 0.668 & 0.619 & 0.668 \\ 
0.89 & 0.476 & 0.732 & 0.491 & 0.735 & 0.616 & 0.667 & 0.618 & 0.667 \\ 
0.90 & 0.475 & 0.731 & 0.490 & 0.734 & 0.615 & 0.667 & 0.616 & 0.666 \\ 
\hline 
0.91 & 0.473 & 0.730 & 0.489 & 0.733 & 0.614 & 0.666 & 0.615 & 0.665 \\ 
0.92 & 0.472 & 0.730 & 0.488 & 0.733 & 0.612 & 0.665 & 0.614 & 0.665 \\ 
0.93 & 0.471 & 0.729 & 0.486 & 0.732 & 0.611 & 0.664 & 0.612 & 0.664 \\ 
0.94 & 0.470 & 0.728 & 0.485 & 0.731 & 0.609 & 0.663 & 0.611 & 0.663 \\ 
0.95 & 0.469 & 0.728 & 0.484 & 0.731 & 0.608 & 0.662 & 0.609 & 0.662 \\ 
\hline 
0.96 & 0.468 & 0.727 & 0.483 & 0.730 & 0.606 & 0.662 & 0.608 & 0.661 \\ 
0.97 & 0.467 & 0.726 & 0.482 & 0.729 & 0.605 & 0.661 & 0.607 & 0.661 \\ 
0.98 & 0.466 & 0.726 & 0.481 & 0.729 & 0.604 & 0.660 & 0.605 & 0.660 \\ 
0.99 & 0.465 & 0.725 & 0.480 & 0.728 & 0.602 & 0.659 & 0.604 & 0.659 \\ 

& & & & & & & & \\ 
\hline \hline
$d$ & $c_1$ & $q_1$ & $c_2$ & $q_2$ & $c_3$ & $q_3$ & $c_4$ & $q_4$ \\  
\hline
\end{tabular}
}
\end{center}
\label{tab:PoiCapData}
\end{table}

\begin{table}[!p]
\caption{Capacity upper bounds for the deletion channel with
deletion probability $d$,
as plotted in Figure~\ref{fig:DelPlot}. For each
$d$, the quantity $(1-d) c_i$ is the capacity upper bound when 
method $i =1,2,3,4$ is used, and $q_i$ is the choice of $q$
that maximizes the function $\CBer(1-d,q)$ under the supremum
in \eqref{eqn:del:upper:formula} for the chosen method. The methods
are: Theorem~\ref{thm:BinChUpper:trunc} (the truncated distribution)
for the distribution of $Y$ ($i=1$,
giving the strongest bounds);
Theorem~\ref{thm:BinChUpper} (inverse binomial) for the distribution of $Y$ ($i=2$);
Analytic upper bounds of Theorem~\ref{thm:ibin:Lerch} ($i=3$); 
The elementary upper bounds of Corollary~\ref{coro:invbin:negBinApprox} 
($i=4$). Note that the analytic upper bound estimates ($i=3,4$) give trivial results
for sufficiently small $d$. %
 }
\begin{center} 
\resizebox{\columnwidth}{!}{%
\begin{tabular}{||c||c|c||c|c||c|c||c|c||}
\hline
$d$ & $c_1$ & $q_1$ & $c_2$ & $q_2$ & $c_3$ & $q_3$ & $c_4$ & $q_4$ \\  
\hline \hline 
0.01 & 0.965 & 0.511 & 0.971 & 0.509 & 1.184 & 0.526 & 3.667 & 0.630 \\ 
0.02 & 0.941 & 0.520 & 0.952 & 0.516 & 1.137 & 0.525 & 2.611 & 0.614 \\ 
0.03 & 0.920 & 0.528 & 0.937 & 0.522 & 1.100 & 0.526 & 2.146 & 0.605 \\ 
0.04 & 0.902 & 0.535 & 0.923 & 0.527 & 1.069 & 0.529 & 1.872 & 0.599 \\ 
0.05 & 0.885 & 0.541 & 0.911 & 0.531 & 1.043 & 0.532 & 1.687 & 0.595 \\ 
\hline
0.06 & 0.869 & 0.547 & 0.900 & 0.535 & 1.020 & 0.534 & 1.552 & 0.593 \\ 
0.07 & 0.854 & 0.552 & 0.890 & 0.539 & 1.000 & 0.538 & 1.448 & 0.591 \\ 
0.08 & 0.840 & 0.558 & 0.881 & 0.542 & 0.982 & 0.541 & 1.365 & 0.590 \\ 
0.09 & 0.827 & 0.563 & 0.872 & 0.546 & 0.966 & 0.544 & 1.297 & 0.589 \\ 
0.10 & 0.815 & 0.568 & 0.864 & 0.549 & 0.951 & 0.547 & 1.240 & 0.588 \\ 
\hline 
0.11 & 0.803 & 0.573 & 0.856 & 0.552 & 0.937 & 0.549 & 1.191 & 0.588 \\ 
0.12 & 0.791 & 0.577 & 0.849 & 0.555 & 0.924 & 0.552 & 1.149 & 0.588 \\ 
0.13 & 0.781 & 0.581 & 0.842 & 0.557 & 0.913 & 0.555 & 1.112 & 0.588 \\ 
0.14 & 0.770 & 0.586 & 0.835 & 0.560 & 0.902 & 0.558 & 1.079 & 0.588 \\ 
0.15 & 0.760 & 0.590 & 0.829 & 0.562 & 0.891 & 0.560 & 1.050 & 0.588 \\ 
\hline
0.16 & 0.750 & 0.594 & 0.823 & 0.565 & 0.882 & 0.563 & 1.024 & 0.589 \\ 
0.17 & 0.741 & 0.598 & 0.817 & 0.567 & 0.872 & 0.565 & 1.000 & 0.589 \\ 
0.18 & 0.732 & 0.601 & 0.811 & 0.569 & 0.864 & 0.567 & 0.978 & 0.590 \\ 
0.19 & 0.723 & 0.605 & 0.806 & 0.571 & 0.856 & 0.570 & 0.958 & 0.591 \\ 
0.20 & 0.714 & 0.609 & 0.800 & 0.574 & 0.848 & 0.572 & 0.940 & 0.591 \\ 
\hline 
0.21 & 0.706 & 0.612 & 0.795 & 0.576 & 0.840 & 0.574 & 0.923 & 0.592 \\ 
0.22 & 0.698 & 0.616 & 0.790 & 0.577 & 0.833 & 0.576 & 0.908 & 0.593 \\ 
0.23 & 0.690 & 0.619 & 0.786 & 0.579 & 0.827 & 0.578 & 0.893 & 0.594 \\ 
0.24 & 0.683 & 0.622 & 0.781 & 0.581 & 0.820 & 0.580 & 0.880 & 0.594 \\ 
0.25 & 0.676 & 0.625 & 0.777 & 0.583 & 0.814 & 0.582 & 0.867 & 0.595 \\ 
\hline
0.26 & 0.669 & 0.628 & 0.772 & 0.585 & 0.808 & 0.584 & 0.855 & 0.596 \\ 
0.27 & 0.662 & 0.631 & 0.768 & 0.586 & 0.802 & 0.585 & 0.844 & 0.597 \\ 
0.28 & 0.655 & 0.634 & 0.764 & 0.588 & 0.797 & 0.587 & 0.834 & 0.598 \\ 
0.29 & 0.649 & 0.637 & 0.760 & 0.590 & 0.791 & 0.589 & 0.824 & 0.599 \\ 
0.30 & 0.643 & 0.640 & 0.756 & 0.591 & 0.786 & 0.591 & 0.814 & 0.600 \\ 
\hline 
0.31 & 0.636 & 0.643 & 0.752 & 0.593 & 0.781 & 0.592 & 0.805 & 0.601 \\ 
0.32 & 0.631 & 0.645 & 0.749 & 0.594 & 0.776 & 0.594 & 0.797 & 0.601 \\ 
0.33 & 0.625 & 0.648 & 0.745 & 0.596 & 0.772 & 0.595 & 0.789 & 0.602 \\ 
0.34 & 0.619 & 0.650 & 0.741 & 0.597 & 0.767 & 0.597 & 0.781 & 0.603 \\ 
0.35 & 0.614 & 0.653 & 0.738 & 0.599 & 0.763 & 0.598 & 0.774 & 0.604 \\ 
\hline
0.36 & 0.609 & 0.655 & 0.735 & 0.600 & 0.758 & 0.600 & 0.767 & 0.605 \\ 
0.37 & 0.604 & 0.657 & 0.731 & 0.602 & 0.754 & 0.601 & 0.760 & 0.606 \\ 
0.38 & 0.599 & 0.660 & 0.728 & 0.603 & 0.750 & 0.603 & 0.754 & 0.607 \\ 
0.39 & 0.594 & 0.662 & 0.725 & 0.604 & 0.746 & 0.604 & 0.748 & 0.608 \\ 
0.40 & 0.590 & 0.664 & 0.722 & 0.606 & 0.743 & 0.605 & 0.742 & 0.609 \\ 
\hline 
0.41 & 0.585 & 0.666 & 0.719 & 0.607 & 0.739 & 0.607 & 0.736 & 0.610 \\ 
0.42 & 0.581 & 0.668 & 0.716 & 0.608 & 0.735 & 0.608 & 0.731 & 0.610 \\ 
0.43 & 0.577 & 0.670 & 0.713 & 0.609 & 0.732 & 0.609 & 0.726 & 0.611 \\ 
0.44 & 0.573 & 0.672 & 0.710 & 0.611 & 0.728 & 0.610 & 0.720 & 0.612 \\ 
0.45 & 0.569 & 0.673 & 0.707 & 0.612 & 0.725 & 0.612 & 0.716 & 0.613 \\ 
\hline
0.46 & 0.565 & 0.675 & 0.704 & 0.613 & 0.722 & 0.613 & 0.711 & 0.614 \\ 
0.47 & 0.562 & 0.677 & 0.702 & 0.614 & 0.719 & 0.614 & 0.706 & 0.615 \\ 
0.48 & 0.558 & 0.678 & 0.699 & 0.615 & 0.715 & 0.615 & 0.702 & 0.616 \\ 
0.49 & 0.555 & 0.680 & 0.696 & 0.616 & 0.712 & 0.616 & 0.698 & 0.617 \\ 
0.50 & 0.552 & 0.681 & 0.694 & 0.618 & 0.709 & 0.617 & 0.694 & 0.618 \\ 
\hline 

\hline \hline
$d$ & $c_1$ & $q_1$ & $c_2$ & $q_2$ & $c_3$ & $q_3$ & $c_4$ & $q_4$ \\  
\hline
\end{tabular}
\begin{tabular}{||c||c|c||c|c||c|c||c|c||}
\hline
$d$ & $c_1$ & $q_1$ & $c_2$ & $q_2$ & $c_3$ & $q_3$ & $c_4$ & $q_4$ \\  
\hline \hline 
0.51 & 0.549 & 0.683 & 0.691 & 0.619 & 0.707 & 0.619 & 0.692 & 0.619 \\ 
0.52 & 0.546 & 0.684 & 0.689 & 0.620 & 0.704 & 0.620 & 0.691 & 0.620 \\ 
0.53 & 0.544 & 0.685 & 0.686 & 0.621 & 0.701 & 0.621 & 0.690 & 0.622 \\ 
0.54 & 0.541 & 0.687 & 0.684 & 0.622 & 0.698 & 0.622 & 0.688 & 0.623 \\ 
0.55 & 0.538 & 0.688 & 0.681 & 0.623 & 0.695 & 0.623 & 0.687 & 0.624 \\ 
\hline
0.56 & 0.536 & 0.689 & 0.679 & 0.624 & 0.693 & 0.624 & 0.685 & 0.626 \\ 
0.57 & 0.533 & 0.690 & 0.677 & 0.625 & 0.690 & 0.625 & 0.684 & 0.627 \\ 
0.58 & 0.531 & 0.691 & 0.675 & 0.626 & 0.688 & 0.626 & 0.682 & 0.628 \\ 
0.59 & 0.529 & 0.692 & 0.672 & 0.627 & 0.685 & 0.627 & 0.681 & 0.629 \\ 
0.60 & 0.526 & 0.694 & 0.670 & 0.628 & 0.683 & 0.628 & 0.679 & 0.630 \\ 
\hline 
0.61 & 0.524 & 0.695 & 0.668 & 0.629 & 0.680 & 0.628 & 0.678 & 0.632 \\ 
0.62 & 0.522 & 0.696 & 0.666 & 0.630 & 0.678 & 0.629 & 0.676 & 0.633 \\ 
0.63 & 0.520 & 0.697 & 0.664 & 0.631 & 0.676 & 0.630 & 0.675 & 0.634 \\ 
0.64 & 0.518 & 0.698 & 0.662 & 0.631 & 0.674 & 0.631 & 0.673 & 0.635 \\ 
0.65 & 0.516 & 0.699 & 0.659 & 0.632 & 0.671 & 0.632 & 0.672 & 0.636 \\ 
\hline
0.66 & 0.514 & 0.700 & 0.657 & 0.633 & 0.669 & 0.633 & 0.670 & 0.637 \\ 
0.67 & 0.512 & 0.701 & 0.655 & 0.634 & 0.667 & 0.634 & 0.669 & 0.638 \\ 
0.68 & 0.510 & 0.701 & 0.653 & 0.635 & 0.665 & 0.635 & 0.667 & 0.639 \\ 
0.69 & 0.508 & 0.702 & 0.652 & 0.636 & 0.663 & 0.635 & 0.665 & 0.640 \\ 
0.70 & 0.506 & 0.703 & 0.650 & 0.637 & 0.661 & 0.636 & 0.664 & 0.641 \\ 
\hline 
0.71 & 0.505 & 0.704 & 0.648 & 0.638 & 0.659 & 0.637 & 0.662 & 0.642 \\ 
0.72 & 0.503 & 0.705 & 0.646 & 0.638 & 0.657 & 0.638 & 0.661 & 0.643 \\ 
0.73 & 0.501 & 0.706 & 0.644 & 0.639 & 0.655 & 0.639 & 0.659 & 0.644 \\ 
0.74 & 0.499 & 0.707 & 0.642 & 0.640 & 0.653 & 0.639 & 0.658 & 0.645 \\ 
0.75 & 0.498 & 0.707 & 0.640 & 0.641 & 0.651 & 0.640 & 0.656 & 0.646 \\ 
\hline
0.76 & 0.496 & 0.708 & 0.639 & 0.642 & 0.649 & 0.641 & 0.655 & 0.646 \\ 
0.77 & 0.494 & 0.709 & 0.637 & 0.642 & 0.647 & 0.642 & 0.653 & 0.647 \\ 
0.78 & 0.493 & 0.710 & 0.635 & 0.643 & 0.645 & 0.643 & 0.652 & 0.648 \\ 
0.79 & 0.491 & 0.711 & 0.633 & 0.644 & 0.644 & 0.643 & 0.650 & 0.649 \\ 
0.80 & 0.490 & 0.711 & 0.632 & 0.645 & 0.642 & 0.644 & 0.649 & 0.650 \\ 
\hline 
0.81 & 0.488 & 0.712 & 0.630 & 0.645 & 0.640 & 0.645 & 0.647 & 0.651 \\ 
0.82 & 0.487 & 0.713 & 0.628 & 0.646 & 0.638 & 0.645 & 0.646 & 0.651 \\ 
0.83 & 0.485 & 0.714 & 0.627 & 0.647 & 0.637 & 0.646 & 0.644 & 0.652 \\ 
0.84 & 0.484 & 0.714 & 0.625 & 0.648 & 0.635 & 0.647 & 0.643 & 0.653 \\ 
0.85 & 0.483 & 0.715 & 0.623 & 0.648 & 0.633 & 0.647 & 0.641 & 0.654 \\ 
\hline
0.86 & 0.481 & 0.716 & 0.622 & 0.649 & 0.632 & 0.648 & 0.640 & 0.654 \\ 
0.87 & 0.480 & 0.716 & 0.620 & 0.650 & 0.630 & 0.649 & 0.638 & 0.655 \\ 
0.88 & 0.479 & 0.717 & 0.619 & 0.651 & 0.629 & 0.649 & 0.637 & 0.656 \\ 
0.89 & 0.477 & 0.718 & 0.617 & 0.651 & 0.627 & 0.650 & 0.635 & 0.657 \\ 
0.90 & 0.476 & 0.718 & 0.616 & 0.652 & 0.626 & 0.651 & 0.634 & 0.657 \\ 
\hline 
0.91 & 0.475 & 0.719 & 0.614 & 0.653 & 0.624 & 0.651 & 0.632 & 0.658 \\ 
0.92 & 0.473 & 0.720 & 0.613 & 0.653 & 0.622 & 0.652 & 0.631 & 0.659 \\ 
0.93 & 0.472 & 0.720 & 0.611 & 0.654 & 0.621 & 0.653 & 0.629 & 0.659 \\ 
0.94 & 0.471 & 0.721 & 0.610 & 0.655 & 0.620 & 0.653 & 0.628 & 0.660 \\ 
0.95 & 0.470 & 0.721 & 0.608 & 0.655 & 0.618 & 0.654 & 0.626 & 0.661 \\ 
\hline
0.96 & 0.469 & 0.722 & 0.607 & 0.656 & 0.617 & 0.655 & 0.625 & 0.661 \\ 
0.97 & 0.467 & 0.723 & 0.605 & 0.657 & 0.615 & 0.655 & 0.623 & 0.662 \\ 
0.98 & 0.466 & 0.723 & 0.604 & 0.657 & 0.614 & 0.656 & 0.622 & 0.663 \\ 
0.99 & 0.465 & 0.724 & 0.602 & 0.658 & 0.612 & 0.656 & 0.621 & 0.663 \\ 
& & & & & & & & \\ 
\hline \hline
$d$ & $c_1$ & $q_1$ & $c_2$ & $q_2$ & $c_3$ & $q_3$ & $c_4$ & $q_4$ \\  
\hline
\end{tabular}
}
\end{center}
\label{tab:DelCapData}
\end{table}

\end{document}